%% file: main.tex
\documentclass[11pt,letterpaper,dvipsnames]{article}
\usepackage{sbm}

\begin{document}

% Declarations for Front Matter
\title{Local Statistics, Semidefinite Programming, and Community Detection}
\author{Jess Banks \thanks{University of California, Berkeley} \and Sidhanth Mohanty \thanks{University of California, Berkeley} \and Prasad Raghavendra\thanks{University of California, Berkeley}}
\maketitle
\thispagestyle{empty}

\begin{abstract}
    We propose a new hierarchy of semidefinite programming relaxations for inference problems.
    As test cases, we consider the problem of community detection in block models.  The vertices are partitioned into $k$ communities, and a graph is sampled conditional on a prescribed number of inter- and intra-community edges. The problem of \emph{detection}, where we are to decide with high probability whether a graph was drawn from this model or the uniform distribution on regular graphs, is conjectured to undergo a computational phase transition at a point called the Kesten-Stigum (KS) threshold.
    
    In this work, we consider two models of random graphs namely the well-studied (irregular) stochastic block model and a distribution over random regular graphs we'll call the \fullmodel.
    For both these models, we show that sufficiently high constant levels of our hierarchy can perform detection arbitrarily close to the KS threshold and that our algorithm is robust to up to a linear number of adversarial edge perturbations.
    Furthermore, in the case of \fullmodel, we show that below the Kesten-Stigum threshold no constant level can do so.

    In the case of the (irregular) Stochastic Block Model, it is known that efficient algorithms exist all the way down to this threshold, although none are robust to a linear number of adversarial perturbations of the graph when the average degree is small. More importantly, there is little complexity-theoretic evidence that detection is hard below the threshold. 
    In the \model with more than two groups, it has not to our knowledge been proven that any algorithm succeeds down to the KS threshold, let alone that one can do so robustly, and there is a similar dearth of evidence for hardness below this point.
    
    Our SDP hierarchy is highly general and applicable to a wide range of hypothesis testing problems.
\end{abstract}

\thispagestyle{empty}
\setcounter{page}{0}
\newpage

\setcounter{page}{0}
\thispagestyle{empty}
\tableofcontents
\thispagestyle{empty}
\setcounter{page}{0}
\newpage

%\clearpage

% \input{Content/glossary.tex}
\input{Content/intro.tex}
\input{Content/technical.tex}
\input{Content/coloring.tex}
\input{Content/drbm.tex}
\input{Content/sbm.tex}
\input{Content/sbm-lower-bound}

\bibliographystyle{amsalpha}
\bibliography{bib/mr,bib/dblp,bib/bibliography,bib/plantedcoloring,bib/listDecoding}

\appendix
\input{Content/recovery.tex}
% \input{Content/planted.tex}
\input{Content/loc-stat-calc.tex}
\input{Content/singleton}
\input{Content/robustness}
%\printbibliography

\end{document}

%% file: Content/intro.tex
\section{Introduction}
\label{sec:intro}

%\emph{Community detection in graphs} is a canonical and widely applicable problem in computer science and machine learning. The setup is both simple and flexible: we are shown a graph and asked for a coarse-grained description in the form of a partition of the vertices into `communities' with atypically many internal edges. The literature contains innumerable algorithms and approaches for this task, but perhaps the most fruitful has been a Bayesian perspective wherein we treat the graph as the output of some generative model, whose unknown parameters we attempt to estimate. In other words, we assume that there are some true and hidden community labels, and that the graph has been drawn probabilistically in a way that respects this `planted' structure.

Community detection is a canonical example of a high-dimensional inference problem, one that is a test-bed to develop algorithmic and lower bound techniques.
Much of the existing literature on community detection concerns the \emph{stochastic block model (SBM)}. For now let us discuss the \textit{symmetric} setting where we first partition $n$ vertices in to $k$ groups, and include each edge independently and with probability $\pin$ or $\pout$ depending on whether or not the labels of its endpoints coincide. Research in this area spans several decades, and it will not be fruitful to attempt a thorough review of the literature here; we refer the reader to \cite{abbe2017community} for a survey. Most salient to us, however, is a rich theory of computational threshold phenomena which has emerged out of the past several years of collaboration between computer scientists, statisticians, and statistical physicists.

The key computational tasks associated with the SBM are \emph{recovery} and \emph{detection}: we attempt either to reconstruct the planted communities from the graph, or to decide whether a graph was drawn from the planted model or the \ER model with the same average degree. A set of fascinating conjectures were posed in Decelle \etal~ \cite{decelle2011inference}, regarding these tasks in the case of `sparse' models where $\pin,\pout = O(1/n)$ and the average degree is $O(1)$ as the number of vertices diverges. 

It is typical to parametrize the symmetric SBM in terms of $k$, the average degree
$$
	d = \frac{n\pin + (k-1)n\pout}{k},
$$
and a `signal-to-noise ratio'
$$
	\lambda \triangleq \frac{n\pin - n\pout}{kd}.
$$
In this setup, it is believed that as we hold $k$ and $\lambda$ constant, then there is an \emph{information-theoretic threshold} $d_{IT} \approx \tfrac{\log k}{k\lambda^2}$, in the sense that when $d<d_{IT}$ both detection and recovery are impossible for any algorithm. Moreover, Decelle et al. conjecture that efficient algorithms for both tasks exist only when the degree is larger than a point known as the \emph{Kesten-Stigum threshold} $d_{KS} = \lambda^{-2}$. Much of this picture is now rigorous \cite{mossel2018proof, massoulie2014community,bordenave2015non,abbe2016exact}. Still, fundamental questions remain unanswered. What evidence can we furnish that detection and recovery are indeed intractible in the so-called `hard regime' $d_{IT} < d < d_{KS}$? How robust are these thresholds to adversarial noise or small deviations from the model?

Zooming out, this discrepancy between information-theoretic and computational thresholds is conjectured to be quite universal among planted problems, where we are to reconstruct or detect a structured, high-dimensional signal observed through a noisy channel\snote{fill in citations}. The purpose behind our work is to begin developing a framework capable of providing evidence for average case computational intractability in such settings. To illustrate this broader motivation, consider a different average-case problem also conjectured to be computationally intractable: refutation of random $3$-SAT. A random instance of $3$-SAT with $n$ literals and, say $m = 1000 n$ clauses is unsatisfiable with high probability. However, it is widely conjectured that the problem of \emph{certifying} that a given random $3$-SAT instance is unsatisfiable is computationally intractable (all the way up to $n^{3/2}$ clauses) \cite{feigerandom3sat}. While proving intractability remains out of reach, the complexity theoretic literature now contains ample evidence in support of this conjecture. Most prominently, exponential lower bounds are known for the problem in restricted computational models such as linear and semidefinite programs \cite{grigoriev2001linear} and resolution based proofs \cite{ben2001short}. Within the context of combinatorial optimization, the Sum-of-Squares (SoS) SDPs yield a hierarchy of successively more powerful and complex algorithms which capture and unify many other known approaches. A lower bound against the SoS SDP hierarchy such as \cite{grigoriev2001linear} provides strong evidence that this refutation problem is computationally intractable. This paper is a step towards developing a similar framework to reason about the computational complexity of detection and recovery in stochastic block models specifically, and planted problems generally.

A second motivation is the issue of robustness of computational thresholds under adversarial perturbations of the graph. Spectral algorithms based on non-backtracking walk matrix \cite{bordenave2015non} achieve weak-detection as soon as $d> d_{KS}$, but are not robust in this sense. Conversely, robust algorithms for recovery are known, but only when the edge-densities are significantly higher than Kesten-Stigum \cite{guedon2016community,makarychev2016learning,CharikarSV17,SteinhardtVC16}. The positive result that gets closest to robustly achieving the conjectured computational phase transition at $d_{KS}$ is the work of Montanari and Sen \cite{montanari2015semidefinite} who observe that their SDP-based algorithm for testing whether the input graph comes from the \ER distribution or a Stochastic Block Model with $k = 2$ communities also works in presence of $o(|E|)$ edge outlier errors.  On the negative side, Moitra et al. \cite{moitra2012singly} consider the problem of weak recovery in a SBM with two communities and $\pin > \pout$ in the presence of {\it monotone errors} that add edges within communities and delete edges between them. Their main result is a statistical lower bound indicating the phase transition for weak recovery changes in the presence of monotone errors. This still leaves open the question of whether there exist algorithms that weakly recover right at the threshold and are robust to $o(|E|)$ perturbations in the graph.

\section{Main Results} % (fold)
\label{sec:main_results}

We define a new hierarchy of semidefinite programming relaxations for inference problems that we refer to as the \emph{Local Statistics} hierarchy, denoted $\LS(D_G,D_x)$ and indexed by parameters $D_G,D_x \in \N$. This family of SDPs is inspired by the technique of pseudocalibration in proving lower bounds for sum-of-squares (SoS) relaxations, as well as subsequent work of Hopkins and Steurer \cite{hopkins2017efficient} extending it to an SoS SDP based approach to inference problems. The $\LS$ hierarchy can be defined for a broad range of inference problems involving a joint distribution $\mu$ on an observation and hidden parameter.

As test cases, we apply our SDP relaxations to community detection in two families of random graphs with planted community structure: the sparse Stochastic Block Model (SBM) discussed above, and a degree-regular analogue that we term the \emph{\fullmodel~(DRBM)}. Our results will concern the problem of \textit{detection}, defined formally as follows.

\begin{definition}[Detection and Robustness]    \label{def:robustness}
    Let $\Planted_n$ and $\Null_n$ denote two sequences of distributions on graphs. We say that an algorithm $\mathsf{A} : \text{Graphs} \to \{\textsc{p},\textsc{n}\}$ \emph{solves the detection problem}, or \textit{can distinguish $\Planted_n$ and $\Null_n$}. if 
    \begin{align*}
        \Planted_n\big[ \mathsf{A}(G) = \textsc{p} \big] = 1 - o_n(1) \qquad \text{ and } \qquad \Null_n\big[\mathsf{A}(G) = \textsc{n} \big] = 1 - o_n(1).
    \end{align*}
    Fix $\epsilon > 0$, and write $G \approx_\epsilon \wt{G}$ to mean that two graphs on the same vertex set $V$ differ at at most $ \epsilon |V|$ edges. If $\mathsf{A}$ solves the detection problem, we say that it does so \textit{$\epsilon$-robustly} if
    $$
        \Planted_n\big[\mathsf{A}(G) = \mathsf{A}(\wt{G}) ,\,\forall G \approx_\epsilon \wt{G}\big] = 1 - o_n(1) \qquad \text{ and } \qquad \Null_n\big[\mathsf{A}(G) = \mathsf{A}(\wt{G}) ,\,\forall G \approx_\epsilon \wt{G}\big] = 1 - o_n(1).
    $$
\end{definition}

\paragraph{The Stochastic Block Model}
Adapting notation from \cite{bordenave2015non}, we will parameterize the SBM by average degree $d$, number of communities $k$, group size distribution $\pi \in \bbR^k$, and symmetric, nonnegative edge probability matrix $M \in \bbR^{k\times k}$. To sample a graph $\bG = (V(\bG),E(\bG))$, first choose the label $\sigma(u)$ of each vertex $u \in V(\bG)$ independently according to $\pi$, and then include each potential edge $(u,v)$ with probability $M_{\sigma(u),\sigma(v)} \cdot d/n$. We adopt the natural requirement that the average degree of a vertex conditional on any group label is $d$, which is equivalent to the normalization condition $M\pi = e$, where the latter is the all-ones vector in $\bbR^k$. We will call the model \textit{symmetric} if
\begin{align}
    M_{i,j} = \begin{cases} 1 + (k-1)\lambda & i=j \\
    1 - \lambda & i\neq j.
    \end{cases} \label{eq:sym-bm}
\end{align}
One can check that this recovers the setup in the previous section.

The general SBM, like this symmetric subcase, is conjectured to undergo a series of phase transitions as $(k,M,\pi)$ are held fixed and the average degree is varied. These include an information-theoretic threshold and, most salient to this paper, a computational `Kesten-Stigum' transition \cite{decelle2011asymptotic}.  To describe the latter, it is necessary to introduce one further piece of notation, which will be of repeated use to us in the course of the paper. Write $T \triangleq M \Diag \pi$, noting that $T$ is the transition matrix for a reversible Markov chain with stationary distribution $\pi$. For any vertex in group $i$, the label of a uniformly random neighbor is roughly distributed according to the $i$th row of $T$, and, more generally, the vertex labels encountered by a random non-backtracking random walk are approximately governed by the Markov process that $T$ defines. As this process is stationary, the spectrum of $T$ is real, and we will write its eigenvalues as $1 = \lambda_1 \ge |\lambda_2| \ge \cdots \ge |\lambda_k|$. The second eigenvalue $\lambda_2$ is a generalization of the signal-to-noise ratio $\lambda$ from equation \pref{eq:sym-bm}; in fact one can verify that in the symmetric SBM, $\lambda_2 = \cdots = \lambda_k = \lambda$. The Kesten-Stigum threshold is thus defined as $d_{\ks} \triangleq \lambda_2^{-2}$.

Our main theorem regarding the SBM is that, when $d > d_{\ks}$, the $\LS(2,D)$ SDP can robustly solve the detection problem for some $D=O(1)$ (albeit tending to infinity as $d \to d_{\ks})$. 

\begin{theorem} \label{thm:main-sbm}
    Let $\Null_n = \calG(n,d/n)$, and $\Planted_n$ denote the $n$-vertex SBM with parameters $(d,k,M,\pi)$. If $d > d_{\ks}$, then there exist $\delta > 0$, $D = O(1)$, and $\rho > 0$ (all dependent on $d$) for which the $\LS(2,D)$ SDP with error tolerance $\delta$ can $\rho$-robustly solve the detection problem. 
\end{theorem}

We additionally show that a simplified version of the $\LS(2,D)$ SDP (\pref{def:lvl-m-ps-sdp}) which is powerful enough to solve the detection problem above the KS threshold, fails to do so below it at every constant level. This is the content of the forthcoming \pref{thm:local-path-stats}.

\paragraph{The Degree Regular Block Model}

We will parametrize the DRBM identically to the SBM, by a quadruple $(d,k,M,\pi)$; this time we of course require that $d$ is an integer. To sample a graph $\bG = (V(\bG),E(\bG))$, first choose a uniformly random ``$\pi$-balanced'' partition $V(\bG) = \bigsqcup_{i\in[k]} V_i(\bG)$, by which we mean that $|V_i(\bG)| = \pi(i)n$ for every $i$. Then, choose a uniformly random $d$-regular graph, conditioned on there being exactly $\pi(i)\pi(j)M(i,j) \cdot dn$ edges between each pair of distinct groups $i\neq j$, and $\pi(i)^2M(i,j) \cdot dn/2$ edges internal to each group $i$. For simplicity, we will assume that the parameters are such that these group sizes and edge counts are integers. As with the SBM, we will call the model \textit{symmetric} if the entries of $M$ are constant on the diagonal and off-diagonal respectively. As a warm-up for the main technical arguments of the paper, we will study in \pref{sec:simple-DRBM} a simplified version of the Local Statistics SDP that can solve the detection problem on the symmetric DRBM.

\begin{remark}
        The DRBM as we have defined it differs from the Regular Stochastic Block Model of \cite{brito2016recovery}, in which each vertex has a prescribed number of neighbors in every community. Although superficially similar, the behavior of this `equitable' model (as it is known in the physics literature \cite{New14}) is quite different from ours. For instance, \cite{brito2016recovery} show that whenever detection is possible in the two community case, one can \textit{exactly} recover the planted labels. This is not true in our setting.
\end{remark}

It is widely believed that the threshold behavior of the general DRBM is analogous to that of the SBM, including an information-theoretic threshold, and Kesten-Stigum threshold at $d_{\ks} \triangleq \lambda_2^{-2} + 1$.  However, most formal treatment in the literature has been limited to random $d$-regular graphs conditional on having a planted $k$-coloring, a case not fully captured by our model. Characterization of the information-theoretic threshold, even in simple cases, remains largely folklore.

Our main result on the DRBM is analogous to Theorem \ref{thm:main-sbm} on the SBM.

\begin{theorem} \label{thm:main-drbm}
	Let $\Null_n$ denote the uniform distribution on $d$-regular graphs with $n$-vertices, and $\Planted_n$ the DRBM with parameters $(d,k,M,\pi)$. If $d > d_{\ks}$, then there exists a constant $m \in \bbN$, $\delta > 0$, and and $\rho > 0$ (all dependent on $d$) so that $\LS(2,m)$ with error tolerance $\delta$ can $\rho$-robustly solve the detection problem. Conversely, if $d < d_{\ks}$, then every constant level, no matter the error tolerance, fails to do so.
\end{theorem}

Along the way we will inadvertently prove that standard spectral detection using the adjacency matrix succeeds above $d_{KS}$, but cannot have the same robustness guarantee. It is a now-classic result of Friedman that, with probability $1 - o_n(1)$, the spectrum of a uniformly random $d$-regular graph is within $o_n(1)$ of $(-2\sqrt{d-1},2\sqrt{d-1})\cup\{d\}$. Conversely, we show:

\begin{corollary} \label{cor:spectral-distinguishing}
	Let $\bG$ be drawn from the \model with parameters $(d,k,M,\pi)$ satisfying $d > d_{\KS} + \epsilon$. There exists some $\eta = \eta(\epsilon)$ such that, for each eigenvalue $\lambda$ of $M$ satisfying $|\lambda| > 1/\sqrt{d-1} + \epsilon$, the adjacency matrix $A_{\bG}$  is guaranteed one eigenvalue $\mu$ satisfying $|\mu| > 2\sqrt{d-1} + \eta$.
\end{corollary}

\paragraph{Future Work} Regrettably, we do not solve the problem of recovery above Kesten-Stigum in either model. However, we will in \pref{app:recovery} reduce recovery in the DRBM to the following conjecture regarding the spectrum of $A_{\bG}$ for $\bG$ drawn from the planted model.

\begin{conjecture} \label{conj:planted-spec}
	Let $\Planted_{(d,k,M,\pi)}$ be any \model with $|\lambda_1|,...,|\lambda_\ell| > (d-1)^{-1/2}$. Then, for any $\eta$, with high probability, $A_{\bG}$ has only $\ell$ eigenvalues with modulus larger than $2\sqrt{d-1} + \eta$.
\end{conjecture}

% We will discuss in Appendix \ref{sec:recovery} that, conditional on this conjecture (or even a weaker version in which we are guaranteed only constantly many eigenvalues outside the bulk), (i) the span of the corresponding eigenvectors is correlated to the community structure, and (ii) the Local Statistics hierarchy can robustly produce vectors with macroscopic correlation to this span. From weak convergence of the empirical spectral distribution of $A_{\bG}$ to the Kesten-McKay law, we know that there must be $o(n)$ eigenvalues with modulus larger than $2\sqrt{d-1}$, it will take substantial technical work to push this down to $O(1)$. We believe the most feasible approach is a careful mirror of the techniques in \cite{bordenave2015non}, but the execution of this is beyond the scope of this paper. These issues and a related conjecture are discussed in \ref{mossel2015reconstruction} in the context of the \model with two groups.

% \begin{remark}
%     \todo{We should comment on why the run-time needs to blow up close to the KS threshold}.
% \end{remark}

\paragraph{Related Work.}

Semidefinite programming approaches have been most studied in the dense, irregular case, where exact recovery is possible (for instance \cite{abbe2016exact,abbe2015community}), and it has been shown that an SDP relaxation can achieve the information-theoretically optimal threshold \cite{hajek2016achieving}. However, in the sparse regime we consider, the power of SDP relaxations for weak recovery remains unclear. Guedon and Vershynin \cite{guedon2016community} show upper bounds on the estimation error of a standard SDP relaxation in the sparse, two-community case of the SBM, but only when the degree is roughly $10^4$ times the information theoretic threshold. More recently, in a tour-de-force,  Montanari and Sen \cite{montanari2015semidefinite} showed that for two communities, the SDP of Guedon and Vershynin achieves the information theoretically optimal threshold for large but constant degree, in the sense that the performance approaches the threshold if we send the number of vertices, and then the degree, to infinity. Semi-random graph models have been intensively studied in \cite{blum1995coloring, feige2000finding, feige2001heuristics,coja2004coloring,krivelevich2006semirandom,coja2007solving, makarychev2012approximation, chen2014clustering,guedon2016community} and we refer the reader to \cite{makarychev2016learning} for a more detailed survey. In the logarithmic-degree regime, robust algorithms for community detection are developed in \cite{cai2015robust, kumar2010clustering, awasthi2012improved}. Far less is known in the case of regular graphs.

% section main_results (end)

%% file: Content/technical.tex
\section{Technical Overview} 
\label{sec:technical}

\paragraph{Notation.}  We will use bold face font for random objects sampled from these distributions. Because we care only about the case when the number of vertices is very large, we will use \emph{with high probability (w.h.p)} to describe any sequence of events with probability $1 - o_n(1)$ in $\Null$ or $\Planted$ as $n\to \infty$. We will write $[n] = \{1,...,n\}$, and in general use the letters $u,v,w$ to refer to elements of $[n]$ and $i,j$ for elements of $[k]$. The identity matrix will be denoted by $\1$, and we will write $X^T$ for the transpose of a matrix $X$, $\langle X,Y \rangle = \Tr X^TY$ for the standard matrix inner product, and $\|X\|_F$ for the associated Frobenius norm. Positive semidefiniteness  will be indicated with the symbol $\succeq$. The standard basis vectors will be denoted $e_1,e_2,...$, the all-ones vector written as $e$, and the all-ones matrix as $\bbJ = ee^T$. Finally, let $\diag : \bbR^{n\times n} \to \bbR$ be the function extracting the diagonal of a matrix, and $\Diag: \bbR^n \to \bbR^{n\times n}$ be the one which populates the nonzero elements of a diagonal matrix with the vector it is given as input.

\subsection{Optimization vs Inference} 
While it was suspected that a semidefinite programming relaxation could be used towards community detection in sparse stochastic block models, many earlier attempts at it  \cite{guedon2016community,montanari2015semidefinite} failed to detect communities right up to the KS threshold at a fixed degree.
These works studied the Goemans-Williamson SDP relaxation for MaxCut applied to the problem of detecting two communities ($k = 2$).  The idea being that if we consider a two community SBM with $p_{out} > p_{in}$, then the partition induced by the communities should have an unusually large number ($\frac{dn}{2} \cdot \frac{ p_{out}}{ p_{out}+p_{in}}$) of crossing edges.  Hence an SDP relaxation of MaxCut could be harnessed towards detecting and possibly recovering the communities.
Indeed, in this special case, the maximum bisection in the graph is a Maximum Likelihood Estimate (MLE) for the communities $x$ given the graph $G$, i.e., $x = \argmax_{x} p(x|G)$.

This approach of casting inference as optimization has its limitations. In particular, as one approaches the KS threshold, the number of crossing edges between the two communities, namely $\frac{dn}{2} \cdot \frac{ p_{out}}{ p_{out}+p_{in}}$, is lower than the value of MaxCut in a random Erdos-Renyi graph!  In other words, if we run an exponential-time algorithm that finds the maximum cut via a brute-force enumeration, then it will find a better MaxCut in a random Erdos-Renyi graph than the true communities in the planted model.  It is therefore unclear whether an SDP relaxation of MaxCut can solve the problem.

In hindsight, the number of crossing edges is but one statistic associated with the partition and there is no canonical reason why optimizing this statistic would be the optimal way to distinguish the two models.  For example, in the same setting one could minimize the number of paths of length two that go between the two sides of the partition,  or maximize the number of paths of length three that cross the partition and so on.  At a more basic level, if we are interested in estimating the moments of the distribution $x | G$, it is not clear that we should cast this problem as optimization.

The local statistics SDP hierarchy that we propose is a "feasibility SDP" that looks for candidate low-degree moments for the distribution $x|G$.  The constraints of the SDP ensure that the value of local statistics such as number of crossing edges is  roughly the same as we would expect in a graph drawn from the communities. % This idea of matching local-statistics underlies pseudo-calibration -- a technique to construct candidate SDP lower bounds \cite{barak2019nearly}.  The work of Hopkins and Steurer \cite{hopkins2017efficient} construct a candidate semidefinite program, but it is not a relaxatio

\subsection{Detection, Refutation, and Sum-of-Squares} % (fold)
\label{sub:detection_refutation_and_sum_of_squares}

We will begin the discussion of the Local Statistics algorithm by briefly recalling Sum-of-Squares programming. Say we have a constraint satisfaction problem presented as a system of polynomial equations in variables $x = (x_1,...,x_n)$ that we are to simultaneously satisfy. In other words, we are given a set
$$
	\calS = \left\{ x \in \bbR^n : f_1(x),...,f_m(x) = 0\right\}
$$
and we need to decide if it is non-empty. Whenever the problem is satisfiable, any probability distribution supported on $\calS$ gives rise to an operator $\expected: \bbR[x] \to \bbR$ mapping a polynomial $x$ to its expectation. Trivially, $\expected$ has the properties:
\begin{align}
	\textit{Normalized} & &\expected 1 &= 1 \label{eq:exp-norm} \\
	\textit{Satisfies of $\calS$}& & \expected f_i(x)\cdot p(x) &= 0 & &\forall i\in[m], \forall p \in\bbR[x] \label{eq:exp-sat}\\
	\textit{Positive} & & \expected p(x)^2 &\ge 0 & &\forall p \in \bbR[x] \label{eq:exp-pos} 
\end{align}
We will extend these definitions to any operator mapping some subset of $\bbR[x] \to \bbR$.

Refuting the constraint satisfaction problem, e.g. proving that $\calS = \emptyset$, is equivalent to showing that no operator obeying \eqref{eq:exp-norm}-\eqref{eq:exp-pos} can exist. The key insight of SoS is that often one can do this by focusing only on polynomials of some bounded degree. Writing $\bbR[x]_{\le D}$ for the polynomials of degree at most $D$, we call an operator $\pseudo : \bbR[x]_{\le D} \to \bbR$ a \emph{degree-$D$ pseudoexpectation} if it is normalized, positive, and satisfies $\calS$ for every polynomial in its domain. It is well-known that one can search for a degree $D$ pseudoexpectation with a semidefinite program of size $O(n^D)$, and if this smaller, relaxed problem is infeasible, we've shown that $\calS$ is empty. This is the \emph{degree-$D$ Sum-of-Squares relaxation} of our CSP.

% subsection detection_refutation_and_sum_of_squares (end)

\subsection{The Local Statistics Hierarchy}

%The virtue of Sum-of-Squares programming is that the degree $D$ necessary to refute a CSP can be used as a proxy for its computational hardness. In this paper, we aim to provide an analogous family of algorithms which to provide evidence for the computational hardness of hypothesis testing problems.

Let $\Planted_n$ denote a sequence of distributions on graphs with a planted community structure, and $\Null_n$ a corresponding `null' distribution with no such prescribed structure. For us, $\Planted_n$ will always denote the DRBM or SBM, and $\Null_n$ the \ER model with average degree $d$, or the uniform distribution on $d$-regular graphs. Our goal is to devise an algorithm that can discern, with high probability, which of these two distributions a graph was drawn from. In this setup, the details of the null and prior distribution are known to us; the main idea of this work is that it is only natural to grant an SDP hypothesis testing algorithm access to this information as well. Our strategy will be do devise an SDP that is satisfiable with high probability when a graph is drawn from $\Planted_n$, and unsatisfiable with high probability when it is drawn from $\Null_n$.

The Local Statistics SDP will be assembled from components of the Sum-of-Squares algorithm, and as such we will need to carefully articulate the null and planted distribution, and their statistical properties, in the language of polynomials. Let us write $x = \{x_{u,i}\}$ for a collection of variables indexed by vertices $u \in [n]$ and group labels $i \in [k]$, and $G = \{G_{u,v}\}$ for a collection indexed by two-element subsets $\{u,v\} \subset [n]$. We will regard a random graph from the null model as a collection of random variables $\bG = \{\bG_{u,v}\}$ indexed in the same way, where $\bG_{u,v}$ is the Boolean indicator for the edge $(u,v)$. Similarly, the planted model is a joint distribution over pairs $(\bx,\bG)$, where $\bG$ is a graph, and $\bx_{i,u}$ is the indicator that vertex $u$ has label $i$. Thus for each polynomial $p \in \bbR[G,x]$, we can compute the \textit{statistic} $\expected p(\bG,\bx)$. We will see below that one can easily construct such a polynomial that counts, for instance, the number of triangles in a graph, or the number of edges between vertices in the same group.

The random variables $\bG$ and $\bx$ take values in the zero locus of the following set of polynomials in $\bbR[G,x]$:
\begin{align}
    G_{u,v}^2 &= G_{u,v} & & \forall u,v\in[n] \label{eq:G-bool}\\
    x_{u,i}^2 &= x_{u,i} & & \forall u \in [n],\,i \in [k] \label{eq:x-bool} \\
    x_{u,1} + \cdots + x_{u,k} &=1 & & \forall u \in [n] \label{eq:x-single}.
\end{align}
For brevity, we will throughout the paper denote by $\calB_k$ the set of polynomials constraints in the $x$ variables appearing in \pref{eq:x-bool} and \pref{eq:x-single}. Moreover, in our case both the null and planted models have a natural symmetry: they are invariant under permutations of the vertices. To a first approximation, the $(D_G,D_x)$ level of the Local Statistics SDP, on input $G_0 \in \{0,1\}^{n \choose 2}$, will endeavor to find a degree-$D_x$ pseudoexpectation $\pseudo: \bbR[x]_{\le D_x} \to \bbR$ that (i) satisfies $\calB_k$, and (ii) obeys \textit{moment constraints} of the form
$$
    \pseudo p(G_0,x) \approx \expected_{(\bG,\bx)\sim \Planted_n} p(\bG,\bx)
$$
for symmetric polynomials $p \in \bbR[G,x]$ with degree $D_G$ in the $G$ variables. We ask that these moment constraints are only approximately satisfied to ensure that, when $(\bG,\bx)$ is drawn from the planted model, the pseudoexpectation $\pseudo p(G,x) \triangleq p(\bG,\bx)$ is with high probability a feasible solution. This formulation is inspired by the technique of pseudocalibration from the SOS lower bounds literature \cite{barak2019nearly,hopkins2017efficient,hopkins2017power}.

Each polynomial $p(G,x)$, when evaluated at a point in the zero locus described above, counts occurrences of a certain combinatorial structure in $G$, in which some of the vertices are restricted to have particular labels. For instance,
$$
    \sum_u \prod_{u\neq v}(1 - G_{u,v}) \qquad \text{and} \qquad \sum_{u\neq v} G_{u,v}x_{u,i}x_{v,j}
$$
count the number of isolated vertices, and the number of edges between vertices in groups $i$ and $j$, respectively. Note that since $\pseudo$ is required to satisfy the Boolean constraints on the $G$ variables and the $\calB_k$ constraints on the $x$ variables, we are free to consider only polynomials that have been reduced modulo these constraints: for simplicity we will assume that they are multilinear in $G$ and $x$, and furthermore that monomial containts $x_{u,i}x_{u,j}$ for $i\neq j$.

\begin{remark}
Although we have stated it in the specific context of the \model, the local statistics framework extends readily to any planted problem involving a joint distribution $\mu$ on pairs $(\bG,\bx)$ of a hidden structure and observed signal, if we take appropriate account of the natural symmetries in $\mu$. For a broad range of such problems, including spiked random matrix models \cite{alaoui2018fundamental, perry2016optimality}, compressed sensing \cite{zdeborova2016statistical, rangan2011generalized, kamilov2011optimal}  and generalized linear models \cite{Barbier5451} (to name only a few) there are conjectured computational thresholds where the underlying problem goes from being efficiently solvable to computationally intractable, and the algorithms which are proven or conjectured attain this threshold are often not robust. We hope that the local statistics hierarchy can be harnessed to design robust algorithms up to these computational thresholds, as well as to provide evidence for computational intractibility in the conjectured hard regime. The relation (if any) between the local statistics SDP hierarchy and iterative methods such as belief propagation or AMP is also worth investigating.
\end{remark}

\subsection{Analyzing the Local Statistics SDP}

By design, the Local Statistics SDP is always feasible when given as input a graph drawn from the planted model. To show that $\LS(2,m)$ can distinguish between the null and planted models, then, it suffices to show that it is with high probability infeasible when passed a graph from the null model.

For a matrix $C \in \R^{n \times n}$, let $C^{(t)}$ denote the $t^{th}$ ``non-backtracking power'' of the matrix:
\[ 
    C^{(t)}_{i,j} \defeq \sum_{\text{n.b. paths } p:i\to j} \prod_{(u,v) \in p} C_{u,v}\]
where the sum is over non-backtracking paths of length $t$ from $i$ to $j$.  
The local statistic that serves as a dual certificate to show infeasibility of $\LS(2,m)$ in the null model is given by,
\[ p^{(m)}(G,x) = \langle \phi(x), (A - (d/n)\bbJ)^{(m)} \phi(x)\rangle \]
for an appropriately chosen $\phi : [k] \to \bbR$.
In particular, we will see in the sections below that, if $\LS(2,m)$ SDP is feasible on input $G$, there is some matrix $X \succeq 0$ with unit trace and bounded entries on its diagonal for which
$$
    |\langle X, (A - (d/n)\bbJ)^{(m)}\rangle| \ge \omega(d^{m/2})n.
$$
The use of this centered non-backtracking walk matrix $\cnb{m}{\bG} = (A - (d/n)\bbJ)^{(m)}$  was inspired by the work of Fan and Montanari \cite{FM17}, who use the centered non-backtracking matrix for $m = 2$. Thus, to show infeasibility it would be sufficient to bound the spectral norm of the matrix $\cnb{m}{\bG} = (A - (d/n)\bbJ)^{(m)}$ by $d^{m/2}$ for sufficiently large constant $m$. 

In the $d$-regular case, the non-backtracking powers of the adjacency matrix $A$ can be expressed as univariate polynomials in the matrix $A$.  Thus spectral norm bounds on the adjacency matrix of a random $d$-regular graph \cite{friedman2003proof} can be translated into spectral norm bounds that we require.  This is roughly the approach taken in the $d$-regular case.

Unfortunately, things are not so simple in the irregular case: the analogous bound fails for constant $m$ due to the presence of high-degree vertices in $\bG$. 
The main challenge in studying $\cnb{m}{\bG}$, when $\bG$ is a sparse \ER random graph, is the presence of of certain localized combinatorial structures which inflate the number of non-backtracking walks: high-degree vertices and small subgraphs with many cycles.
Instead, we show the spectral norm bound after deleting these structures from the random graph $\bG$ and that the deletion does not affect the global statistic significantly.

Let us make this precise. In any graph $G$, write $\sfB_t(v,G)$ for the set of vertices with distance at most $t$ from $v$; call $v$ $(t,\eps)$\textit{-heavy} if $|\sfB_t(v,G)| \ge (1+\epsilon)^td^t$. We will call a vertex $v$ \textit{$(t,r,\eps)$-vexing} if either it participates in a cycle of length less than $r$ or it is $(t,\eps)$-heavy.

Fix $r = \Theta(\frac{\log n}{(\log logn)^2})$.
Let $\bG$ be an Erd\H{o}s-Renyi $G(n,d/n)$ graph, let $\bS$ its the set of $(t,r,\eps)$-vexing vertices, and let $\bG_{t,r,\eps}$ be the $(t,r,\eps)$-truncation obtained by deleting all the vertices in $\bS$ from $\bG$.  Let $\bA$ be the adjacency matrix of $\bG_{t,\eps,r}$.  Define
\[
    \left(\bA-\frac{d}{n}1_{[n]\setminus \bS}1^{\top}_{[n]\setminus\bS}\right)^{(\ell)}[u,v] = \sum_{\substack{W~\text{length-$\ell$ nonbacktracking walk}\\\text{from $u$ to $v$ in complete graph $K_{[n]\setminus S}$}}} \prod_{ij\in W} \left(\bA-\frac{d}{n}11^{\top}\right)[i,j]
\]
We prove the following spectral norm bound via the trace method:
\begin{theorem}  \label{thm:main-spec-norm-bound}
    With probability $1-n^{-100}$, $\displaystyle\left\|\locA\right\|\le \left((1+\eps)^4\sqrt{d}\right)^\ell$.
\end{theorem}

\subsection{Proving the spectral norm bound}
The proof of the above spectral norm bound is the most technical argument of the paper.
As expected, the proof of the spectral norm bound via trace method reduces to the problem of computing the expected number of copies of combinatorial structures that we call linkages in the underlying graph $\bG$ .
\begin{definition}[Linkages]
    A closed walk $W$ of length $k\ell$ is a \emph{$(k\times\ell)$-linkage} if it can be split into $k$ segments each of length-$\ell$ such that the walk $W$ is nonbacktracking on each segment.  Each $\ell$-step non-backtracking segment is a ``link''. 
\end{definition}
We will bound the number of $(k \times \ell)$-linkages using an encoding argument.  

It is instructive to consider the encoding argument in the case when the graph $\bG$ is a $d+1$-regular tree and the walk $W$ starts at the root.
Let us encode a $(k \times \ell)$-linkage starting at the root, one link at a time.
Each link which is a $\ell$-step n.b.walk in a tree consists of $t$-steps towards the root followed by $\ell-t$ steps away from the root for some $t \in \{0,\ldots,r\}$.
We refer to the steps towards the root as "up-steps" and steps away from the root as "down-steps".
Encode each link by specifying:
\begin{itemize}
    \item The number of up-steps $t$ using $ \log \ell $ bits.
    \item For each down-step, the index of the child as an integer from $\{1,\ldots,d\}$.
\end{itemize}
Since the walk begins and ends at the root, the number of up-steps is equal to the number of down-steps.  Therefore the number of down-steps is precisely $k\ell/2$.  Hence the above encoding uses precisely $k\ell/2 \cdot (\log d) + k \log \ell$ bits.  As $\ell \to \infty$, this is approximately $\frac{1}{2}\log d$ bits on average per step.   Therefore the number of $k \times \ell$-linkages starting at the root in a $d$-regular tree  is at most $((1+\epsilon) \sqrt{d})^{k \ell}$ for sufficiently large constant $\ell$.

In an Erdos-Renyi random graph $\bG$, there will be cycles of length $< k \ell$ thus breaking the above encoding argument.  In other words, if we consider the graph $G(W)$ formed by the edges in the $(k \times \ell)$-linkage $W$, then $G(W)$ can include cycles once we set $k = \Omega(\log n)$.  However, since we deleted all $(t,r,\epsilon)$-vexing vertices $G(W)$ has no cycles of length $< \Theta(\frac{\log n}{(\log \log n)^2}$.

The starting point of our encoding argument is a decomposition of $G(W)$ into a spanning forest $F$ and a few additional edges $E(W) \setminus F$, such that the non-forest edges $E(W) \setminus F$ are in total traversed $o(k \ell)$ times during the walk.  We prove the existence of such a decomposition using a linear programming based argument.

Roughly speaking, this decomposition lets us encode the walk $W$ by breaking it up into closed walks in trees, with the decomposition only introducing a negligible overhead in the  encoding.  Therefore, one recovers a bound analogous to the bound in a $d$-regular tree, which is approximately $\frac{1}{2} \log d$ bits per step in the walk.
%
%The details of the argument are 

The remainder of the paper will be laid out as follows. Before embarking on our investigation of the Local Statistics SDP in the DRBM and SBM in full generality, we will in \pref{sec:simple-DRBM} study a simplified SDP that can robustly solve the detection problem for the symmetric Degree Regular Block Model. Having done so, we will move on in \pref{sec:drbm} to the case of the general DRBM, proving \pref{thm:main-drbm} by way of a reduction to some key results from this simpler, symmetric case. Finally, in \pref{sec:SBM} we prove \pref{thm:main-sbm} regarding the SBM.

%% file: Content/coloring.tex
\newcommand{\Min}{M_{\text{in}}}
\newcommand{\Mout}{M_{\text{out}}}

\section{A Simplified SDP for the Symmetric DRBM}   \label{sec:simple-DRBM}

Many key ideas from the remainder of the paper are captured by the symmetric case of the Degree Regular Block Model, in which each group has size exactly $n/k$, and the edge probability matrix is
$$
    M = k\lambda\1 + (1-\lambda)\bbJ.
$$
Since the communities have equal sizes, we have $T = k^{-1}M$, and the Kesten-Stigum threshold is $d_{\ks}\triangleq \lambda^{-2} + 1$. Throughout this section, let $\Planted$ denote this symmetric case of the DRBM, and $\Null$ the uniform distribution on $d$-regular graphs. The purpose of this section is to show, in this symmetric case, that a simplified version of the Local Statistics SDP can robustly solve the detection problem.

To introduce this simpler SDP, let $G = (V,E)$ be any graph on $n$ vertices, and write $\nb{s}{G}$ for the $n\times n$ matrix that counts non-backtracking random walks of length $s$; we will develop some further theory regarding these matrices in Section 4.1 below. Now, let $(\bG,\by) \sim \Planted$ be drawn from the symmetric DRBM, and---thinking of $\by$ as an $n\times k$ matrix---write
\begin{align}
    \bY \triangleq \frac{k}{k-1}\left(\by\by^\ast - \frac{1}{k}\bbJ\right) \succeq 0.   \label{eq:planted-sol-simple}
\end{align}
This is a rank-$(k-1)$ positive semidefinite matrix that is $n/k$ times the projector onto the subspace spanned by the indicator vectors for the $k$ groups and orthogonal to the all-ones vector. The inner product $\langle \bY, \nb{s}{\bG}\rangle$ counts non-backtracking walks weighted according to the labels of their initial and terminal vertices.

\begin{lemma}   \label{lem:loc-stat-simple-drbm}
    Let $(\bG,\bY) \sim \Planted$. Then for every $s\ge 1$,
    $$
        \expected \langle \bY, \nb{s}{\bG}\rangle = \lambda^s d(d-1)^{s-1}n + o(n)
    $$
    and with high probability these quantities enjoy concentration of $o(n)$.
\end{lemma}

\begin{definition}
    Fix a small number $\delta >0$. The \textit{level $m$ symmetric path statistics SDP} with error tolerance $\delta>0$, on input $G_0$, is the feasibility problem
    \begin{align*}
        \text{Find $Y \succeq 0$ s.t. } & & Y_{u,u} &=1 & &\forall u \in [n] \\
        & & \langle Y, \bbJ\rangle &= 0 \\
        & & \left|\langle Y,\nb{s}{G} \rangle - \lambda^sd (d-1)^{s-1} n\right| &\le \delta n & &\forall s \in [m]  \numberthis \label{eq:simple-SDP}
    \end{align*}
    We will refer to this as the $SPS(m,\lambda)$ SDP. To handle adversarial edge corruption, it is necessary to include the following contingency if the input $G_0$ is not $d$-regular: before running the above SDP, delete all edges incident to vertices with degree greater than $d$, and then greedily add edges between vertices with degree less than $d$ to obtain a $d$-regular graph.
\end{definition}

\begin{theorem} \label{thm:main-simple-drbm}
    If $(d-1)\lambda^2 > 1$, then there exists constant $m \in \bbN$, $\delta > 0$, and $\rho >0$ so that $SPS(m,\lambda)$ solves the detection problem $\rho$-robustly. Conversely if $(d-1)\lambda^2$ then no such $m,\delta,\rho$ exist.
\end{theorem}

\subsection{Non-backtracking Walks and Orthogonal Polynomials}

\label{sec:nbw-poly}
The central tool in our proofs will be \emph{non-backtracking walks}---these are walks which on every step are forbidden from visiting the vertex they were at two steps previously. We will collect here some known results on these walks specific to the case of $d$-regular graphs. Write $\nb{s}{G}$ for the $n\times n$ matrix whose $(v,w)$ entry counts the number of length-$s$ non-backtracking walks between verties $v$ and $w$ in a graph $G$. It is standard that the $\nb{s}{G}$ satisfy a two-term linear recurrence,
\begin{align*}
	\nb{0}{G} &= \1 \\
	\nb{1}{G} &= A_{G} \\
	\nb{2}{G} &= A^2_{G} - d\1 \\
	\nb{s}{G} &= A\nb{s-1}{G} - (d-1)\nb{s-2}{G} \qquad s > 2,
\end{align*}
since to enumerate non-backtracking walks of length $s$, we can first extend each such walk of length $s-1$ in every possible way, and then remove those extensions that backtrack.

On $d$-regular graphs, the above recurrence immediately shows that $\nb{s}{G} = q_s(A_{G})$ for a family of monic, scalar \textit{non-backtracking polynomials} $\{q_s\}_{s\ge 0}$, where $\deg q_s = s$. To avoid a collision of symbols, we will use $z$ as the variable in all univariate polynomials appearing in the paper. It is well known that these polynomials are an orthogonal polynomial sequence with respect to the \textit{Kesten-McKay measure}
\[
	\dee\mu_{\km}(z) = \frac{1}{2\pi} \frac{d}{\sqrt{d-1}} \frac{\sqrt{4(d-1) - z^2}}{d^2 - z^2}\dee z\,\indicator{|z|< 2\sqrt{d-1}},
\]
with its associated inner product 
$$
    \langle f,g \rangle_{\km} \triangleq \int f(z)g(z) d\mu_{\km}(z)
$$
on the vector space of square integrable functions on $(-2\sqrt{d-1},2\sqrt{d-1})$. One quickly verifies that
\[
	\|q_s\|_\km^2 \triangleq \int q_s(z)^2\dee\mu_{\km} = q_s(d) =  \begin{cases} 1 & s = 0 \\ d(d-1)^{s-1} & s \ge 1 \end{cases} = \frac{1}{n}\left(\text{\# length-$s$ n.b. walks on $G$}\right)
\]
in the normalization we have chosen \cite{alon2007non}. Thus any function $f$ in this vector space can be expanded as
$$
    f = \sum_{s \ge 0} \frac{\langle f, q_s \rangle_{\km}}{\|q_s\|^2_{\km}} q_s.
$$

We will also need the following lemma of Alon et al. \cite[Lemma 2.3]{alon2007non} bounding the size of the polynomials $q_s$:

\begin{lemma}   \label{lem:NBW-poly-bound}
    For any $\varepsilon>0$, there exists an $\eta>0$ such that for $z\in[-2\sqrt{d-1}-\eta,2\sqrt{d-1}+\eta]$,
    $$
        |q_s(z)| \le 2(s+1)\|q_s\|_{\km} + \varepsilon.
    $$
\end{lemma}

The behavior of the non-backtracking polynomials with respect to the inner product $\langle \cdot,\cdot \rangle_{\km}$ idealizes that of the $\nb{s}{G} = q_s(A_{G})$ under the trace inner product. In particular, if $s + t < \girth(G)$
$$
    \langle \nb{s}{G},\nb{t}{G} \rangle = n\langle q_s,q_t \rangle_{\km} = \begin{cases} n (\text{\# length-$s$ n.b. walks on $G$}) & s=t \\ 0 & s \neq t \end{cases}.
$$
This is because the diagonal entries of $\nb{s}{G}\nb{t}{G}$ count pairs of non-backtracking walks with length $s$ and $t$ respectively: if $s\neq t$ any such pair induces a cycle of length at most $s+t$, leaving only the degenerate case when $s=t$ and the two walks are identical. Above the girth, if we can control the number of cycles, we can quantify how far the $\nb{s}{G}$ are from orthogonal in the trace inner product.

Luckily for us, sparse random graphs have very few cycles. To make this precise, call a vertex \emph{bad} if it is at most $L$ steps from a cycle of length at most $C$. These are exactly the vertices for which the diagonal entries of $\nb{s}{G}\nb{t}{G}$ are nonzero, when $s+t < C+L$.

\begin{lemma} \label{lem:bad-vtx}
    For any constant $C$ and $L$, with high probability any graph $\bG \sim \Planted$ has at most $O(\log n)$ bad vertices.
\end{lemma}

\noindent We will defer the proof of this lemma to the appendix, but one can immediately observe the consequence that, with high probability,
$$
    \langle \nb{s}{\bG},\nb{t}{\bG}\rangle = O(\log n)
$$
for any $s,t = O(1)$.

\subsection{Distinguishing}

Let us now prove the first assertion in \pref{thm:main-simple-drbm}, namely that if $(d-1)\lambda^2 > 1$, then the $SPS(m,\lambda)$ SDP, for some $\delta > 0$ sufficiently large $m$, can distinguish the null and planted models. From \pref{lem:loc-stat-simple-drbm}, if $(\bG,\bY) \sim \Planted$, then the matrix $\bY$ from equation \pref{eq:planted-sol-simple} is with high probability a feasible solution to SDP \pref{eq:simple-SDP}. Thus, it remains only to show that with high probability over $\bG \sim \Null$, some round of the $SPS(m,\lambda)$ SDP is infeasible. Our strategy will be to first reduce this infeasibility to a univariate polynomial design problem, and then solve this with the machinery developed in the prior subsection.

\begin{proposition} \label{prop:poly-implies-infeasibility}
    If there exists a degree-$m$ polynomial $f \in \bbR[z]$ which is (i) strictly nonnegative on the interval $[-2\sqrt{d-1},2\sqrt{d-1}]$ and (ii) satisfies
    $$
        \langle f, \sum_{s=0}^m \lambda^s q_s \rangle_{\km} < 0,
    $$
    then with high probability the $SPS(m,\lambda)$ SDP is infeasible for $\bG \sim \Null$, at any error tolerance
    $$
        \delta < \frac{|\langle f,\sum_{s = 0}^m \lambda^s q_s\rangle_{km}|}{\sqrt{m}\|f\|_{\km}}.
    $$
\end{proposition}

\begin{proof}
    First note that, for any such polynomial $f$, our discussion in the previous section implies
    \begin{align}
        f &= \sum_{s=0}^m \frac{\langle f,q_s\rangle_{\km}}{\|q_s\|^2_{\km}}q_m.
    \end{align}
    Moreover, since $f$ is strictly positive on $[-2\sqrt{d-1},2\sqrt{d-1}]$, it is nonnegative on some fattening $I$ of this interval.
    
    Now, let $\bG$ be a uniformly random $d$-regular graph. By Friedman's Theorem \cite{friedman2008proo}, the spectrum of $A_{\bG}$ consists of a `trivial' eigenvalue at $d$, plus $n-1$ eigenvalues whose magnitudes---with high probability---are at most $2\sqrt{d+1} + o_n(1)$. In particular, these remaining eigenvalues with high probability lie inside the fattening of $[-2\sqrt{d-1},2\sqrt{d-1}]$ on which $f$ is nonnegative. We can project away this trivial eigenvalue by passing to the centered adacency matrix $\overline A_{\bG} = (\1 - \bbJ/n)A_{\bG}(\1 - \bbJ/n) = A_{\bG} - d\bbJ/n$, and observe that $0 \preceq f(\overline A_{\bG}).$
    
    Assume, seeking contradiction, that $\bY$ is a feasible solution to the $SPS(m)$ SDP. We can compute that
    \begin{align*}
        0
        &\le \langle \bY, f(\overline A_{\bG})\rangle \\
        &= \langle \bY, \sum_{s=0}^m \frac{\langle f, q_s\rangle_{\km}}{\|q_s\|^2_{\km}}q_m(\overline A_{\bG}) \rangle \\
        &= \langle \bY, 
       \sum_{s=0}^m \frac{\langle f, q_s\rangle_{\km}}{\|q_s\|^2_{\km}}\left(q_m( A_{\bG}) - q_s(d)\bbJ/n\right)\rangle \\
       &= \langle \bY,        \sum_{s=0}^m \frac{\langle f, q_s\rangle_{\km}}{\|q_s\|^2_{\km}}\nb{s}{\bG} \rangle \\
       &\le \sum_{s=0}^m \frac{\langle f,q_s\rangle_{\km}}{\|q_s\|^2_{\km}}\cdot\lambda^s \|q_s\|^2_{\km} n + \delta \sum_{s = 0}^m \frac{|\langle f,q_s \rangle_\km|}{\|q_s\|^2_{\km}} \\
       &\le \langle f, \sum_{s=0}^m \lambda^s q_s \rangle + \delta \sqrt{m}\|f\|_{\km} < 0
    \end{align*}
\end{proof}

The following proposition implies a proof of the first part of \pref{thm:main-simple-drbm}.
\begin{proposition} \label{prop:exists-a-poly}
    If $\lambda^2(d-1) > 1$, there exists a polynomial satisfying the hypotheses of \pref{prop:poly-implies-infeasibility}.
\end{proposition}
\begin{proof}
    Call $m'$ the largest even number less than or equal to $m$, let $\varepsilon > 0$ be a very small number, and take
    $$
        f(z) = -q_{m'}(z) + 2m'\|q_{m'}\|_{\km} + \varepsilon,
    $$
    which by \pref{lem:NBW-poly-bound} has the desired positivity property. This choice of $f$ satisfies
    $$
        \langle f, \sum_{s=0}^m \lambda^s q_s \rangle = -\|q_{m'}\|^2_{\km} |\lambda|^{m'} + 2m'\|q_{m'}\|_{\km} + \varepsilon,    
    $$
    which is negative when
    $$
        \lambda^2 > \left(\frac{2m'}{\|q_{m'}\|_{\km}} + \frac{\varepsilon}{\|q_{m'}\|^2_{\km}}\right)^{\frac{2}{m'}} = \left(\frac{2m'}{\sqrt{d(d-1)^{m'-1}}} + \frac{\varepsilon}{d(d-1)^{m-1}}\right)^{\frac{2}{m'}};
    $$
    this tends to $\tfrac{1}{d-1}$ as $m \to \infty$.
\end{proof}

\subsection{Lower Bound}

We now turn to the complementary bound: when $(d-1)\lambda^2<1$, no constant level of the symmetric path statistics SDP can distinguish the null and planted distributions. It suffices to show that, for $d$ in this regime, $SPS(m,\lambda)$ is feasible for every constant $m$. Once again, we will reduce to and solve a univariate polynomial design problem.

\begin{proposition} \label{prop:poly-implies-lb}
    If there exists a polynomial $g \in \bbR[z]$ that is (i) strictly positive on $(-2\sqrt{d-1},2\sqrt{d-1})$, and (ii) satisfies
    $$
        \langle g, q_s \rangle_{\km} = \lambda^s\|q_s\|^2_{\km} \qquad \text{For all $s = 0,...,m$},
    $$
    then the $SPS(m,\lambda)$ SDP at any constant error tolerance $\delta > 0$ is with high probability feasible for a uniformly random $d$-regular graph.
\end{proposition}

\begin{proof}
    Letting $\bG$ be the random regular graph in question, and fixing arbitrary $\delta > 0$, we need to produce $Y \succeq 0$ with ones on the diagonal, zero inner product with the matrix $\bbJ$, and satisfying
    $$
        \left|\langle Y,\nb{s}{\bG}\rangle - \lambda^s \|q_s\|^2_{\km} n \right| \le \delta n.
    $$
    Our strategy will be to modify the matrix $g(\overline A_{\bG}) = g(A_{\bG}) - g(d)\bbJ/n$. 
    
    First, note that by expanding $g$ in the non-backtracking basis and invoking \pref{lem:bad-vtx}, for any $0 \le s \le m$ we have
    $$
        \langle g(\overline A_{\bG}), \nb{s}{\bG}\rangle =
        \langle g(A_{\bG}), \nb{s}{\bG}\rangle + g(d)\|q_s\|^2_{\km} =
        \lambda^s \|q_s\|^2_{\km}\cdot n + O(\log n),
    $$
    since $g(d)\|q_s\|^2_{\km}$ is a constant. Moreover, as $g$ is strictly positive on $[-2\sqrt{d-1},2\sqrt{d-1}]$ it is by continuity nonnegative on any constant size fattening of this interval, and by Friedman's theorem the spectrum of $A_{\bG}$ other than the eigenvalue at $d$ is contained w.h.p. in such a set. Thus $g(\overline A_{\bG})$ is positive semidefinite, and as a polynomial in the centered adjacency matrix, is orthogonal to the all-ones matrix.
    
    However, the diagonal of $g(\overline A_{\bG})$ may not be equal to one, for two different reasons. The diagonal entries of $g(A_{\bG}) = g(\overline A_{\bG}) + g(d)\bbJ/n$ different from one are exactly those corresponding to vertices within $\deg g$ steps of a constant length cycle; from \pref{lem:bad-vtx} we know that there are at most $O(\log n)$ of these \textit{bad} vertices (keeping the terminology from the aforementioned Lemma). However, when we subtract $g(d)\bbJ/n$, even the $\Omega(n - \log n)$ diagonal entries equal to one---those corresponding to \textit{good} vertices---are shifted. Let us therefore define
    $$
        \tilde Y = \frac{1}{1 - g(d)/n} g(\overline A_{\bG}),
    $$
    which restores the diagonal entries of the good vertices.
    
    Now, $\tilde Y$ is PSD, and is accordingly the Gram matrix of some vectors $\alpha_1,...,\alpha_n \in \bbR^n$. The scale factor we have applied ensures that for every good vertex $u$, $\|\alpha_u\| = 1$, and orthogonality to the all-ones matrix---which is preserved by this constant scaling---is equivalent to $\sum_u \alpha_u = 0$.
    
    The remaining diagonal elements are at worst some constant $C$ dependent on $d$ and $g$, since the diagonal entries of each $\nb{s}{\bG}$ are all $O(1)$. Thus, writing $\Gamma$ for the set of good vertices, we know
    $$
        \left\|\sum_{u \in \Gamma} \alpha_u \right\| = \left\|\sum_{u\notin \Gamma} \alpha_u \right\| \le C\log n
    $$
    It is clear that by removing at most $C\log n$ vertices from $\Gamma$ to create a new set $\Gamma'$ we can choose a collection of unit vectors $\beta_u$ for each $u \in U'$ so that $$
        \sum_{u \notin \Gamma'} \beta_u = \sum_{u\in \Gamma'} \alpha_u. 
    $$
    Our final matrix $Y$ will be the Gram matrix of these new $\beta$ and remaining $\alpha$ vectors. We must finally check that the affine constraints against the $\nb{s}{\bG}$ matrices are still approximately satisfied. However, even starting from a bad vertex, there are at most a constant number of vertices within $s$ steps of it, and at most a constant number of non-backtracking walks to any such vertex. Thus
    \begin{align*}
        \left|\langle Y,\nb{s}{\bG} \rangle - \langle \tilde Y,\nb{s}{\bG} \rangle \right| &= \left| 2\sum_{u\notin \Gamma',v\in \Gamma'} (\nb{s}{\bG})_{u,v}\alpha_u^T(\alpha_v - \beta_v) + \sum_{u,v\notin \Gamma'} (\nb{s}{\bG})_{u,u}\left(\|\alpha_u\| - \|\beta_u\|\right)\right| \\
        &= O(\log n)
    \end{align*}
    where we have used that $\max_u \|\alpha_u\| = O(1)$ and broken up both summations by first enumerating the $O(\log n)$ vertices in $U'$ and then the at most $O(1)$ vertices in its depth $s$ neighborhood. Thus, for our fixed $\delta > 0$, we have
    $$
        \left| \langle \tilde Y, \nb{s}{\bG} \rangle - \lambda^s\|q_s\|^2_{\km}\right| = O(\log n) \le \delta n
    $$
    for $n$ sufficiently large.
\end{proof}
The second part of \pref{thm:main-simple-drbm} ensues from the following proposition.
\begin{proposition} \label{prop:poly-exists-lb}
    Whenever $\lambda^2(d-1) < 1$, there exists a polynomial satisfying the conditions of \pref{prop:poly-implies-lb}.
\end{proposition}

\begin{proof}
       Such a polynomial $y$ is exactly of the form
    $$
        g = \sum_{s = 0}^m \lambda^s q_s + \text{terms with larger $q_s$'s}.
    $$
    We will use the extremely simple construction of letting the coefficients on the terms $q_{m+1},q_{m+1}, \cdots$ \emph{also} be powers of $\lambda$. The idea here is that, whenever $\lambda^2(d-1) < 1$, the series $\sum_{s\ge 0}\lambda^sq_s$ converges to a positive function on $(-2\sqrt{d-1},2\sqrt{d-1})$, so by taking a long enough initial segment, we can get a positive approximant. 

    In particular, let $p \gg m$ be even, and set
    $$
        g = \sum_{s = 0}^p \lambda^s q_s.
    $$
    It is a standard calculation, employing the recurrence relation on the polynomials $q_s$, that
    $$
        g(z) = \frac{1 - \lambda^2 + \lambda^{p+2}(d-1)q_p(z) - \lambda^{p+1}q_{p+1}(z)}{(d-1)\lambda^2 - \lambda z + 1}.
    $$
    One an quickly verify that
    $$
        \frac{1 - \lambda^2}{(d-1)\lambda^2 - \lambda z + 1} > 0 \qquad \text{for all } |z| \le 2\sqrt{d-1},
    $$
    so we only need to check that $\lambda^2(d-1) < 1$ ensures $\lambda^{p+2}(d-1)q_p - \lambda^{p+1}q_{p+1} \to_p 0$. This follows immediately from \pref{lem:NBW-poly-bound}, as $|q_p| \le 2p\sqrt{d(d-1)^p}$.
\end{proof}

\subsection{Robustness} \label{sec:col-robustness}

We have shown already that if $(d-1)\lambda^2 > 1$, then for some constant $m(\lambda)$ and error tolerance $\delta(\lambda) > 0$, the level $m$ symmetric path statistics SDP can solve the detection problem, and that otherwise no such $\delta $ and $m = O(1)$ can exist. In this section we show that this result is \textit{robust}. To do so, we need to argue (i) that when $\bG \sim \Planted$, or $\bG \sim \Null$ with $(d-1) \lambda^2 < 1$, the SDP with high probability remains feasible for any error tolerance $\delta$, even after perturbing $\rho n$ edges, and (ii) that when $\bG \sim \Null$ and $(d-1)\lambda^2 > 1$, for some $\rho > 0$ and $\delta ' < \delta(\lambda)$, the SDP remains infeasible at tolerance $\delta'$, even after perturbing $\rho n$ edges.

Assume that $\bG$ was drawn from either the planted or null distribution, and that $\tilde\bH \approx_\rho \bG$. When we defined the $SPS(m,\lambda)$ SDP, we stipulated that in the event of an irregular input, we greedily remove edges until the maximum degree is $d$, and then greedily add edges among degree-deficient vertices until the minimum degree is $d$ as well. Thus the actual input to the SDP is a graph $\bH$, which one can verify satisfies $\bH \approx_{\rho\xi} \bG$ for some absolute constant $\xi$. Call a vertex $v \in [n]$ \emph{corrupted} if its $(m+1)$-neighborhood in $\bH$ differs from its $(m+1)$-neighborhood in $\bG$.  We begin by analyzing the difference $\nb{s}{\bG} -\nb{s}{\bH}$ for $s\in[m]$. Supposing $v$ is not a corrupted vertex, then $\nb{s}{\bG}$ and $\nb{s}{\bH}$ agree on the $v$th row and column, which means $(\nb{s}{\bG}-\nb{s}{\bH})_{v,:} = 0$. On the other hand, if $v$ is a corrupted vertex,

\begin{align*}
	\left\|\left(\nb{s}{\bG}-\nb{s}{\bH}\right)_{v,-}\right\|_1 &\le \left\|\nb{s}{\bG}\right\|_1 + \left\|\nb{s}{\bH}\right\|_1
	\le 2d(d-1)^{s-1}
\end{align*}

\noindent In particular, this means the entrywise 1-norm of $\nb{s}{\bG}-\nb{s}{\bH}$, is bounded by $2\xi\rho n\cdot 2d(d-1)^{s-1}$ since there are at most $2\xi\rho n$ corrupted vertices (i.e. if all corrupted edges had disjoint endpoints).

To prove (i), assume that the SDP is feasible at error tolerance $\delta$ on input $\bG$, and write $Y$ for a solution. Then
$$
    \left|\langle Y,\nb{s}{\bH}\rangle - \langle Y,\nb{s}{\bG}\rangle\right| \le \left\|\nb{s}{\bH} - \nb{s}{\bG}\right\|_1 \le 2\xi\rho d(d-1)^{s-1},
$$
and thus $Y$ is feasible on input $\bH$ with error tolerance
$$
    \delta' = 2\xi\rho d(d-1)^{m-1} + \delta.
$$
Since on $\bG \sim \Planted$, or $\bG \sim \Null$ with $(d-1)\lambda^2 < 1$ the SDP is feasible for every $\delta > 0$, we can take $\delta \to 0$, and find we are free to choose $\rho$ so long as we work at tolerance $2\xi\rho d(d-1)^m$.

To prove (ii), assume $\bG \sim \Null$ and $(d-1)\lambda^2 > 1$. Infeasibility of the SDP on input $\bG$ is witnessed by the polynomial $f$ from \pref{prop:exists-a-poly}. So, let $\bY$ be a putative solution to the SDP on input $\bH$, at tolerance $\delta'$, seeking a contradiction: recycling some computations from the proof of \pref{prop:poly-implies-infeasibility}
\begin{align*}
    0 &\le \langle \bY, f(\overline A_{\bG}) \rangle \\
    &= \langle \bY, \sum_{s = 0}^m \frac{\langle f, q_s \rangle_{\km}}{\|q_s\|^2_{\km}}\nb{s}{\bG} \rangle \\
    &\le \sum_{s = 0}^m \frac{\langle f, q_s \rangle_{\km}}{\|q_s\|^2_{\km}}\left(\langle \bY, \nb{s}{\bH}\rangle \pm 2\rho\xi d(d-1)^s\right) \\
    &\le \langle f, \sum_{s = 0}^m \lambda^s q_s \rangle_{\km} + 2\rho\xi \sum_{s = 0}^m |\langle f, q_s \rangle|_{\km} + \delta'\sqrt{m}\|f\|_{\km} \\
    &\le \langle f, \sum_{s = 0}^m \lambda^s q_s \rangle_{\km} + (\delta' + 2\rho\xi)\sqrt{m}\|f\|_{\km}.
\end{align*}
Thus we have a contradiction if
$$
    \delta' < \frac{|\langle f, \sum_{s = 0}^m \lambda^s q_s \rangle_{\km}|}{\sqrt m \|f\|_{\km}} - 2\rho\xi.
$$
Here our choice of $\rho$ must be constrained so that the right hand side of this expression is positive. This indicates a tradeoff between proximity to the KS threshold and robustness.

% From \pref{rem:dual-slack}, if $\bG$ is drawn from the planted distribution, the feasible solution $Y$ that we constructed is PSD and satisfy the affine constraints regarding inner products with the $\nb{s}{\bG}$ matrices with slack $\Omega(n)$. Every diagonal entry of $Y$ is one, so by PSD-ness their off-diagonal entries have modulus at most one. Thus
% $$
%     \left|\langle Y,\nb{s}{\bG}  \rangle - \langle Y,\nb{s}{\bH}\rangle \right| = \left|\langle Y,  \nb{s}{\bG} - \nb{s}{\bH}\rangle\right| \le \left\|\nb{s}{\bG} - \nb{s}{\bH}\right\|_1 \le 2\xi\rho d(d-1)^{s-1}.
% $$
% Because of the $\Omega(n)$ slack, the $Y$ we constructed from $\bG$ will \textit{still} satisfy the affine constraints to the SDP on input $\bH$, for small enough $\rho$.

% On the other hand, when $\bG$ is drawn from the null model, we noted in \pref{rem:primal-slack} that any putative solution $Y$ with ones on the diagonal and zero inner product with $\bbJ$ violates some linear combination of the above affine constraints by a margin of $\Omega(n)$. Thus, if we try these consraints with $\bH$ instead of $\bG$, this constraint will still be violated for $\rho$ sufficiently small.

% \begin{remark}
%     The parameter $\rho$ controlling the number of adversarial edge insertions and deletions made to random input $\bG$ that the level-$m$ $LocalStatistic$ SDP can tolerate can be seen to decrease with $m$, which is indicative of a tradeoff between how close to the threshold an algorithm in this hierarchy works and how robust it is to perturbations.
% \end{remark}

%% file: Content/drbm.tex
\section{The Degree Regular Block Model} \label{sec:drbm} % (fold)

In this section we generalize the results from the previous section in two ways simultaneously: we study the fully general Degree Regular Block Model, and the full Local Statistics SDP. Both add some technical hurdles, but we will find that once these have been dealt with, the core arguments reduce to the symmetric results from \pref{sec:simple-DRBM}. Throughout, assume that $\calN$ is the uniform distribution on $d$-regular graphs, and $\Planted$ is the DRBM with fixed parameters $(d,k,M,\pi)$.  In this section we prove \pref{thm:main-drbm}.

\subsection{Local Statistics and Partially Labelled Subgraphs}  \label{sec:local-stat-treat} % (fold)

As in the introduction let $x = \{x_{u,i}\}$ and $G = \{G_{u,v}\}$ be sets of variables indexed by $u \in [n]$ and $i \in [k]$. Our random graphs $\bG$ and community labels $\bx$ take values in the subset of $\{0,1\}^{n\choose 2} \times \{0,1\}^{n\times k} \subset \bbR^{{n\choose 2}} \times \bbR^{n\times k}$ defined by the polynomial equations
\begin{align*}
    G_{u,v}^2 - G_{u,v} &= 0 \\
    x_{u,i}^2 - x_{u,i} &= 0 \\
    \sum_i x_{u,i} - 1 &= 0 \numberthis \label{eq:graph-var-bool}
\end{align*}
as in the introduction, we will write the ideal generated by the polynomials on the left of the second two equations as $\calB_k$. Any point $x$ in the vanishing locus of $\calB_k$ corresponds to a map $\sigma_x : [n] \to [k]$. Write $\bbS[G,x] \subset \bbR[G,x]$ for the vector subspace of multilinear polynomials, fixed under the action of the symmetric group $\frS_n$ on the index set $[n]$, and for which no monomial contains $x_{u,i}x_{u,j}$ for $i\neq j$. This contains some polynomials that vanish modulo the equations above, but is convenient to work with.

The local statistics SDP, given as input a graph $G_0 \in \{0,1\}^{n\choose 2}$, attempts to find a pseudoexpectation $\pseudo : \bbR[x] \to \bbR$ that (i) evaluates to zero on any polynomial in $\calB_k$, and (ii) assigns certain prescribed values to polynomials $p(G_0,x)$ obtained by evaluating a low-degree-polynomial $p \in \bbS[G,x]$ at the input graph. To state it fully, we will first construct a combinatorially  meaningful vector space basis for $\bbS[G,x]$.

\begin{definition}[Partially Labelled Subgraph]
    A \textit{partially labelled graph} $(H,S,\tau)$ consists of a graph $H$, distinguished subset of vertices $S \subset V(H)$, and a labelling $\tau : S \to [k]$. An \textit{occurrence} of $(H,S,\tau)$ in a fully labelled graph $(G,\sigma)$ is an injective homomorphism  $\varphi : H \to G$ which respects the labelling. In other words, it is an injective map $\varphi: V(H) \to V(G)$ satisfying (i) $(\varphi(u),\varphi(v)) \in E(G)$ for every edge $(u,v) \in E$, and (ii) $\sigma(\varphi(v)) = \tau(v)$ for every $v \in S$.
\end{definition}

\begin{lemma}[Partially Labelled Subgraphs are a Basis] \label{lem:def-ls-poly}
    Let $(H,S,\tau)$ be a partially labelled subgraph. Then there is a symmetric polynomial $p_{H,S,\tau} \in\bbR[G,x]$ with degree $|S|$ in $x$ and $|E(H)|$ in $G$ that, for any $(\bG,\bx)$ satisfying equations \pref{eq:graph-var-bool}, counts occurrences of $H$ in $(\bG,\sigma_{\bx})$. Furthermore, these polynomials form a basis for $\bbS[G,x]$.
\end{lemma}

\begin{proof}
    These polynomials are exactly the \textit{monomial basis} obtained by considering the $\frS_n$ orbit of each multilinear monomial in $G$ and $x$ which does not conatin $x_{u,i}x_{u,j}$ for $i,j \in [k]$. Each such monomial is of the form
    $$
        \prod_{(u,v) \in E}G_{u,v} \prod_{u\in S}x_{u,\tau(u)},
    $$
    where $E \subset {[n] \choose 2}$, $S \subset [n]$, and $\tau : S \to [k]$. Letting $H$ be the graph whose vertices are those present either in $S$ or in one of the pairs in $E$, when this monomial is evaluated at $(G_0,x_0)$ satisfying the above equations, it is simply the indicator for one occurrence of $(H,S,\tau)$. By symmetrizing with respect to $\frS_n$, one obtains indicators for all possible such occurrences.
\end{proof}

The Local Statistics $L(2,m)$, on input $G_0$, contains constraints of the form
$$
    \pseudo p_{H,S,\tau}(G_0,x) \approx \expected_{(\bG,\bx) \sim \Planted} p_{H,S,\tau}(\bG,\bx).
$$
where $|S| \le 2$ and $|E(H)| \le m$. The following theorem computes the right hand side of the above equation in the planted, for this class of partially labelled subgraphs. We will discuss it briefly below and remit the proof to the appendix. Let $(H,S,\tau)$ be a partially labelled graph, and define
\begin{align}
    C_H(d) & \triangleq \frac{\prod_{v \in V(H)} (d)_{\deg(v)}}{d^{|E(H)|}} \label{eq:c-def} \\
    L_{(H,S,\tau)}(M,\pi) &\triangleq \sum_{\hat \tau : \hat \tau|_S = \tau} \prod_{v \in v(H)} \pi(\hat\tau(v)) \prod_{(u,v) \in E(H)} M_{\hat\tau(u),\hat\tau(v)}. \label{eq:l-def}
\end{align}
Here $(d)_s = d(d-1) \cdots (d-s+1)$ is the falling factorial, and the sum in the second line is over all $\hat\tau : V(H) \to [k]$ which agree with $\tau$ on $S$. 
Define also $\chi(H) = |V(H)| - |E(H)|$ and $c(H) = \text{\# connected components of $H$}$. 

\begin{theorem}[Local Statistics]
    \label{thm:local-stats}
    Let $(H,S,\tau)$ be a partially labelled graph with $O(1)$ edges. Then, with high probability over $(\bx,\bG) \sim \Planted$,
    $$
        p_{(H,S,\tau)}(\bx,\bG) = n^{\chi(H)} L_{(H,S,\tau)}(M,\pi) \cdot C_H(d) \pm o(n^{c(H)}).
    $$
\end{theorem}

The proof may be found in  \pref{app:loc-stat-calc}, but some comments are in order here. First, when $H$ is a forest and $\chi(H) = c(H)$, we see that $p_{(H,S,\tau)}(\bx,\bG)$ concentrates, and that
$$
    n^{-c(H)} p_{(H,S,\tau)}(\bx,\bG) \to L_{(H,S,\tau)}(M,\pi) C_H(d),
$$
Conversely, it is well-known that (for instance) the number of cycles in $\bG$ is Poisson distributed with constant mean, and thus all we can say with high probability is that there are $o(n)$ of them. This fact is reflected in greater generality in the discrepency between the $O(n^{\chi(H)})$ and $O(n^{c(H)})$ scales of $p_{(H,S,\tau)}(\bx,\bG)$ and its fluctuations, respectively, when $H$ contains at least one cycle. Since we need to give the Local Statistics algorithm affine constraints that are satisfied with high probability in the planted model, we will include these two distinct scales in our full statement of the algorithm.

Second, the constants $L_{(H,S,\tau)}(M,\pi)$ and $C_H(d)$ have a pleasant interpretation in the case when $H$ is a forest. If $G$ is an unlabelled and locally treelike $d$-regular graph, in the sense that the shortest cycle is much larger than the longest path in $H$, then there are exactly $n^{c(H)} C_H(d)$ injective homomorphisms of $H$ into $G$. On the other hand, $L_{(H,S,\tau)}(M,\pi)$ describes the probability of a certain outcome in a natural Markov process: start at some vertex $s\in S$, choose its label $i$ according to $\pi$, and for each neighbor choose a label $j$ with probability $T_{i,j} = M_{i,j}\pi(j)$. If one continues this until all of $H$ is labelled, $L_{(H,S,\tau)}(M,\pi)$ gives the probability that every vertex $s \in S$ is given label $\tau(s)$.

We may finally define formally the Local Statistics algorithm.

\begin{definition}
    \label{def:lost-formal}
    The \textit{degree $(D_x,D_G)$ Local Statistics algorithm} with error tolerance $\delta > 0$, on input $G_0$, is the following SDP: find a pseudoexpectation $\pseudo : \bbR[x]_{\le D_x} \to \bbR$ that is positive, normalized, satisfies $\calB_k$, and for which
    $$
        \pseudo p_{(H,S,\tau)}(x,G_0) = n^{\chi(H)} L_{(H,S,\tau)}(M,\pi) C_H(d) \pm \delta n^{c(H)}
    $$
    for every $(H,S,\tau)$ with $|S| \le D_x$ and $|E(H)| \le D_G$.
\end{definition}

\begin{lemma}
    For any $\delta > 0$, the $\LS(D_x,D_G)$ algorithm is with high probability feasible on input $\bG \sim \Planted$.
\end{lemma}

\begin{proof}
    Let $\bx$ be the hidden signal; we will set $\pseudo p_{(H,S,\tau)}(x,\bG) = p_{(H,S,\tau)}(\bx,\bG)$. This is clearly positive, satisfies $\calB_k$, and from Theorem \ref{thm:local-stats} it satisfies the affine constraints in Definition \ref{def:lost-formal}.
\end{proof}

\subsection{Distinguishing} \label{sec:distinguish-drbm} % (fold)

Let us prove the first part of \pref{thm:main-drbm}: when $(d-1)\lambda_2^2 > 1$, there exist constant $m,\rho,$ and $\delta >0$ for which the $\LS(2,m)$ SDP at error tolerance $\delta$ solves the detection problem $\rho$-robustly. Since the SDP is with high probability feasible for any $m$ and $\delta> 0$ when $\bG \sim \Planted$, it remains only to show infeasibility for some $m,\rho,\delta$ when $\bG \sim \Null$.

Let $\bG \sim \Null$, and assume we have a viable pseudoexpectation $\pseudo$ for the $\LS(2,m)$ SDP with some tolerance $\delta > 0$. Write $X \succeq 0$ for the $nk\times nk$ matrix whose $(u,i),(v,j)$ entry is $\pseudo x_{u,i}x_{v,j}$; it is routine that positivity of $\pseudo$ implies positive semidefiniteness of $X$. It will at times be useful to think of $X$ as a $k\times k$ matrix of $n\times n$ blocks $X_{i,j}$, and at others as an $n\times n$ matrix of $k\times k$ blocks $X_{u,v}$. Let us also define matrices $\sa{s}{\bG}$ that count \textit{self-avoiding} walks of length $s$, as opposed to the non-backtracking walks counted by the matrices $\nb{s}{\bG}$ whose notation they echo. Our strategy will be to first write the moment matching constraints on $\pseudo$ as affine constraints of the form $\langle X_{i,j}, Y \rangle = C$, and then combine these to contradict feasibility of $X$.

\begin{lemma} \label{lem:sdp-affine}
    For any $i,j$, and any $s = 0,...,m$, recall that $\nb{s}{G}$ is the matrix counting non-backtracking walks of length $s$, and $\bbJ$ is the all-ones matrix. For any $\delta' > \delta$,
    \begin{align*}
        \langle X_{i,j}, \nb{s}{\bG} \rangle &= \pi(i) T_{i,j}^s \|q_s\|^2_{\km} n \pm \delta' n \\
        \langle X_{i,j}, \bbJ \rangle &= \pi(i)\pi(j)n^2 \pm \delta' n^2
    \end{align*}
\end{lemma}

\begin{proof}
    For the first assertion, let $(H,S,\tau)$ be the path of length $s$ whose endpoints are labelled $i,j\in[k]$. In this case $C_H(d) = d(d-1)^{s-1} = \|q_s\|^2_{\km}$, and one can quickly verify that $L_{(H,S,\tau)} = \pi(i) T_{i,j}^s$. Each \textit{self-avoiding} walk of length $s$ in $G$ is an occurrence of $H$, so from \pref{thm:local-stats}
    $$
        \langle X_{i,j},\sa{s}{\bG}\rangle = \pseudo p_{H,S,\tau}(x,\bG) = \pi(i)M_{i,j}^s \|q_s\|^2_{\km}n \pm \delta n
    $$
    It is an easy consequence of \pref{lem:bad-vtx} that for every constant $s$, $\nb{s}{\bG}$ and $\sa{s}{\bG}$ differ only on $O(\log n)$ rows, and since each row has constant $L_2$ norm,
    $$
        \left\|\nb{s}{\bG} - \sa{s}{\bG}\right\|^2_{F} = O(\log n).
    $$
    The matrix $X$ has diagonal elements $X_{(u,i),(u,i)} = \pseudo x_{u,i}^2 = \pseudo x_{i,u}$ by the Boolean constraint, and $\pseudo (x_{u,1} + \cdots + x_{u,k}) = 1$ by the Single Color constraint. By PSD-ness of $X$, every $\pseudo x_{u,i}^2 = \pseudo x_{u,i}$ is nonnegative, so each is between zero and one. It is a standard fact that the off-diagonal entries of such a PSD matrix have magnitude at most one, so from \pref{lem:NBW-poly-bound}
    $$
        \langle X_{i,j}, \nb{s}{\bG} \rangle = \langle X_{i,j}, \sa{s}{\bG} \rangle + \langle X_{i,j} \sa{s}{\bG} - \nb{s}{\bG}\rangle = \langle X_{i,j}, \sa{s}{\bG} \rangle \pm O(\log n) = \pi(i)M^s_{i,j} \|q_s\|^2_{\km}n \pm \delta' n
    $$
    for $s = 0,...,m$, any $\delta' > \delta$, and $n$ sufficeintly large. For the second assertion, when $i\neq j$ take $(H,S,\tau)$ to be the partially labelled graph on two disconnected vertices, with labels $i$ and $j$ respectively. In this case $C_H(d) = 1$, and $L_{(H,S,\tau)}(M,\pi) = \pi(i)\pi(j)$. We then have
    $$
        \langle X_{i,j},\bbJ \rangle = \pseudo p_{H,S,\tau}(x,\bG) = \pi(i)\pi(j)n^2 \pm \delta n^2
    $$
    For the case $i = j$, let $(H',S',\tau')$ be a single vertex labelled $i$, for which $C_H(d) = 1$ and $L_{(H',S',\tau')}(M,\pi) = \pi(i)$. We can write
    $$
        \langle X_{i,i}, \bbJ \rangle = \pseudo p_{(H,S,\tau)}(x,\bG) + \pseudo p_{(H',S',\tau')}(x,\bG) = \pi(i)^2 n^2 + \pi(i) n \pm \delta n^2 = \pi(i)\pi(j) n^2 + \delta' n^2
    $$
    for any $\delta' > \delta$ and $n$ sufficiently large.
\end{proof}

We will now apply a fortuitous change of basis furnished to us by the transition matrix $T$. Let us write $F$ for the matrix of right eigenvectors of $T$, normalized so that every column has unit norm, and sorted so that the first column is a multiple of the all-ones vector. Thus $TF = F\Lambda$, where $\Lambda$ is a diagonal matrix containing the eigenvalues, sorted in decreasing order of magnitude. It is a standard fact from the theory of reversible Markov chains that $F^{-1}\Diag(\pi)F = \1$.

Now, define a matrix $\check X \triangleq (F^T \otimes \1) X (F\otimes \1)$, by which we mean that 
$$
    \check X = \begin{pmatrix} F_{1,1}\1 & \cdots & F_{1,k}\1 \\
    \vdots & \ddots & \vdots \\
    F_{k,1}\1 & \cdots & F_{k,k}\1 
    \end{pmatrix} 
    \begin{pmatrix} X_{1,1} & \cdots & X_{1,k} \\
    \vdots & \ddots & \vdots \\
    X_{k,1} & \cdots & X_{k,k} 
    \end{pmatrix}
    \begin{pmatrix} F_{1,1}\1 & \cdots & F_{1,k}\1 \\
    \vdots & \ddots & \vdots \\
    F_{k,1}\1 & \cdots & F_{k,k}\1 
    \end{pmatrix}.
$$
We will think of $\check X$, analogous to $X$, as a $k\times k$ matrix of $n\times n$ blocks $\check X_{i,j}$. Note that we can also think of this as as a change of basis $x \mapsto F^T x$ directly on the variables appearing in polynomials accepted by our pseudoexpectation.

\begin{lemma} \label{lem:sdp-affine-check}
    For any $s=0,...,m$, and any $\delta'' > \|F\|^2\sqrt k \delta$, we have
    \begin{align*}
        \langle \check X_{i,j} \nb{s}{\bG}\rangle &= \begin{cases} 0 & i\neq j \\ \lambda_i^s \|q_s\|^2_{\km} & i = j \end{cases} \,\,\,\pm \delta'' n \\
        \langle \check X_{i,j}, \bbJ \rangle &= \begin{cases} n^2 & i=j=1 \\ 0 & \text{else} 
        \end{cases} \,\,\, \pm \delta'' n^2
    \end{align*}
\end{lemma}

\begin{proof}
    Our block-wise change of basis commutes with taking inner products between the blocks $X_{i,j}$ and the non-backtracking walk matrices. In other words, invoking Lemma \ref{lem:sdp-affine} with $\delta' > \delta$ and keeping track of how the additive errors compound as we take linear combinations,
    \begin{align*}
        \begin{pmatrix} \langle \check X_{1,1},\nb{s}{\bG}\rangle & \cdots & \langle \check X_{1,k},\nb{s}{\bG}\rangle \\
        \vdots & \ddots & \vdots \\
        \langle \check X_{k,1},\nb{s}{\bG}\rangle & \cdots & \langle \check X_{k,k},\nb{s}{\bG}\rangle 
        \end{pmatrix}_{i,j}
        &= \left(F^T\begin{pmatrix}
        \langle X_{1,1},\nb{s}{\bG}\rangle & \cdots & \langle X_{1,k},\nb{s}{\bG}\rangle \\
        \vdots & \ddots & \vdots \\
        \langle X_{k,1},\nb{s}{\bG}\rangle & \cdots & \langle X_{k,k},\nb{s}{\bG}\rangle 
        \end{pmatrix} F\right)_{i,j} \\
        &= \left(F^T \Diag(\pi)T^s F\right)_{i,j} \cdot \|q_s\|^2_{\km} n \pm \|F\|^2\sqrt k \delta' n  \\
        &= \left(F^T\Diag(\pi)F\Lambda^s\right)_{i,j} \cdot \|q_s\|^2_{\km}n \pm \|F\|^2\sqrt{k}\delta' n\\
        &= \Lambda^s_{i,j} \cdot \|q_s\|^2_{\km}n \pm \|F\|^2\sqrt k \delta' n.
    \end{align*}
    A parallel calculation gives us 
    \begin{align*}
        \begin{pmatrix} \langle \check X_{1,1},\bbJ\rangle & \cdots & \langle \check X_{1,k},\bbJ\rangle \\
        \vdots & \ddots & \vdots \\
        \langle \check X_{k,1},\bbJ\rangle & \cdots & \langle \check X_{k,k},\bbJ\rangle 
        \end{pmatrix}_{i,j}
        &= \left(F^T\begin{pmatrix}
        \langle X_{1,1},\bbJ\rangle & \cdots & \langle X_{1,k},\bbJ\rangle \\
        \vdots & \ddots & \vdots \\
        \langle X_{k,1},\bbJ\rangle & \cdots & \langle X_{k,k},\bbJ\rangle 
        \end{pmatrix} F\right)_{i,j} \\
        &= \left(F^T \pi\pi^T F\right)_{i,j} \cdot n^2 \pm \|F\|^2\sqrt k \delta' n^2\\
        &= \left(e_1 e_1^T\right)_{i,j} \cdot n^2 \pm \|F\|^2\sqrt k \delta' n^2
    \end{align*}
    where $e_1$ is the first standard basis vector. The final line comes since $\pi$, being the left eigenvector associated to $\lambda_1 = 1$, is (up to scaling) the first row of $F^{-1}$.
\end{proof}

With \pref{lem:sdp-affine-check} in hand, the remainder of the proof follows from \pref{prop:poly-implies-infeasibility} and \pref{prop:exists-a-poly} in the previous section. In particular, each block $\check X_{i,i}$ for $i = 2,...k$ is a feasible solution to $SPS(m,\lambda_i)$ SDP with error tolerance $\delta'' > \|F\|^2\sqrt k \delta$. We showed already that when $\lambda^2(d-1) > 1$, and for small enough error tolerance and large enough $m$, this SDP is w.h.p. infeasible on input $\bG\sim\Null$. Thus we need simply to make $\delta$ small enough so that $\delta''$ is below the minimum tolerance in Proposition \ref{prop:poly-implies-infeasibility}.

% \begin{remark}  \label{rem:primal-slack}
%     We can choose constants $a$ and $b$ such that the $LocalStatistic$ SDP from \pref{def:lost-formal} is infeasible on $\bG$ drawn from $\Null$ if we set the distance from the Kesten-Stigum bound $\epsilon$ and global error tolerance $\delta$ as $(\eps, \delta) = (a, b)$, and also if we choose these as $(\eps, \delta) = (a, 2b)$. In particular, this means that when $\delta = b$, for any PSD matrix $X$ with an all-ones diagonal, there is a polynomial $f$ such that the constraint
%     \[
%         \langle f(A_{\bG}), X\rangle = \|q_s\|^2_{\km}\lambda^s n \pm \delta n
%     \]
%     is violated by a margin of $\Omega(n)$.
% \end{remark}

% (end)

\subsection{Spectral Distinguishing} % (fold)

Our argument in the previous section can be recast to prove \pref{cor:spectral-distinguishing}, namely that above the Kesten-Stigum threshold the spectrum of the adjacency matrix can also be used to distinguish the null and planted distributions. 

Let $(\bG,\bx) \sim \calP_{d,k,M,\pi}$, and write $\bX \triangleq \bx \bx^T$, and 
$$
 \check \bX = (F^T \otimes \1)\bX (F \otimes \1) = (F^T \bx)(F^T\bx)^T \triangleq \check \bx \check \bx^T.
$$ 
Think of $\check \bX$ as a block matrix $(\bX_{i,j})_{i,j \in [k]}$, as we did $X$ in the previous section, and $\check \bx$ as a block vector $(\check \bx_i)_{i\in [k]}$. Applying \pref{thm:local-stats} and repeating the calculations in \pref{lem:sdp-affine} and \pref{lem:sdp-affine-check} \textit{mutatis mutandis} with $\bX$ instead of $X$, we can show that w.h.p.
$$
    \langle \check \bX_{i,j}, \nb{s}{\bG} \rangle =  \lambda_i \|q_s\|^2_{\km} n + o(n) \qquad \text{ if $i=j$}
$$
and zero otherwise, for every $s = O(1)$ and
$$
    \langle \check \bX_{1,1}, \bbJ \rangle = \begin{cases} n^2 & i=j=1 \\ 0 & \text{else} \end{cases},
$$
with strict equality following from the rigidity of the group sizes in the planted model. Because $\nb{s}{\bG} =\1$, we know 
$$
    \check \bx_i^T \check \bx_j = \langle \check \bX_{i,j}, \1 \rangle = 0
$$
when $i\neq j$. In other words, the $k$ vectors $\check \bx_1,...,\check \bx_k$ are orthogonal.

We can show that $A_{\bG}$ has an eigenvalue with a separation $\eta > 0$ from the bulk spectrum by proving 
$$
\check\bx_i^T f(A_{\bG})\check\bx_i = \langle \check \bX_{i,i}, f(A_{\bG}) \rangle < 0
$$ 
for some polynomial $f(x)$ positive on of $(-2\sqrt{d-1}- \eta,2\sqrt{d-1} + \eta)$. As long as $(d-1)\lambda_i^2 > 1$, the same polynomial from \pref{prop:exists-a-poly} works here. As the $\check \bx_i$ are orthogonal, we get one distinct eigenvalue outside the bulk for each eigenvalue of $T$ satisfying this property.

\begin{remark}
    To distinguish the null model from the planted one using the spectrum of $A_{\bG}$, simply return \textsc{planted} if $A_{\bG}$ has a single eigenvalue other than $d$ whose magnitude is bigger than $2\sqrt{d-1} + \delta$ for any error tolerance $\delta$ you choose, and \textsc{null} otherwise. Unfortunately, this distinguishing algorithm is not robust to adversarial edge insertions and deletions. For instance, given a graph $\bG \sim \Null$, the adversary can create a disjoint copy of $K_{d+1}$, the complete graph on $d+1$ vertices, whose eigenvalues are all $\pm d$. The spectrum of the perturbed graph is the disjoint union of $\pm d$ and the eigenvalues of the other component(s), so the algorithm will be fooled.  We will show in \pref{sec:robust-drbm} that the Local Statistics SDP is robust to this kind of perturbation.
\end{remark}

% (end)

\subsection{Lower Bounds} % (fold)

In this section, we prove the second half of \pref{thm:main-drbm}, which gives a complementary lower bound: if every one of $\lambda_2,...,\lambda_k$ has modulus at most $1/\sqrt{d-1}$ there exists some feasible solution to the Local Path Statistics SDP for every $m \ge 1$. We can specify a pseudoexpectation completely by way of an $(nk + 1)\times(nk+1)$ positive semidefinite matrix
$$
   \begin{pmatrix} 1 & \pseudo x^T \\ \pseudo x & \pseudo x^Tx \end{pmatrix} 
    \triangleq \begin{pmatrix} 1 & l^T \\ l & X \end{pmatrix}.
$$
After first writing down the general properties required of \textit{any} quadratic pseudoexpectation satisfying $\calB_k$, we'll show that in order for $\pseudo$ to match every moment asked of it by the $\LS(2,m)$ SDP, it suffices for it to satisfy
$$
    \pseudo p_{H,S,\tau}(x,G) \approx \expected p_{H,S,\tau}(\bG,\bx)
$$
when $(H,S,\tau)$ is a path of length $0,...,m$ with labelled endpoints, or a pair of disjoint, labelled vertices. Finally, we'll construct a pseudoexpectation matching these path moments out of feasible solutions to the symmetric path statistics SDP from the previous section.

\begin{lemma} \label{lem:bk-description}
    The set of $\calB_k$-satisfying pseudoexpectations is parameterized by pairs $(X,l) \in \bbR^{nk\times nk}\times \bbR^{nk}$ for which
    \begin{align}
        \begin{pmatrix} 1 & l^T \\ l& X \end{pmatrix} &\succeq 0 \label{eq:bk1} \\
        \diag(X) &= l \label{eq:bk2} \\
        \Tr X_{u,u} &= e^T l = 1 & &\forall u\in[n]  \label{eq:bk3} \\
        X_{u,v}e &= l_u & &\forall u,v\in[n] \label{eq:bk4}
    \end{align}
\end{lemma}

\begin{proof}
    Recall that the set $\calB_k$ is defined by the polynomial equations
    \begin{align*}
    	\textit{Boolean} & & x_{u,i}^2 &= x_{u,i} & &\forall u\in[n] \text{ and } i\in [k] \\
    	\textit{Single Color} & & \sum_i x_{u,i} &= 1 & &\forall u\in[n] 
    \end{align*}
    That a degree-two pseudoexpectation \textit{satisfies} these constraints means 
    \begin{align*}
        \pseudo p(x)x_{u,i}^2 &= \pseudo p(x) x_{u,i} & &\forall p \,\st\,\deg p = 0 \\
        \pseudo p(x)\sum_i x_{u,i} &= \pseudo p(x) & &\forall p \,\st\,\deg p \le 1. 
    \end{align*}
   Writing $X = \pseudo x^T x$ and $l = \pseudo x$ as above, the first constraint is equivalent to $l = \diag(X)$, since the degree-zero polynomials are just constants, and we can guarantee that the second holds for every polynomial of degree at most one by requiring it on $p = 1$ and $p= ,x_{v,j}$ for all $v$ and $j$. The Lemma is simply a concise packaging of these facts, using the block notation $X = (X_{u,v})_{u,v\in[n]}$ and $l = (l_u)_{u\in [n]}$.
\end{proof}

\begin{proposition} \label{prop:paths-suffice}
    Let $\bG \sim \Null$, and let the pair $(X,l) \in \bbR^{nk\times nk} \times \bbR^{nk}$ satisfies \eqref{eq:bk1}-\eqref{eq:bk4} and
    \begin{align*}
        \langle e, l_i \rangle &= \pi(i)n \pm \delta n \\
        \langle X_{i,j}, \nb{s}{\bG} \rangle &= \pi(i)T_{i,j}^s n \pm \delta n \\
        \langle X_{i,j}, \bbJ \rangle &= \pi(i)\pi(j)n^2 \pm \delta n^2,
    \end{align*}
    then with high probability the degree-two pseudoexpectation that they induce is a feasible solution to the $\LS(2,m)$ SDP with any error tolerance $\delta' > \delta$.
\end{proposition}

We will defer the proof of \pref{prop:paths-suffice} to \pref{app:loc-stat-calc}. Its conclusion in hand, we can now set about constructing a pseudoexpectation. Since $(d-1)\lambda_2^2 < 1$, the $SPS(\lambda_i,m)$ SDP is feasible for every error tolerance $\delta' > 0$. Thus for each $i = 2,...,k$ there exists a feasible solution in the form of a PSD matrix $Y(\lambda_i)$ satisfying
\begin{align*}
    Y(\lambda_i)_{u,u} &= 1 & & \forall u \in [n] \\
    \langle Y(\lambda_i), \nb{s}{\bG} \rangle &= \lambda_i^s\|q_s\|^2_{\km} n \pm \delta' n & & \forall s \in [m] \\
    \langle Y(\lambda_i), \bbJ \rangle &= 0 \pm \delta' n^2.
\end{align*}
Now, define $\check X$ to be the $k\times k$ block diagonal matrix
$$
    \check X = \begin{pmatrix} 
        \bbJ & & & \\
           &Y(\lambda_2)& & \\
           & &\ddots & \\
           & & & Y(\lambda_k),
    \end{pmatrix}
$$
i.e. $\check X_{i,j} = 0$ when $i\neq j$, and the diagonal blocks are as above, and similarly let $\check l = (e,0,...,0)^T$. We claim that the pair $X = (F^{-T}\otimes \1) \check X (F^{-1} \otimes \1)$ and $l = (F^{-1}\otimes \1)\check l$ satisfies the conditions of \pref{lem:bk-description} and \pref{prop:paths-suffice}.

First
$$
    \begin{pmatrix} 1 & l^T \\ l & X \end{pmatrix}
    = \begin{pmatrix} 1 & \\ & F^{-T}\otimes \1 \end{pmatrix} \begin{pmatrix} 1 & \check l^T \\ \check l & \check X \end{pmatrix} \begin{pmatrix} 1 & \\ & F^{-1}\otimes \1 \end{pmatrix} \succeq 0
$$
by taking a Schur complement. Since $\pi$ is the first row of $F^{-1}$, we know $l_i = \pi(i)e$ for each $i = \in [k]$. Moreover, since $X$ is obtained by changing basis block-wise, the diagonal of $X$ depends only on the diagonals of $\bbJ$ and the $Y(\lambda_i)$, all of which are all ones, so
\begin{align*}
    \diag X 
    &= \diag\big((F^{-T}\otimes \1)\Diag(\diag (\check X)) (F^{-1}\otimes \1)\big) \\
    &=\diag\big((F^{-T}\otimes \1) \1 (F^{-1}\otimes \1)\big) \\
    &= \diag\left(F^{-T}F^{-1} \otimes \1\right) \\
    &= \diag\left(\Diag\pi \otimes \1\right) \\
    &= (\pi(1)e,...,\pi(k)e) = l
\end{align*}
as desired. Similarly, because $\check X$ is block diagonal when regarded as $k\times k$ matrix of $n\times n$ blocks, if we treat it instead as an $n\times n$ matrix of $k\times k$ blocks $\check X_{u,v}$, then $\check X_{u,u} = \1$ for every $u \in [n]$, and
$$
    \Tr X_{u,u} = \Tr F^{-T}\check X_{u,u} F^{-1} = \Tr F^{-T}F^{-1} = \Tr \Diag \pi = 1.
$$
Finally, the top row of each $\check X_{u,v}$ is the vector $e_1^T$, so
$$
    X_{u,v}e = F^{-T}\check X_{u,v}F^{-1}e = F^{-T}\check X_{u,v}e_1 = F^{-T}e_1 = \pi = l_u.
$$

It remains to verify the affine conditions in \pref{prop:paths-suffice}. As in the proof of \pref{lem:sdp-affine-check}, since each $Y(\lambda_i)$ is a feasible solution to the $SPS(\lambda_i,m)$ SDP with error tolerance $\delta'$,
\begin{align*}
    \langle X_{i,j}, \nb{s}{\bG}\rangle
    &= \left(F^{-T}\begin{pmatrix} 
        \langle \bbJ, \nb{s}{\bG} \rangle & & & \\
        & \langle Y(\lambda_2), \nb{s}{\bG} \rangle & & \\
        & & \ddots & \\
        & & & \langle Y(\lambda_k), \nb{s}{\bG} \rangle 
        \end{pmatrix} F^{-1}\right)_{i,j} \\
    &= \left(F^{-T}\Lambda^s F^{-1}\right)_{i,j} \cdot \|q_s\|^2_{\km}n \pm \|F^{-1}\|^2 \sqrt k \delta' n \\
    &= \left(\Diag(\pi)T^s\right)_{i,j}\cdot \|q_s\|^2_{\km}n \pm \|F^{-1}\|^2 \sqrt k \delta' n \\
    &= \pi(i) T^s_{i,j}\cdot \|q_s\|^2_{\km}n \pm \|F^{-1}\|^2 \sqrt k \delta' n
\end{align*}
and
\begin{align*}
    \langle X_{i,j}, \bbJ \rangle 
    &= \left(F^{-T}\begin{pmatrix} 
        \langle \bbJ, \bbJ \rangle & & & \\
        & \langle Y(\lambda_2), \bbJ \rangle & & \\
        & & \ddots & \\
        & & & \langle Y(\lambda_k), \bbJ \rangle 
        \end{pmatrix} F^{-1}\right)_{i,j} \\
    &= \left(F^{-T} e_1 e_1^T F^{-1}\right)_{i,j} \cdot n^2 \pm \|F^{-1}\|^2\sqrt k \delta' n^2 \\
    &=  \left(\pi\pi^T\right)_{i,j} \cdot n^2 \pm \|F^{-1}\|^2\sqrt k \delta' n^2 \\
    &= \pi(i)\pi(j) \cdot n^2 \pm \|F^{-1}\|^2\sqrt k \delta' n^2,
\end{align*}
and by setting $\delta'$ sufficiently small, we can make each of these errors at most any $\delta > 0$ of our choosing.

% \begin{remark} \label{rem:dual-slack}
%     By correctly setting the $\epsilon$ from \pref{prop:poly-exists-lb} and the global error tolerance $\delta$, we can once again choose $(\epsilon,\delta) = (a,b)$ with the property that whenever the $\LS(2,m)$ SDP is feasible $\delta = b$, it is feasible as well with $\delta = 2b$. Thus every one of the $\simeq$ constraints---i.e. those that depend on the observed graph $\bG$---are satisfied with slack $\Omega(n)$.
% \end{remark}

% (end)

\subsection{Robustness} \label{sec:robust-drbm} % (fold)

The proof of robustness largely reduces to the discussion in \pref{sec:col-robustness}. Recall that we need to produce a $\rho > 0$ for which (i) when $\bG \sim \Planted$, or $\bG \sim \Null$ with $(d-1)\lambda_2^2 < 1$, the SDP with high probability remains feasible for any error tolerance $\delta$, even after perturbing $\rho n$ edges, and (ii) that when $\bG \sim \Null$ and $(d-1) > \lambda_2^2$, if the SDP is infeasible at tolerance $\delta$, it remains so at some  tolerance $\delta' < \delta$ even after perturbing $\rho n$ edges.

For (i), assume that the SDP is feasible at error tolerance $\delta$ on input $\bG$. We build the SDP as a linear combination of solutions $Y(\lambda_i)$ to the $SPS(m,\lambda_i)$ SDP, which we argued in \pref{sec:col-robustness} is robust in the desired sense. For (ii), when $\bG \sim \Null$ and $(d-1)\lambda_2^2 > 1$, we reduced infeasibility of the $\LS(2,m)$ SDP to that of the $SPS(m,\lambda_2)$ SDP, which we showed already is infeasible. Moreover, from \pref{sec:col-robustness}, the latter remains infeasible after a sufficiently small perturbation.

%% file: Content/sbm.tex
\section{The Stochastic Block Model}    \label{sec:SBM} % (fold)

We turn, finally, to the proof of \pref{thm:main-sbm} concerning the local statistics algorithm and Stochastic Block Model. For the sake of exposition, as we did for the DRBM, we will first write down a simpler SDP that can robustly
solve the detection problem above the KS threshold, and then show that feasibility of the full SDP implies feasibility of this simpler one.

Throughout this section, let $\Planted$ denote the SBM with fixed parameters $(d,k,M,\pi)$, and $\Null = \calG(n,d/n)$. Recall that to sample a pair $(\bG,\bx)$ from the planted model, we first choose a partition $V_1(\bG) \sqcup \cdots \sqcup V_k(\bG) = [n]$ by placing each vertex in group $V_i$ with probability $\pi(i)$, setting $\bx_{u,i}$ equal to $1$ if $u \in V_i$; it will be convenient to write $\bm{\sigma}: [n] \to [k]$ for this random labelling map. Then, we include each edge $(u,v) \in E(\bG)$ with probability $M_{\bm{\sigma}(u),\bm{\sigma}(v)} d/n$, setting $\bG_{u,v} = 1$ in this event.

Now, for any graph $G$, define the matrices $\cnb{s}{G}$ as follows. For each walk in the complete graph $K_n$, write $\gamma : u \to v$ if it begins with $u$ and ends with $v$, let $w_{G}(\gamma) = \prod_{e \in \gamma}(G_{e} - d/n)$, and set
\begin{equation}
    \left(\cnb{s}{G}\right)_{u,v} = \sum_{\gamma : u \to v, |\gamma| = s} w_G(\gamma).
\end{equation}
When $(\bG,\bx) \sim \Planted$, define as in \pref{sec:local-stat-treat} and \pref{sec:distinguish-drbm} an $nk \times nk$ matrix $\bX \triangleq \bx \bx^\ast$. As before, we will at times think of $\bX$ as an $n\times n$ matrix of $k\times k$ blocks $\bX_{u,v}$, and at others a $k\times k$ matrix of $n\times n$ blocks $\bX_{i,j}$. 

\begin{claim}
    Let $\overline T = T - e^T\pi$. Then
    $$    
        \expected \langle \bX_{i,j}, \cnb{s}{\bG} \rangle = \pi(i)\overline T^s_{i,j} \cdot d^s n + O(1),
    $$
    and this inner product enjoys concentration of $O(\sqrt n)$.
\end{claim}
\begin{proof}
    Let $\gamma$ be some length-$s$ self-avoiding walk in the complete graph; WLOG we can label its vertices with the set $[s+1]$. We need to calculate the expectation of $w_{\bG}(\gamma)$ on the event that its endpoints are labelled $i$ and $j$:
    \begin{align*}
        \expected[w_{\bG}(\gamma),\text{labels $i$ and $j$}]
        &= \sum_{\eta : [s+1], \eta(1) = i, \eta(s+1) = j} \pi(i) \prod_{t \in [s]}(M_{\eta(t),\eta(t+1)} - 1)(d/n)\pi(\eta(t+1)) \\
        &= \pi(i) (T - e^T\pi)^s_{i,j} \cdot (d/n)^s.
    \end{align*}
    There are $n(n-1)(n-2)\cdots(n-s)$ length-$s$ self avoiding walks in the complete graph, so already the total expectation of $w_{\bG}(\gamma)$ among these walks accounts for the quantity in the claim. On the other hand, those $\gamma$'s which contain a cycle contribute negligiblly: there are $O(n^v)$ $\gamma$'s with $v$ vertices, and for each one $\expected w_{\bG}(\gamma) = O(n^{-e})$, so the total contribution of $\gamma$'s with $e \ge v$ is at best $O(1)$. 
\end{proof}
 
This motivates the following SDP:

\begin{definition}  \label{def:lvl-m-ps-sdp}
    For each $m \ge 1$, the \textit{level $m$ path statistics SDP} with error tolerance $\delta > 0$ is the feasibility problem
    \begin{align}
        \text{Find $X = (X_{i,j}) \succeq 0$ s.t.   } 
        (X_{i,i})_{u,u} &\le 1 & &\forall u\in[n],i\in[k] \\
        \Tr (X_{u,u}) &= 1 & & \forall u \in [n] \\
        \langle X_{i,j}, \cnb{s}{G}\rangle &= \pi(i)\overline T^s_{i,j} \cdot d^s n \pm \delta n & & \forall i,j\in[k],s = 0,...,m.
    \end{align}
\end{definition}

\begin{theorem} \label{thm:local-path-stats}
    When $\lambda_2^2d > 1$, there exists $m = O(1)$ and $\delta>0$ for which the level $m$ path statistics SDP can solve the detection problem. Conversely, when $\lambda_2^2 d < 1$, no such $m$ and $\delta$ exist.
\end{theorem}

Recall that $\Lambda$ is a $k\times k$ diagonal matrix containing the eigenvalues of $T$, sorted in descending order of modulus from the upper left corner. Since $e$ and $\pi^\ast$ are right and left eigenvectors, respectively, of $T$, $\overline T$ commutes with $T$ and satisfies $\overline T F = F\overline \Lambda$, where $\overline \Lambda$ is obtained from $\Lambda$ by deleting the upper left entry (which in our setup is equal to $1$). We will accordingly take the same change-of-basis approach as in the DRBM. For any feasible solution $X$ to this SDP, we can form an analogous matrix $\check X \triangleq (F^T \otimes \1)X(F\otimes \1)$, with blocks $\check X_{i,j}$. Following \pref{lem:sdp-affine-check}, observe that
\begin{align*}
    \langle \check X_{i,j}, \cnb{s}{\bG} \rangle 
    = \begin{pmatrix} \langle \check X_{1,1}, \cnb{s}{G} \rangle & \cdots & \langle \check X_{1,k}, \cnb{s}{G} \rangle \\ \vdots & \ddots & \vdots \\ \langle \check X_{k,1}, \cnb{s}{G} \rangle & \cdots & \langle \check X_{k,k}, \cnb{s}{G} \rangle \end{pmatrix}_{i,j} 
    &= \left(F^T\begin{pmatrix} \langle X_{1,1}, \cnb{s}{G} \rangle & \cdots & \langle  X_{1,k}, \cnb{s}{G} \rangle \\ \vdots & \ddots & \vdots \\ \langle X_{k,1}, \cnb{s}{G} \rangle & \cdots & \langle X_{k,k}, \cnb{s}{G} \rangle \end{pmatrix}F\right)_{i,j} \\
    &= \left(F^T \Diag(\pi)(T - e\pi^\ast)^s F\right)_{i,j} \cdot d^s n \pm \|F\|^2\sqrt k \delta n \\
    &= \left(F^T \Diag(\pi)F\overline\Lambda^s\right)_{i,j} \cdot d^s n \pm \|F\|^2\sqrt k \delta n  \\
    &= \overline \Lambda^s_{i,j} \cdot d^s n \pm \|F\|^2\sqrt k \delta n.
\end{align*}

Moreover, the diagonal entries of $\check X$ are bounded by a constant dependant only on $M$ and $\pi$, which we can check by considering the blocks $\check X_{u,u} = F^T X_{u,u} F$. Since $\Tr X_{u,u} = 1$, the maximal diagonal entry of $\check X_{u,u}$ is at most $\Tr F^T X_{u,u} F \le \|F\|^2$.

Our observations about the matrices $\check X_{i,i}$ from \pref{sec:drbm} carry over here---namely each of these is PSD with ones on the diagonal. Thus we have shown that if the level-$m$ Path Statistics SDP is feasible, then for some constant $C$,
$$
    \sup_{Y \succeq 0, \Tr Y = n, Y_{i,i} \le C} |\langle Y, \cnb{m}{G} \rangle| \ge |d\lambda_2|^s \cdot n - O(\delta)n,
$$
(where the constant in the $O(\delta)$ may be taken as the quantity $\|F\|^2\sqrt k \delta$ above). In particular this is true when $\bG \sim \Planted$. On the other hand, we will prove the following upper bound on this quantity when $\bG$ is drawn from the null model.

\begin{theorem}     \label{thm:sdp-val-m-walk}
    Let $\bG \sim \Null$. Then for any $\epsilon,C > 0$ there exists $m \in \bbN$ so that with high probability
    $$
        \sup_{Y \succeq 0, \Tr Y = n, Y_{i,i} \le C} |\langle Y, \cnb{m}{G} \rangle| \le ((1 + \epsilon)d)^{m/2} n.
    $$
\end{theorem}

\noindent This, and the preceding dicsussion, prove one half of \pref{thm:local-path-stats}, namely that whenever $\lambda_2^2d > 1$, there are some $m = O(1)$ and $\delta > 0$ for which the level $m$ Path Statistics SDP is with high probability infeasible on inout $\bG \sim \null$, but feasible on input $\bG \sim \Planted$. We will prove the other half in \pref{sec:sbm-lb}. \pref{thm:main-sbm}, the analogous statement to \pref{thm:local-path-stats} for the full local statistics algorithm, follows from a final observation:

\begin{observation}
    With high probability over $\bG \sim \Null$, if the level-$m$ Path Statistics SDP is infeasible at error tolerance $\delta$, then the $\LS(2,m)$ SDP at some error tolerance $\delta'(\delta)$ is infeasible as well.
\end{observation}

\begin{proof}
    The quadratic block of the $\LS(2,m)$ SDP concerns $nk\times nk$ matrices, and includes all hard constraints---bounds on diagonal entries, trace of diagonal blocks---present in the Path Statistics SDP. Moreover, it has access to affine constraints involving the counts of subgraphs with at most $m$ edges. Since the entries of $\cnb{s}{G}$ for $s \le m$ are simply linear combinations of such counts, $\LS(2,m)$ has access to the affine constraints from the Local Path Statistics SDP as well.
\end{proof}

The promised robustness guarantee in \pref{thm:main-sbm} can be achieved by choosing $B$ as per \pref{thm:sbmrobustness}, deleting edges incident to all vertices of degree $> B$, and inputting the resulting graph into the $\LS(2,m)$ SDP.

\subsection{Local Statistics in the SBM}

We pause to compute the local statistics of the SBM; by setting $k=1$ and $M = 1$, we recover analogous results for the \ER model. Recall that for a partially subgraph $(H,S,\tau)$,
$$
    L_{(H,S,\tau)}(M,\pi) \triangleq \sum_{\hat \tau : \hat \tau|_S = \tau} \prod_{v \in v(H)} \pi(\hat\tau(v)) \prod_{(u,v) \in E(H)} M_{\hat\tau(u),\hat\tau(v)}.
$$

\begin{theorem} \label{thm:loc-stat}
    Let $(H,S,\tau)$ be a partially labelled graph with $O(1)$ edges and $\ell$ connected components. Then with high probability
    \begin{align*}
        p_{(H,S,\tau)}(\bG,\bx) 
        = n^{\chi(H)} L_{(H,S,\tau)}(M,\pi) \cdot d^{|E(H)|} n^{|V(H)| - |E(H)|} + o(n^{c(H)})
    \end{align*}
\end{theorem}

\begin{proof}
    Fix $(H,S,\tau)$. There are
    $$
        {n \choose |V(H)|}|V(H)|! = n^{|V(H)|} + O(n^{|V(H)| - 1})
    $$
    injective maps from $V(H) \inj [n]$. The probability that each is an occurrence, once we condition on the labels $\sigma$ of the relevant vertices, is given by $\prod_{(u,v) \in E(H)} M_{\sigma(u),\sigma(v)} d/n$. The probability of each labelling $\sigma$ is $\prod_{u \in V(H)} \pi(\sigma(u))$, and we only consider labellings that agree with $\tau$ at the relevant vertices. Thus
    $$
        \expected p_{(H,S,\tau)}(\bx,\bG) = n^{\chi(H)}L_{(H,S,\tau)}(M,\pi) + O(n^{\chi(H) - 1}).
    $$
    and one immediately sees that this expectation decomposes as a product of analogous expectations over the connected components of $H$.
    
    To prove concentration, in the case when $H$ has at least one cycle, $c(H) > \chi(H)$ and the assertion follows from Markov. Otherwise, let us consider $\expected p_{H,S,\tau,U,W}(\bG,\bx)^2$. This is a sum over pairs of maps $\phi,\psi : V(H) \to [n]$, and as $H$ is acyclic, it is dominated by terms in which the images of these two maps are disjoint. Thus, to leading order, this variance is equal to the expected number of occurrences of two disjoint copies of $(H,S,\tau)$, which we just observed is $(\expected p_{H,S,\tau,U,W}(\bG,\bx))^2$ to leading order. We finish by using Chebyshev and observing that $\chi(H) = c(H)$.
\end{proof}

\subsection{Proof of \pref{thm:sdp-val-m-walk}} \label{sec:proof-sdp-val-m-walk}

The main challenge in studying $\cnb{m}{\bG}$, when $\bG$ is a sparse \ER random graph, is the presence of of certain localized combinatorial structures which inflate the number of non-backtracking walks: high-degree vertices and small subgraphs with many cycles. Our strategy will be to decompose $\cnb{s}{\bG}$ as a sum of two matrices, one of which ``avoids'' these structures and admits spectral norm bounds, and the other of which has a small entrywise $L_1$ norm. Let us make this precise. In any graph $G$, write $\sfB_t(v,G)$ for the set of vertices with distance at most $t$ from $v$; call $v$ $(t,\eps)$\textit{-heavy} if $|\sfB_t(v,G)| \ge (1+\epsilon)^td^t$. We will call a vertex $v$ \textit{$(t,r,\eps)$-vexing} if either it participates in a cycle of length less than $r$ or it is $(t,\eps)$-heavy. Let $\bH$ be the subgraph obtained by deleting every vexing vertex, and write
$$
    \left(\cnb{m}{\bH}\right)_{u,v} = \sum_{\gamma : u \to v, |\gamma| = s, \gamma \in V(\bH)} w_{\bG}(\gamma).
$$
We will also refer to $\bH$ as the \emph{$(t,r,\eps)$-truncation of $G$}.  In the sequel, we assume $r = \Theta\left(\frac{\log n}{(\log \log n)^2}\right)$.  Then \pref{thm:sdp-val-m-walk} is an immediate consequence of the following two results.

\begin{theorem}[Truncated Spectral Norm Bound]  \label{thm:spec-norm-bound}
    For every $\epsilon > 0$, there exist $t,m$ satisfying $m = t^3$ so that with high probability
    $$
        \|\cnb{m}{\bH}\| \le ((1 + \epsilon)d)^{m/2}.
    $$
\end{theorem}

\begin{proposition}[$L_1$ Bound]    \label{prop:vex-doesnt-vex}
    For every $\delta > 0$, and every $r = O(1)$, for any $t \ge \Omega(\tfrac{\log m - \log \delta}{\log (1 + \epsilon)})$ so that with high probability
    $$
        \|\cnb{m}{\bG} - \cnb{m}{\bH}\|_1 \le \delta n.
    $$
\end{proposition}

\noindent With these two results in hand, \pref{thm:sdp-val-m-walk} quickly follows: any matrix $Y \succeq 0$ with unit diagonal satisfies $|Y_{u,v}| \le 1$, so
\begin{align*}
    |\langle Y, \cnb{m}{\bG}\rangle| 
    &\le |\langle Y, \cnb{m}{\bH}\rangle| + |\langle Y, \cnb{m}{\bG}- \cnb{m}{\bH}\rangle| \\
    &\le n\|\cnb{m}{\bH}\| + \|\cnb{m}{\bG} - \cnb{m}{\bH}\|_1 \\
    &\le \left(((1 + \epsilon)d)^{m/2} + \delta \right)n.
\end{align*}

\noindent \pref{thm:spec-norm-bound} is the heavier technical lift, so we will warm up with the proof of \pref{prop:vex-doesnt-vex}.  The proof of \pref{thm:spec-norm-bound} is deferred to \pref{sec:trace-method}.

\subsubsection{Proof of \pref{prop:vex-doesnt-vex}}    
    For a non-backtracking walk $\gamma$ on the complete graph, write $\calV(\gamma)$ for the event that $\gamma$ visits a vexing vertex. Then
    $$
        \|\cnb{m}{\bG} - \cnb{m}{\bH}\|_1 \le \sum_{\gamma \in K_n, |\gamma| = m} |w_{\bA}(\gamma)| \indicator{\calV(\gamma)} \le \sum_{\gamma \in K_n, |\gamma| = m} (d/n)^{\text{\# non-edges}}\indicator{\calV(\gamma)}.
    $$
    Once we choose $\bG$, $\gamma$ alternates between segments of edges on $\bG$ and segments of non-edges. We do not lose too much by relaxing slightly the condition that $\gamma$ is non-backtracking, instead asking only that it is non-backtracking whenever it walks on $\bG$. 
    
    Let us define an \textit{$m$-scribble} $\frs$ with type $(p_1|q_1|\cdots|p_l|q_l)$ on the complete graph to be a path comprised of $l$ non-backtracking segments of lengths $p_1,...,p_l$ interspersed with $l$ `free' segments of lengths $q_1,...,q_l$. We require that $\sum p_i + \sum q_i = m$, and all but perhaps $p_1,q_l$ are strictly positive. Define $w(\frs) = (d/n)^{\sum q_i}$, and let us write $\frs \subset \bG$ to mean that every non-backtracking segment of $\frs$ appears in $\bG$. We will call a scribble \textit{vexing} and write $\calV(\frs)$, if any of the vertices of $\frs$ is vexing. In view of the preceding paragraph, it suffices to bound
    $$
        \sum_{\frs \in K_n, |\frs| = m} w(\frs) \indicator{\frs \subset \bG}\indicator{\calV(\frs)}.
    $$
    We will divide the event $\calV(\frs)$ that $\frs$ is vexing into two subcases: write $\calH(\frs)$ if $\frs$ contains a heavy vertex, and $\calC(\frs)$ if it ever encounters a vertex on a cycle of length at most $r$. 
    
    We will need the following simplified version of the forthcoming \pref{lem:typical-nbd-sample}.
    
    \begin{lemma}   \label{lem:heavy-bound}
        Let $\Gamma \subset K_n$, and write $\Gamma \subset \bG$ to mean that every edge of $\Gamma$ appears in $\bG$. Then there exist universal $C,c$ so that
        $$
            \bbP[\Gamma \subset \bG \text{ and contains a $(t,\eps)$-heavy vertex}] \le \bbP[\Gamma \subset \bG] \cdot |V(\Gamma)| \cdot C \exp\left(-\frac{c}{1 + |V(\Gamma)|}(1 + \eps)^t\right) 
        $$
    \end{lemma}
    
    \begin{proof}
        Write $\bG^c$ for the graph obtained by removing every one of $\Gamma$'s edges, and write $\sfB_t(v,\bG^c)$ for the $t$-neighborhood of a vertex in this modified graph. We claim that if $\Gamma \subset \bG$ and one of its vertices is $(t,\eps)$-heavy, then one of its vertices is $(t,\eps')$-heavy in $\bG^c$, where $\eps' = (1 + \eps)(1 + |V(\Gamma)|)^{-1/t} - 1$. Assume that $v$ is the heavy vertex in $\bG$, noting that
        $$
            |B_t(v,G)| \le |B_t(v,\bG^c)| + |B_t(V(\Gamma),\bG^c)|
        $$
        by dividing the shortest paths of length $t$ emanating from $v$ according to whether they use edges from $\Gamma$ or not. Since $v$ is $(t,\eps)$-heavy, the left hand side is at least $(1 + \epsilon)^td^t$. If for some $\eps'$ no other vertex in $\Gamma$ is $(t,\eps')$-heavy in $\bG^c$, then $|B_t(V(\Gamma),\bG^c)| \le |V(\Gamma)|(1 + \eps')^td^t$, and we conclude that $|B_t(v,\bG^c)| \ge (1 + \eps)^t(1 - \eps^t|V(\Gamma)|)$, which is a contradiction if $\eps'$ is set as in the theorem statement. 
        
        Thus we have shown that the event we care about is contained in the intersection of two independent ones: that $\Gamma \subset \bG$, and that there exists a vertex in $\Gamma$ that is $(t,\epsilon')$-heavy in $\bG^c$. We can bound this second probability by taking a union bound over all vertices in $\Gamma$, and noting that the probability of being heavy in $\bG^c$ is at most the probability of being heavy in $\bG$. From \pref{lem:typical-nbd-sample}, the probability that a given vertex is $(t,\epsilon')$-heavy in $\bG$ is, for some universal $C,c$, at most $C\exp(-c(1 + \epsilon')^t) \le C\exp(-\tfrac{c}{|V(\Gamma)| + 1}(1 + \epsilon)^t)$. We then execute the union bound and assemble everything.
    \end{proof}
    
    With this lemma in hand, and using the fact that $\frs$ contains at most $m+1$ vertices,
    $$
        \bbP[\frs \subset \bG, \calH_N(\frs)] \le \bbP[\frs \subset \bG] C(m+1)\exp\left(-\frac{c}{m + 2}(1 + \epsilon)^t\right) \triangleq \bbP[\frs \subset \bG] \Upsilon(m,\eps).
    $$
    Thus
    $$
        \expected \sum_{\frs \in K_n,|\frs| = m} w(\frs)\indicator{\frs \subset\bG}\indicator{\calH(\frs)} \le \Upsilon(m,\eps) \expected\sum_{\frs \in K_n|\frs| = m}w(\frs)\indicator{\frs \subset \bG}.
    $$
    
    We need to perform a similar calculation for the scribbles which visit a vertex on a cycle. Fixing $\frs$, if $\frs$ itself contains a cycle, then $\bbP[\frs\subset \bG, \calC(\frs)] = \bb[\frs\subset \bG]$. Otherwise, $\frs$ does not contain a cycle, and there must be a path of length at most $r$, using no edges in $\frs$, that connects two of its vertices. This event is independent of the event $\frs \subset \bG$; for any fixed length $s$, there are at most $n^{s-1}$ such paths, and each occurs with probability $O(n^{-s})$, meaning that the total probability is bounded by $O(r/n)$. 
    
    Combining all of this,
    \begin{align*}
        \expected \sum_{\frs \in K_n,|\frs| = m} w(\frs) \indicator{\frs \subset \bG}\indicator{\calV(\frs)}
        &\le \left(\Upsilon(m,\eps) + O(r/n)\right)\expected\sum_{\frs \in K_n,|\frs| = m} w(\frs)\indicator{\frs \subset \bG} \\
        &+ \expected\sum_{\frs \in K_n,|\frs| = m}w(\frs)\indicator{\frs \subset \bG}\indicator{\frs \text{ contains a cycle}} \\
        &+ \text{lower order terms}.
    \end{align*}
    To compute term in the second line, fix a scribble of type $(p_1|q_2|\cdots|q_l)$. To choose a scribble with this type in $\bG$, one needs to select a subgraph in $\bG$ with at most $l$ connected components, at least one of which contains a cycle of length at most $m$. In expectation there are $o(n^{l-1})$ of these. For each of $q_1,...,q_{l-1}$, there are $q_i - 1$ choices of a free vertex, and we pay a weight of $O(n^{-q_i})$; for $q_l$, if it is nonzero, there are $q_l$ free vertices at a cost of $O(n^{-q_l})$. Thus the final term, the expected, weighted counts of scribbles that contain a cycle, contributes $o(n)$.
    
    It therefore remains only to compute the expected weighted sum of all $m$-scribbles in $\bG$. Analogous to the previous paragraph, to choose a scribble of type $(p_1|q_1|\cdots|p_l|q_l)$ in $\bG$, one first selects a tuple of non-backtracking walks in $\bG$ with lengths $p_1,...,p_l$, and then connects them with free segments. In expectation there are $d^{p_1 + \cdots p_l}n^l + O(n^{l-1})$ such tuples of walks in $\bG$. For each of $q_1,...,q_{l-1}$, there are $q_i - 1$ choices of a free vertex, and we pay a weight $(d/n)^{q_i}$; for $q_l$, if it exists, there are $q_l$ free vertices at a cost of $(d/n)^{q_l} = d^{q_l}$. There are at most $2^{m+1}$ types, giving
    $$
        \expected \sum_{\frs \in K_n, |\frs| = m} w(\frs)\indicator{\frs \subset \bG} \le 2(2d)^m n + O(1).
    $$
    
    Having computed its expectation, we now need to show that the number of vexing scribbles is concentrated. We begin by recalling the well-known Efron-Stein inequality.
    
    \begin{lemma}
            Let $(Y,X_1,X_2,...,X_T)$ be i.i.d. real random variables. Then for any function $f : \bbR^T \to \bbR$,
            $$
                \Var f(X_1,...,X_T) \le \frac{1}{2}\sum_{u\in [T]}\bbE\left[(f(X_1,...,X_u,...,X_T) - f(X_1,...,Y,...,X_T))^2\right].
            $$
    \end{lemma}
    
    We will apply this to the function $f$ that counts weighted, vexing scribbles. Let $\bG$ be an ER random graph, and $\tilde\bG$ be the same graph with some edge re-randomized. With probability $1 - 2d/n$, the graphs $\bG = \tilde\bG$ and the weighted scribble counts are the same. With the remaining probability, we are comparing the weighted, vexing scribble counts on two graphs that differ at an edge. Since the addition of an edge can only make more vertices vexing, the count can only increase; thus we can clumsily bound the difference by the total number of scribbles (vexing or not) that use the added edge. 
    
    \begin{fact}
        Let $\bG \sim \Null$. Then the probability that there is a vertex with degree larger than $\Delta d$ is at most $n(e/\Delta)^{d\Delta}$. In particular, setting $\Delta = 2\log n$, this probability is $o(n^{-c})$ for every $c$.
    \end{fact}
    
    In a graph with maximum degree $\Delta$, the weighted sum of $m$-scribbles with the property that at least one of the non-backtracking segments uses a given edge is upper bounded by $m(2\Delta)^m$, so let us split into the events that the maximum degree in $\bG$ is less than vs. greater than $\log n$. On the first event, whose probability we will upper bound by $1$, we get $m(2\log n)^m$ weighted scribbles. The second event gives us at most $m(2n)^m$ weighted scribbles, has probability better than any inverse polynomial in $n$. Thus by Efron-Stein inequality the variance of the number of weighted scribbles is at most
    $$
        2(d/n)\cdot {n\choose 2}\left(m^2(2\log n)^{2m} + o(1)\right) = O(n\log^{2m}n),
    $$
    so we get concentration of $O(\sqrt{n}\log^m n)$.
    
    All told, then, we have that with high probability
    $$
        \|\cnb{m}{\bG} - \cnb{m}{\bH}\|_1 \le (m+1)C\exp\left(-\frac{c}{m+2}(1 + \epsilon)^t\right) \cdot 2(2d)^m n + O(\sqrt{n}\log^{m+1} n).
    $$
    To make this smaller than $\delta n + o(n)$, it suffices to set $t = \Omega(\tfrac{\log m - \log\delta}{\log(1 + \epsilon)})$.

% (end)

%% Trace Method

\input{Content/trace}

% (end)

%% file: Content/trace.tex
\subsection{Spectral norm bounds} \label{sec:trace-method}
In this section, we prove \pref{thm:spec-norm-bound}.
\subsubsection{Setup}
\paragraph{Choosing parameters.}
Let $\eps$ and $d$ be constants given to us.  With the privilege of hindsight, we choose a small constant $\delta < \frac{1}{100\sqrt{(1+\eps)d}}$; $t$ to be a large enough integer (depending on $\eps,d$ and $\delta$) so that:
\begin{enumerate}
\item the hypothesis of \pref{thm:expectation-estimate} holds on parameters $d'\coloneqq (1+\eps)d, d$ and $\delta$,
\item $((1+\eps)d)^{1/t^2}t^{30/t^3} < 1+\eps$,
\item $\delta^{-24/t}<1+\eps$,
% \item $t^3 \le \frac{kt^3}{\ln^3(kt^3)}$;
\end{enumerate}
$\ell \coloneqq t^3$; $k$ is any even integer in $\left[\frac{\log n\log\log n}{2\ell},\frac{4\log n\log\log n}{\ell}\right]$; and $r\coloneqq\frac{k\ell}{\ln^3 (k\ell)}$.  Observe that since $t$ is constant, $r = O\left(\frac{\log n}{(\log\log n)^2}\right)$.

Let $\bG$ be an Erd\H{o}s-Renyi $G(n,d/n)$ graph, let $\bS$ its the set of $(t,r,\eps)$-vexing vertices, and let $\bG_{t,r,\eps}$ be the $(t,r,\eps)$-truncation of $\bG$.  Let $\bA$ be the adjacency matrix of $\bG_{t,\eps,r}$.  Define
\[
    \left(\bA-\frac{d}{n}1_{[n]\setminus \bS}1^{\top}_{[n]\setminus\bS}\right)^{(\ell)}[u,v] = \sum_{\substack{W~\text{length-$\ell$ nonbacktracking walk}\\\text{from $u$ to $v$ in $K_{[n]\setminus S}$}}} \prod_{ij\in W} \left(\bA-\frac{d}{n}11^{\top}\right)[i,j]
\]
We are interested in obtaining bounds on the spectral norm of $\left(\bA-\frac{d}{n}1_{[n]\setminus \bS}1^{\transp}_{[n]\setminus \bS}\right)^{(\ell)}$, and towards doing so we employ the trace method.  In particular, we prove:
\begin{theorem}  \label{thm:main-spec-norm-bound}
    With probability $1-n^{-100}$, $\displaystyle\left\|\locA\right\|\le \left((1+\eps)^4\sqrt{d}\right)^\ell$.
\end{theorem}

We will obtain spectral norm bounds on $\locA$ that hold with high probability by achieving high probability bounds on
\[
    \bF \coloneqq \Tr\left(\left(\locA\right)^{2k}\right).
\]
When $\bF$ is bounded by $R$,
\[
    \left\|\locA\right\| \le R^{\frac{1}{2k}}.
\]

\subsubsection{From High Trace to Counting}
We borrow some more terminology from \cite{MOP19b}:
\begin{definition}[Linkages]
    We call a closed walk of length $k\ell$ on $K_n$ a \emph{$(k\times\ell)$-linkage} if it can be split into $k$ segments each of length-$\ell$ such that $W$ is nonbacktracking on each segment.  We refer to each such length-$\ell$ nonbacktracking segment as a \emph{link}. We use $V(W)$ to denote the vertices visited by $W$ and $E(W)$ to denote the (undirected) edges visited by $W$.
\end{definition}
Within a linkage $W$, we use $a_{ij}(W)$ to denote the number of times the \emph{undirected} edge $\{i,j\}$ is walked on (which in this exposition we will simply abbreviate to $a_{ij}$), $S(W)$ to denote the set of \emph{singleton} edges in $E(W)$, i.e. all edges $\{i,j\}$ such that $a_{ij}=1$, and $D(W)$ to denote all the remaining edges (each of which has $a_{ij}\ge 2$), which we call \emph{duplicative edges}.  We use $e(W)$ to denote the ``excess'' number of edges in $W$, i.e., $e(W) = |E(W)|- |V(W)|-1$  Finally, let $\calE(W)$ denote the event that $V(W)\cap\bS$ is empty.  We will call a subset of edges $E'$ \emph{good} if there are no $(t,r,\eps)$-vertices in the graph induced by $E'$.  We have,
\begin{align*}
    \bF =& \sum_{W~\text{is $(k\times\ell)$-linkage of $K_n$}} \prod_{ij\in W} \left(\bA[i,j]-\frac{d}{n}\right)\cdot\Ind[\calE(W)] \\
    =& \sum_{W~\text{is $(k\times\ell)$-linkage of $K_n$}} \prod_{ij\in S(W)}\left(\bA[i,j]-\frac{d}{n}\right)\prod_{ij\in D(W)}\left(\bA[i,j]-\frac{d}{n}\right)^{a_{ij}(W)}\Ind[\calE(W)] \numberthis \label{eq:trace-basic}
\end{align*}
We write one of the terms in the above expression in a more convenient form:
\begin{align*}
    \left(\bA[i,j]-\frac{d}{n}\right)^{a_{ij}(W)} &= \sum_{t=0}^{a_{ij}(W)}\bA[i,j]^t\left(-\frac{d}{n}\right)^{a_{ij}(W)-t}\cdot{a_{ij}(W)\choose t} \\
    &= \bA[i,j]\sum_{i=1}^{a_{ij}(W)}\left(-\frac{d}{n}\right)^{a_{ij}(W)-t}\cdot{a_{ij}(W)\choose t} + \left(-\frac{d}{n}\right)^{a_{ij}(W)}\\
    &= \bA[i,j]\left(\left(1-\frac{d}{n}\right)^{a_{ij}(W)}-\left(-\frac{d}{n}\right)^{a_{ij}(W)}\right) + \left(-\frac{d}{n}\right)^{a_{ij}(W)}.
\end{align*}
Writing $\gamma_{ij} = \left(1-\frac{d}{n}\right)^{a_{ij}(W)}-\left(-\frac{d}{n}\right)^{a_{ij}(W)}$, we can rewrite \pref{eq:trace-basic} as
\[
    \sum_{W~\text{is $(k\times\ell)$-linkage of $K_n$}} \prod_{ij\in S(W)} \left(\bA[i,j]-\frac{d}{n}\right)\prod_{ij\in D(W)}\left(\bA[i,j]\gamma_{ij}+\left(-\frac{d}{n}\right)^{a_{ij}(W)}\right)\Ind[\calE(W)]
\]
In the subsequent steps we will use $\displaystyle \sum_{W}$ as short for $\displaystyle \sum_{W~\text{is $(k\times\ell)$-linkage of $K_n$}}$.  Thus, we have:
\begin{align*}
    \bF &= \sum_W \prod_{ij\in S(W)}\left(A[i,j]-\frac{d}{n}\right)\sum_{L\subseteq D(W)}\prod_{ij\in L}A[i,j]\gamma_{ij}\prod_{ij\notin L}\left(-\frac{d}{n}\right)^{a_{ij}(W)}\Ind[\calE(W)]
\end{align*}
where $ij\notin L$ actually means $ij\in D(W)\setminus L$.
We are interested in bounding $|\E[\bF]|$.  We first point out that $\gamma_{ij}\le 1$ for large enough $n$.  Then:
\begin{align*}
    |\E[\bF]| &= \left|\E\left[\sum_W \prod_{ij\in S(W)}\left(A[i,j]-\frac{d}{n}\right)\sum_{L\subseteq D(W)}\prod_{ij\in L}A[i,j]\gamma_{ij}\prod_{ij\notin L}\left(-\frac{d}{n}\right)^{a_{ij}(W)}\Ind[\calE(W)]\right]\right|\\
    &= \left|\sum_W \sum_{L\subseteq D(W)}\prod_{ij\in L}\gamma_{ij}\prod_{ij\notin L}\left(-\frac{d}{n}\right)^{a_{ij}(W)}\E\left[\prod_{ij\in S(W)}\left(A[i,j]-\frac{d}{n}\right)\prod_{ij\in L}A[i,j]\Ind[\calE(W)]\right]\right| \\
    &\le \sum_W \sum_{L\subseteq D(W)} \prod_{ij\in L}\gamma_{ij}\prod_{ij\notin L}\left(\frac{d}{n}\right)^{a_{ij}(W)}\left|\E\left[\prod_{ij\in S(W)}\left(A[i,j]-\frac{d}{n}\right)\prod_{ij\in L}A[i,j]\Ind[\calE(W)]\right]\right|\\
    \intertext{By \pref{thm:expectation-estimate}:}
    &\le\sum_W \sum_{\substack{L\subseteq D(W)\\ L~\text{good}}} \prod_{ij\notin L} \left(\frac{d}{n}\right)^{a_{ij}(W)}\cdot C\log^2 n\cdot\left(\frac{d}{n}\right)^{|S(W)\cup L|}\cdot n^{.8e(W)}\cdot 4^{|S(W)|}\cdot\delta^{|S(W)|-24kt} \\
    &\le \sum_W \sum_{\substack{L\subseteq D(W)\\ L~\text{good}}} \prod_{ij\notin L} \left(\frac{d}{n}\right)^{a_{ij}(W)-1}\cdot C\log^2 n\cdot\left(\frac{d}{n}\right)^{|S(W)\cup L|}\cdot\left(\frac{d}{n}\right)^{|D(W)|-L}\cdot \\ &\hspace{100px}n^{.8e(W)}\cdot 4^{|S(W)|}\cdot\delta^{|S(W)|-24kt}\\
    &\le C(n)\sum_W\left(\sum_{\substack{L\subseteq D(W)\\ L~\text{good}}}\prod_{ij\notin L}\left(\frac{d}{n}\right)^{a_{ij}(W)-1}\right)\cdot\left(\frac{d}{n}\right)^{|S(W)|+|D(W)|}\cdot n^{.8e(W)}\cdot 4^{|S(W)|}\cdot\delta^{|S(W)|-24kt}\\
    \numberthis \label{eq:expec-F}
\end{align*}
where $C(n) = C\log^2 n$.  Now we analyze
\[
    \sum_{\substack{L\subseteq D(W)\\ L~\text{good}}}\prod_{ij\notin L}\left(\frac{d}{n}\right)^{a_{ij}(W)-1}.
\]
Call the \emph{weight} of a subset $L$ of $D(W)$ as $w(L) \coloneqq \sum_{ij\in L}(a_{ij}(W)-1)$.  Let $D^*(W)$ be a maximum weight good subset of $D(W)$, and define $\Delta(W)$ as $w(D(W))-w(D^*(W))$.  We say $\Delta(W)$ is the number of \emph{profligate steps} in the graph.  Then:
\begin{align*}
    \sum_{\substack{L\subseteq D(W)\\ L~\text{good}}}\prod_{ij\notin L}\left(\frac{d}{n}\right)^{a_{ij}(W)-1} &= \sum_{\substack{L\subseteq W\\ L~\text{good}}}\left(\frac{d}{n}\right)^{w(D(W))-w(L)}\\
    \intertext{Since $a_{ij}(W)$ for every edge is at least $2$, we can bound the above by:}
    &\le \sum_{L\subseteq W} \left(\frac{d}{n}\right)^{\max\{|D(W)|-|L|,~\Delta(W)\}}\\
    &= \sum_{\eta\le\Delta(W)}\left(\frac{d}{n}\right)^{\Delta(W)}\cdot{|D(W)|\choose\eta} + \sum_{\eta > \Delta(W)}\left(\frac{d}{n}\right)^{\eta}\cdot{|D(W)|\choose\eta} \\
    &\le (\Delta(W)+1)\left(\frac{d|D(W)|}{n}\right)^{\Delta(W)} + \sum_{\eta>\Delta(W)}\left(\frac{d|D(W)|}{n}\right)^{\eta} \\
    &\le (\Delta(W)+2)\left(\frac{d|D(W)|}{n}\right)^{\Delta(W)}\\
    &\le 2\left(\frac{2d|D(W)|}{n}\right)^{\Delta(W)}.
\end{align*}
Plugging the above back into \pref{eq:expec-F} and ``absorbing'' a factor of $2$ into $C(n)$ tells us:
\begin{align*}
    \pref{eq:expec-F} \le C(n)\sum_W \left(\frac{2d|D(W)|}{n}\right)^{\Delta(W)}\cdot\left(\frac{d}{n}\right)^{|S(W)|+|D(W)|}\cdot n^{.8e(W)}\cdot 4^{|S(W)|}\cdot\delta^{|S(W)|-24kt}. \numberthis\label{eq:exp-parked-1}
\end{align*}
We will split the above sum based on properties of the walk (such as $|S(W)|, e(W), \Delta(W), |V(W)|$) and count the number of terms in each split part using an encoding argument.  Before we get into the counting argument, we make a key definition:
\begin{definition}
    We say a step from $u$ to $v$ in a linkage $W$ is \emph{fresh} if $v$ was never visited earlier in $W$.  We will use $f(W)$ to denote the number of fresh steps in $W$.
\end{definition}

\begin{remark}  \label{rem:vertex-count-fresh}
    For a linkage $W$, $|V(W)| = f(W) + 1$.
\end{remark}
A consequence of \pref{rem:vertex-count-fresh} along with the fact that $|S(W)|+|D(W)| = |E(W)|$, we get that $|E(W)| = f(W) + e(W)$.  Thus, \pref{eq:exp-parked-1} is bounded by
\[
    C(n)\sum_W\left(\frac{2d|D(W)|}{n}\right)^{\Delta(W)}\cdot\left(\frac{d}{n}\right)^{e(W)+f(W)}\cdot n^{.8e(W)}\cdot 4^{|S(W)|}\cdot\delta^{|S(W)|-24kt}. \numberthis\label{eq:exp-parked-2}
\]
To complete the proof, we need the following, which is proved in \pref{sec:counting-walks}:
\begin{theorem} \label{thm:main-count}
    The total number of $(k,\ell)$-linkages with $f$ fresh edges, $e$ excess edges, $s$ singleton edges and $\Delta$ profligate steps is at most:
    \begin{align*}
        n^{f+1}\cdot(4\lambda(W))^{7\lambda(W)+1}\cdot(k\ell)^{3\lambda(W)+1}\cdot(\ell+1)^{6k}\cdot
        ((1+\eps)d)^{tk+k\ell/2-|D(W)|-s/2}
    \end{align*}
    where $\lambda(W)\le 3e + \frac{12k\ell\ln(k\ell)}{r} + 3\Delta$.
\end{theorem}
Now recall that we wished to obtain bounds on the following from \pref{eq:exp-parked-2}:
\[
    Q \coloneqq C(n)\sum_W\left(\frac{2d|D(W)|}{n}\right)^{\Delta(W)}\cdot\left(\frac{d}{n}\right)^{e(W)+f(W)}\cdot n^{.8e(W)}\cdot 4^{|S(W)|}\cdot\delta^{|S(W)|-24kt}.
\]
From \pref{thm:main-count}:
\begin{align*}
    Q &\le C(n)\sum_{f,e,\Delta,s\ge 0} \left(\frac{2d|D(W)|}{n}\right)^{\Delta}\cdot\left(\frac{d}{n}\right)^{f}\cdot\left(\frac{d}{n^{.2}}\right)^e\cdot 4^s\cdot\delta^{s-24kt}\cdot\\
    &n^{f+1}\cdot(4\lambda(W))^{7\lambda(W)+1}\cdot(k\ell)^{3\lambda(W)+1}\cdot(\ell+1)^{6k}\cdot((1+\eps)d)^{tk+k\ell/2-|D(W)|-s/2}\\
    &\le C(n)\cdot n\sum_{f,e,\Delta,s\ge 0} \left(\frac{2d|D(W)|}{n}\right)^{\Delta}\cdot ((1+\eps)d)^{f+k\ell/2-|D(W)|-s/2}\cdot\left(\frac{d}{n^{.2}}\right)^e \cdot (4\delta)^s\cdot\delta^{-24kt}\cdot\\
    &(4\lambda(W)k\ell)^{7\lambda(W)+1}\cdot((\ell+1)^6((1+\eps)d)^t)^k\\
    \intertext{Defining $C'(n)\coloneqq C(n)\cdot n$ and the fact $f=s+|D(W)|-e$, we get:}
    &\le C'(n)\sum_{f,e,\Delta,s\ge 0}\left(\frac{2d|D(W)|}{n}\right)^{\Delta}\cdot((1+\eps)d)^{s/2-e+k\ell/2}\cdot\left(\frac{d}{n^{.2}}\right)^e \cdot (4\delta)^s\cdot\delta^{-24kt}\cdot\\
    &(4\lambda(W)k\ell)^{7\lambda(W)+1}\cdot((\ell+1)^6((1+\eps)d)^t)^k \\
    &\le C'(n)\cdot ((1+\eps)d)^{k\ell/2}\sum_{f,e,\Delta\ge 0} \left(\frac{2d|D(W)|}{n}\right)^{\Delta}\cdot\left(\frac{1}{n^{.2}}\right)^e\cdot \\
    &\delta^{-24kt}\cdot (4\lambda(W)k\ell)^{7\lambda(W)+1}\cdot((\ell+1)^6((1+\eps)d)^t)^k\sum_{s\ge 0} \left(4\delta\sqrt{(1+\eps)d}\right)^s\\
    \intertext{where the inequality above is true since no other term depends on $s$.  By our choice of $\delta$, the summation over $s$ is bounded by $2$.}
    &\le 2C'(n)\cdot((1+\eps)d)^{k\ell/2} \cdot((\ell+1)^6((1+\eps)d)^t)^k\cdot \delta^{-24kt} \\
    &\sum_{f,e,\Delta\ge 0}\left(\frac{2d|D(W)|}{n}\right)^{\Delta}\cdot\left(\frac{1}{n^{.2}}\right)^e
    \cdot (4\lambda(W)k\ell)^{7\lambda(W)+1}\\
    \intertext{By noting that $4\lambda(W)k\ell\le \poly(k,\ell)$:}
    &\le 2C'(n)\cdot((1+\eps)d)^{k\ell/2} \cdot((\ell+1)^6((1+\eps)d)^t)^k\cdot \delta^{-24kt} \\
    &\sum_{f,e,\Delta}\left(\frac{2d|D(W)|\cdot\poly(k,\ell)}{n}\right)^{\Delta}\cdot\left(\frac{\poly(k,\ell)}{n^{.2}}\right)^e\cdot\poly(k,\ell)^{84k\ell\ln(k\ell)/r}\\
    &\le 8C'(n)\cdot ((1+\eps)d)^{k\ell/2} \cdot((\ell+1)^6((1+\eps)d)^t)^k\cdot \delta^{-24kt} \cdot \poly(k,\ell)^{84\ln^6(k\ell)}\cdot(k\ell)
\end{align*}
We know that $\bF$ is a nonnegative random variable since it is the trace of an even power of a Hermitian matrix, i.e., it is the trace of a positive semidefinite matrix.  Hence, by Markov's inequality, we know that except with probability $n^{-100}$, the random variable $\bF$ defined in \pref{eq:trace-basic} is bounded by
\[
    n^{100}\cdot 8C'(n)\cdot ((1+\eps)d)^{k\ell/2} \cdot((\ell+1)^6((1+\eps)d)^t)^k\cdot \delta^{-24kt} \cdot \poly(k,\ell)^{84\ln^6(k\ell)}\cdot(k\ell)    \numberthis \label{eq:high-prob-bound}
\]
By our choice of parameters, the $k\ell$-th root of the above is bounded by:
\[
    (1+\eps)^4\sqrt{d}.
\]
In particular, this means the $k$-th root of \pref{eq:high-prob-bound} is bounded by $\left((1+\eps)^4\sqrt{d}\right)^\ell$ with probability $1-n^{-100}$.

In summary, we have shown that whp:
\[
    \Tr\left(\left(\locA\right)^{k}\right)^{1/k} \le \left((1+\eps)^4\sqrt{d}\right)^\ell
    \numberthis \label{eq:trace-power-root-bound}
\]
thereby establishing:
\begin{theorem}[Restatement of \pref{thm:main-spec-norm-bound}]
    With probability $1-n^{-100}$:
    \[
        \displaystyle\left\|\locA\right\|\le \left((1+\eps)^4\sqrt{d}\right)^\ell.
    \]
\end{theorem}

\subsection{Counting Walks} \label{sec:counting-walks}
This section is dedicated to proving \pref{thm:main-count}.  Let $W$ be a $(k\times\ell)$-linkage with $f$ fresh steps, $e$ excess edges, $s$ singleton edges, $\Delta$ profligate steps.  In this section, we will give an efficient encoding of $W$ which will help us upper bound the number of such linkages.

We use $G(W)$, defined to have vertex set $V(W)$ and edge set $E(W)$, to denote the graph of the linkage $W$.  Further, each edge $\{i,j\}$ has a weight $a_{ij}$, which is the number of times edge $\{i,j\}$ is walked on in $W$.  We can write $E(W)$ as the disjoint union $S(W)\cup D(W)$ where $S(W)$ is the set of singleton edges  and $D(W)$ is the set of duplicative edges.  Let $D^*(W)$ denote a maximum weight subset of $D(W)$ such that no vertices in the graph induced by those edges on vertices in $W$ are $(t,r,\eps)$-vexing within $W$.
\begin{remark}
    For any set of edges $E'$, we will use $V(E')$ to denote the set of endpoints of edges in $E'$.
\end{remark}
\begin{remark}
$|D(W)| = |D^*(W)|+\Delta$ and $|E(W)|=|D^*(W)|+\Delta+s$.
\end{remark}
\begin{remark}
    For any subset of edges $E'\subseteq E(W)$, the total number of times $W$ walks on an edge in $E'$ is given by
    \[
        \sum_{ij\in E'}a_{ij}.
    \]
    A special case of the above is that when $E'=E(W)$, the sum above is equal to $k\ell$.
\end{remark}
Our next step is to prove:
\begin{claim}   \label{claim:small-spanning-forest}
    There is a spanning forest $F$ of $D^*(W)\cup S(W)$ such that:
    \begin{enumerate}
    \item $\displaystyle \sum_{ij\in D^*(W)\setminus F}a_{ij}\le \frac{4k\ell\ln(k\ell)}{r}$.
    \item $|S(W)\setminus F|\le e$.
    \end{enumerate}
\end{claim}
In service of proving the above claim we will need the following standard fact which can be found in \cite[Theorem 13.21]{KV12}:
\begin{fact}    \label{fact:convex-hull-sf-LP}
    Let $P$ be the polytope in $\R^{D^*(W)}$ given by the convex hull of indicator vectors of spanning forests of $D^*(W)$.  $P$ is also the feasible region of the following linear program:
    \begin{align*}
        x&\in\R^{D^*(W)}\\
        x&\ge 0\\
        \sum_{ij\in R}x_{ij} &\le |V(R)|-1 &\forall R\subseteq D^*(W). \numberthis \label{eq:sf-LP}
    \end{align*}
\end{fact}
We additionally also state the following from \cite[Corollary 2.18]{MOP19} which is a consequence of the ``irregular Moore bound'' of \cite{AHL02}:
\begin{fact}    \label{fact:moore-bound}
    Let $H$ be a graph with $v\ge 3$ vertices and girth $g\ge 20\ln v$.  Then $|E(H)|-v\le \frac{2v\ln v}{g}$.
\end{fact}

\begin{proof}[Proof of \pref{claim:small-spanning-forest}]
    Consider the following assignment to the variables of LP \pref{eq:sf-LP}:
    \[
        \wt{x}_{ij} = 1-\frac{4\ln(k\ell)}{r}.
    \]
    We will first show that this assignment is indeed feasible for the LP.  Recall that we need to show that for every $R\subseteq D^*(W)$, $\sum_{ij\in R}\wt{x}_{ij}\le|V(R)|-1$.  The LHS of this expression is simply $|R|\left(1-\frac{4\ln(k\ell)}{r}\right)$ so it suffices to prove:
    \begin{align*}
        |R|\left(1-\frac{4\ln(k\ell)}{r}\right) \le |V(R)|-1 & & \forall R\subseteq D^*(W).
    \end{align*}
    \paragraph{Case 1: $|V(R)|<r$.} In this case $R$ is a forest as there $R$ has no cycles of length smaller than $r$.  Since $R$ is a forest $|R|\le |V(R)|-1$ and hence the above inequality we wish to prove is definitely true.
    \paragraph{Case 2: $|V(R)|\ge r$:} Since the girth of $(V(R),R)$ is at least $r$, by \pref{fact:moore-bound},
    \begin{align*}
        |R| &\le |V(R)|\left(1+\frac{2\ln|V(R)|}{r}\right) \\
        \frac{|R|}{1+\frac{2\ln|V(R)|}{r}} &\le |V(R)|\\
        |R|-|R|\frac{2\ln|V(R)|}{r} - 1 &\le |V(R)| - 1\\
        |R|-|R|\frac{4\ln(k\ell)}{r} &\le |V(R)|-1.
    \end{align*}
    \\
    If we augment the linear program \pref{eq:sf-LP} with the objective function
    \[
        \max \sum_{ij\in D^*(W)} a_{ij}x_{ij},
    \]
    by \pref{fact:convex-hull-sf-LP} we know that there is a spanning forest $\wt{F}$ such that maximum of the above objective is achieved at the indicator vector of $\wt{F}$.  Since we showed $\wt{x}$ is feasible,
    \[
        \sum_{ij\in \wt{F}} a_{ij} \ge \sum_{ij\in D^*(W)} a_{ij}\wt{x}_{ij}.
    \]
    Subtracting both sides of the above inequality from $\sum_{ij\in D^*(W)} a_{ij}$ yields
    \[
        \sum_{ij\in D^*(W)\setminus \wt{F}} a_{ij} \le \sum_{ij\in D^*(W)}a_{ij}(1-\wt{x}_{ij}) \le \frac{4k\ell\ln(k\ell)}{r}.
    \]
    Now we extend $\wt{F}$ to a spanning forest $F$ of $D^*(W)\cup S(W)$.  Initially, we set $F=\wt{F}$ and process edges of $S(W)$ sequentially in an arbitrary order $i_1j_1,\dots,i_sj_s$.  We add $i_tj_t$ to $F$ if its addition does not create a cycle and ``reject'' it otherwise.  The number of rejected edges is bounded by $e$ and hence $F$ must contain at least $s-e$ edges of $S$.  Thus, $|S(W)\setminus F|\le e$.  Furthermore, since $D^*(W)\setminus F = D^*(W)\setminus \wt{F}$,
    \[
        \sum_{ij\in D^*(W)\setminus F} a_{ij} \le \frac{4k\ell\ln(k\ell)}{r}.
    \]
\end{proof}
\pref{claim:small-spanning-forest} lends itself to a natural decomposition of the edges of $G(W)$ into \emph{forest edges}, which we denote $F(W)$, and \emph{crossing edges}, which we denote $C(W)$.
\begin{remark}  \label{rem:cross-decomp}
    $C(W)$ can be written as the following natural disjoint union of sets:
    \[
        C(W) = (S(W)\setminus F(W))\cup (D^*(W)\setminus F(W))\cup (D(W)\setminus D^*(W)).
    \]
\end{remark}
\noindent At a high level, the linkage $W$ breaks into stretches of steps on $F(W)$ between steps on $C(W)$; a large chunk of this section is dedicated to showing how to encode the portions of $W$ on forest edges highly efficiently.

Let's now express the linkage $W$ in terms of the sequence of vertices walked on: in particular $W = w_0w_1w_2\dots w_{k\ell}$.
\begin{definition}
    We call each consecutive pair $w_iw_{i+1}$ a \emph{step}.  If the edge $\{w_i,w_{i+1}\}$ is a crossing edge, we call the step $w_iw_{i+1}$ a \emph{crossing step}, and a \emph{forest step} otherwise.  We call a maximal contiguous sequence of forest steps a \emph{cruise}.
\end{definition}
\begin{remark}
    Any $(k\times\ell)$-linkage $W$ can be expressed as
    \[
        W = C_1s_1C_2s_2\dots C_{\gamma(W)}s_{\gamma(W)}C_{\gamma(W)+1}
    \]
    where each $C_i$ is a (possibly empty) cruise, each $s_i$ is a crossing step, and $\gamma(W)$ is the number of crossing steps in $W$.
\end{remark}
Next, we wish to bound $\gamma(W)$.
\begin{claim}   \label{claim:bound-gamma}
    $\gamma(W)\le e + \frac{4k\ell\ln(k\ell)}{r} + \Delta$.
\end{claim}
\begin{proof}
By \pref{rem:cross-decomp}:
\begin{align*}
    \gamma(W) &= \sum_{ij\in S(W)\setminus F(W)}a_{ij} + \sum_{ij\in D^*(W)\setminus F(W)}a_{ij} + \sum_{ij\in D(W)\setminus D^*(W)}a_{ij}\\
    &\le \sum_{ij\in S(W)\setminus F(W)}1 + \frac{4k\ell\ln(k\ell)}{r} + \Delta &\text{(by \pref{claim:small-spanning-forest})}\\
    &\le e + \frac{4k\ell\ln(k\ell)}{r} + \Delta. &\text{(by \pref{claim:small-spanning-forest})}
\end{align*}
\end{proof}

\begin{definition}
    We refer to endpoints of edges in $C(W)$ as well as the start/end vertex of $W$\footnote{which are the same since $W$ is closed} as \emph{terminal vertices}.  We use $T(W)$ to refer to the set of terminal vertices of $G(W)$.
\end{definition}

\begin{remark}  \label{rem:terminal-bound}
    $|T(W)|\le 2|C(W)|+1 \le 2\gamma(W)+1$.  We will use $\lambda(W)$ to refer to $2\gamma(W)+1$.
\end{remark}
\begin{remark}
    Each cruise starts and ends at terminal vertices.
\end{remark}

\begin{definition}[Skeleton forest]
    We use $\skel(F(W))$ to refer to the subforest of $F(W)$ given by the union of paths in $F(W)$ connecting terminal vertices.  Formally,
    \[
        \skel(F(W)) \coloneqq \bigcup_{\substack{P~\text{path in $F(W)$}\\ \text{endpoints of $P$ in $T(W)$}}} P.
    \]
\end{definition}

\begin{observation} \label{obs:few-leaves-skel}
    Every leaf in $\skel(F(W))$ is a terminal vertex and hence by \pref{rem:terminal-bound} the number of leaves in $\skel(F(W))$ is at most $\lambda(W)$.
\end{observation}
\paragraph{Goal 1: Encoding $\skel(F(W))$.} Our first goal is to find an efficient encoding of $\skel(F(W))$.  Towards this goal, we first prove the following.
\begin{lemma}   \label{lem:encode-forest}
    Let $\Gamma_{L,v}$ be the set of forests with at most $L$ leaves on vertex set $\{1,\dots,v\}$.  There is a subset $Q\subseteq\Gamma_{L,v}$ of at most $(4Lv)^{2L+1}$ forests such that any forest in $\Gamma_{L,v}$ is isomorphic to a forest in $Q$.
\end{lemma}
We will need the following classical graph theory fact called Cayley's formula; a reader can find multiple proofs in \cite{Cas06}:
\begin{fact}    \label{fact:Cayley}
    The number of labeled spanning trees on $v$ vertices is $v^{v-2}$.
\end{fact}
We will also need the following fact about trees:
\begin{fact}    \label{fact:tree-deg}
    Let $T$ be a tree where $v_i$ is the number of vertices of degree $i$.  Then the following are true:
    \begin{align*}
        3(v_1-2) &\ge \sum_{i\ge 3} iv_i\\
        v_1-2 &\ge \sum_{i\ge 3} v_i
    \end{align*}
\end{fact}
\begin{proof}
    Using the fact that the sum of degrees in a tree is $2|V(T)|-2$ we have:
    \begin{align*}
        2\sum_{i\ge 1} v_i - 2 &= \sum_{i\ge 1} iv_i\\
        v_1 + \sum_{i\ge 3} v_i - 2 &= \sum_{i\ge 3} iv_i\\
        v_1 - 2 &= \sum_{i\ge 3} (i-2)v_i   \numberthis \label{eq:tree-degree}
    \end{align*}
    Lower bounding $i-2$ by $1$ in the RHS of \pref{eq:tree-degree}, it follows that:
    \[
        v_1 - 2 \ge \sum_{i\ge 3} v_i   \numberthis \label{eq:conseq1}
    \]
    Adding $2\cdot$\pref{eq:conseq1} and \pref{eq:tree-degree} gives us:
    \[
        3(v_1-2) \ge \sum_{i\ge 3} iv_i.
    \]
\end{proof}

\begin{proof}[Proof of \pref{lem:encode-forest}]
    Let $\Xi$ be any tree in $\Gamma_{L,v}$.  If we split $V(\Xi)$ into leaves $V_1(\Xi)$, degree-$2$ vertices $V_2(\Xi)$, and degree-$\ge 3$ vertices $V_{\ge 3}(\Xi)$, we have the following from \pref{fact:tree-deg}:
    \begin{align*}
        |V_1(\Xi)| - 2 &\ge |V_{\ge 3}(\Xi)|.
    \end{align*}
    Thus, $|V_{\ge 3}(\Xi)|\le L-2$.  Let $\wt{\Xi}$ be the weighted tree described in the following way:
    \begin{displayquote}
        Its vertex set is $[|V_1(\Xi)\cup V_{\ge3}(\Xi)|]$.  Let $\pi$ be an arbitrary bijection from $[|V_1(\Xi)\cup V_{\ge3}(\Xi)|]$ to $V_1(\Xi)\cup V_{\ge3}(\Xi)$.  Place an edge between vertices $i$ and $j$ if there is a path between $\pi(i)$ and $\pi(j)$ such that all vertices in between are in $V_2(\Xi)$.  The weight of an edge $ij$ in $\wt{\Xi}$ is the distance between $i$ and $j$ in $\Xi$.
    \end{displayquote}
    Observe that $\wt{\Xi}$ has $\wt{v}\le 2L$ vertices and the weight of an edge is an integer between $1$ and $|V(\Xi)|$.  By Cayley's formula (\pref{fact:Cayley}) the number of labeled spanning trees on $\wt{v}$ vertices is at most $(\wt{v})^{\wt{v}-2}$.  Consequently the number of spanning forests on $\wt{v}$ vertices is at most $(2\wt{v})^{\wt{v}}$ (since every spanning tree on $\wt{v}$ vertices has $2^{\wt{v}-1}$ subforests).  Since $\wt{v}\le 2L$ and there are at most $2L$ possibilities for $\wt{v}$, each labeled spanning forest on vertex set $[\wt{v}]$ that can be encoded by a number in $\left[(4L)^{2L+1}\right]$.  In particular this gives us a way to encode the edge set of any $\wt{\Xi}$ by a number in $\left[(4L)^{2L+1}\right]$.

    All the weights of the edges can be encoded by a number in $\left[|V(\Xi)|^{2L}\right]$, and consequently we can encode $\wt{\Xi}$ by a number in $\left[(4L|V(\Xi)|)^{2L+1}\right]$.  It is possible to reconstruct a forest isomorphic to $\Xi$ from $\wt{\Xi}$ and hence our proof is complete.
\end{proof}

\begin{lemma}   \label{lem:skeleton-encoding}
    $\skel(F(W))$ can be encoded by a number in $\left[(4\lambda(W)k\ell)^{2\lambda(W)+1}\cdot n^{|V(\skel(F(W)))|}\right]$.
\end{lemma}
\begin{proof}
    At a high level, our proof uses \pref{lem:encode-forest} to encode an unlabeled version of $\skel(F(W))$ in $\left[(4\lambda(W)k\ell)^{2\lambda(W)+1}\right]$ bits and encodes labels using a number in $\left[n^{|V(\skel(F(W)))|}\right]$.

    \paragraph{Encoding ``unlabeled'' version of $\skel(F(W))$.}  Let $\pi$ be an arbitrary function that maps $V(F(W))$ to $\{1,\dots,|V(F(W))|\}$.  Note that the graph $\pi(\skel(F(W)))$ is isomorphic to $\skel(F(W))$.  By \pref{obs:few-leaves-skel}, \pref{lem:encode-forest}, and bounding $|V(\skel(F(W)))|$ by $k\ell$, $\pi(\skel(F(W)))$ can be encoded (up to isomorphism) by a number in $\left[(4\lambda(W)|V(\skel(F(W)))|)^{2\lambda(W)+1}\right]$.

    \paragraph{Encoding labels of $\skel(F(W))$.}  From the encoding of $\pi(\skel(F(W)))$, we can recover a graph on vertex set $\{1,\dots,|V(\skel(F))|\}$ isomorphic to $\skel(F(W))$, which we call $\phi(\skel(F(W)))$.  We thus encode the map $\phi^{-1}$ as it is possible to reconstruct $\skel(F(W))$ from $\phi(\skel(F(W)))$ and $\phi^{-1}$; such a map can be encoded using a number in $\left[n^{|V(\skel(F(W)))|}\right]$.\\
    $~$\\
    Combining the above two encodings proves the lemma.
\end{proof}

\paragraph{Goal 2.}  Our next goal is to give an encoding of the collection of start and end points of each cruise.
\begin{lemma}   \label{lem:encode-cruise-locs}
    Given the encoding of $\skel(F(W))$ from \pref{lem:skeleton-encoding}, the collection of start and end points of each cruise
    \[
        \calC = (C_1[\mathrm{start}],C_1[\mathrm{end}]),\dots,(C_{\gamma+1}[\mathrm{start}],C_{\gamma+1}[\mathrm{end}])
    \]
    can be encoded by a number in $\left[(k\ell)^{\lambda(W)}\right]$.
\end{lemma}
\begin{proof}
    Let $\phi$ be the function from the proof of \pref{lem:skeleton-encoding}.  The sequence
    \[
        \Phi = \phi(C_1[\text{start}]),\phi(C_1[\text{end}]),\dots \phi(C_{\gamma+1}[\text{start}])\footnote{We skip out on $C_{\gamma+1}[\text{end}]$ since it is equal to $C_1[\text{start}]$.}
    \]
    is a sequence of length $\lambda(W)$ of elements in $\{1,\dots,|V(\skel(F))|\}$ and $|V(\skel(F))|\le k\ell$, and hence can be encoded by a number in $\left[(k\ell)^{\lambda(W)}\right]$.  $\calC$ can be recovered from $\Phi$ and $\phi^{-1}$, and since the encoding of $\skel(F(W))$ gives us $\phi^{-1}$, so we are done.
\end{proof}

\paragraph{Goal 3: Encoding cruises.}  Now we move on to encoding cruises.  Let $C_i$ be a cruise that starts at terminal vertex $t_{\Start}$ and ends at terminal vertex $t_{\End}$. 
\begin{remark}  \label{rem:path-in-cruise}
There is a unique path between $t_{\Start}$ and $t_{\End}$ in $F$ as follows:
\begin{align*}
    v_0v_1v_2\dots v_pv_{p+1}.
\end{align*}
where $v_0=t_{\Start}$ and $v_{p+1}=t_{\End}$.
\end{remark}
\begin{definition}
    Let $C_i$ be a cruise.  We say a contiguous subwalk of $C_i$ is a \emph{detour} if it starts and ends at the same vertex.
\end{definition}
\begin{claim}
    Cruise $C_i$ can be constructed by taking the path from $t_{\Start}$ to $t_{\End}$ as described in \pref{rem:path-in-cruise} and inserting at most one detour after each vertex in the path.  In particular, $C_i$ can be written in the form
    \[
        C_i = v_0\dots v_{j_1}\Det_{i,j_1}\dots v_{j_2}\Det_{i,j_2}\dots \dots v_{j_b}\Det_{i,j_b}\dots v_{p+1}
    \]
    where $0\le j_1\le\dots\le j_{b}\le p+1$.
\end{claim}
\begin{proof}
    We can express $C_i$ in the desired form using the following recursive procedure:
    \begin{displayquote}
        If every vertex is visited once, the path from \pref{rem:path-in-cruise} is the cruise.  If there exists a vertex that occurs more than once,
        find the first such visited vertex $v_{j_1}$, and define $\Det_{i,j_1}$ as the subwalk of $C_i$ between the first and last occurrence of $v_{j_1}$; now repeat this procedure on the walk starting at the last occurrence of $v_{j_1}$ and ending at the end of the cruise.
    \end{displayquote}
\end{proof}

\paragraph{Goal 3.1: Encoding locations of detours.}  Recall that $W$ is composed of $k$ links of length-$\ell$ each.  We utilize this structure of $W$ to encode the locations as well as the length of all detours in $W$.
\begin{definition}
    Given a detour $\Det$ in $W$, we say the \emph{timestamp} of $\Det$ is the tuple $(a,b)$ where $a$ is the position of the start step of $\Det$ in $W$ and $b$ is the position of the end step of $\Det$ in $W$.
\end{definition}

\begin{lemma}   \label{lem:encode-detour-loc}
    There is an encoding of the timestamps of all detours in $W$ in $\left[(\ell+1)^{2k}\right]$.
\end{lemma}
\begin{proof}
    Let $L_1,\dots,L_k$ denote the $k$ links that compose $W$.  Due to the nonbacktracking nature of links, each link can have at most one ``start step'' of a detour and at most one ``end step'' of a detour.  We associate a tuple $(a_i,b_i)$ to link $L_i$ where $a_i$ is $0$ if there is no start step of a detour in $L_i$ and the position of that step (which is a number in $[\ell]$) if there is such a step.  Likewise, $b_i$ is $0$ if $L_i$ contains no end step, and is the position of the end step otherwise.  It is possible to reconstruct timestamps of all detours from the $m$ tuples $(a_i,b_i)$, and since each tuple can be encoded by a number in $\left[(\ell+1)^2\right]$, this list of tuples can be encoded by a number in $\left[(\ell+1)^{2k}\right]$.
\end{proof}

\paragraph{Goal 3.2: Encoding detours.}  Before describing how we encode detours we make some structural observations about detours.
\begin{claim}   \label{claim:det-in-well-behaved-part}
    All the edges visited by any detour $\Det$ are in $D^*(W)\cap F(W)$.
\end{claim}
\begin{proof}
    Since $\Det$ is contained inside a cruise, all its edges are in $F(W)$.  Hence, all edges of $\Det$ are in $D^*(W)\cup S(W)$ because $F(W)$ is a spanning forest of $D^*(W)\cup S(W)$.  Since $\Det$ is a closed walk in a tree, it must visits each edge an even number of times; in particular, $\Det$ does not contain any singleton edges and hence is completely contained in $D^*(W)$.
\end{proof}

\begin{corollary}   \label{cor:no-vexing}
    For any detour $\Det$, the graph $G(\Det)$ has no $(t,r,\eps)$-vexing vertices.
\end{corollary}

\begin{observation} \label{obs:det-to-links}
    Any detour $\Det$ can be decomposed into a sequence of links of length exactly $\ell$, with the exception of the first and last link, which can both have any length between $1$ and $\ell$.
\end{observation}

\begin{definition}
    Any detour $\Det$ starts and ends at some vertex $v$.  We call $v$ the \emph{root} of $\Det$ and denote it with $\Root(\Det)$.
\end{definition}
\begin{remark}
    One should think of a detour as a closed walk on a tree rooted at a distinguished vertex.
\end{remark}

\begin{definition}
    We call a step from $u$ to $v$ in $\Det$ an \emph{up-step} if $v$ is closer to $\Root(\Det)$ than $u$.  In similar spirit, we call that step a \emph{down-step} if $v$ is further from $\Root(\Det)$ than $u$.
\end{definition}

\begin{definition}
    We further classify down-steps in a detour $\Det$ into three types:
    \begin{enumerate}
        \item We call a down-step from $u$ to $v$ a \emph{fresh skeleton step} if the edge $\{u,v\}$ is part of $\skel(F(W))$ and has not been traversed by any detour so far.
        \item We call a down-step from $u$ to $v$ a \emph{fresh intrepid step} if the edge $\{u,v\}$ is \emph{not} part of $\skel(F(W))$ and has not been traversed so far.  We use $f_i$ to denote the total number of fresh intrepid steps across all detours in the walk.
        \item We call a down-step from $u$ to $v$ a \emph{stale step} if it is not a fresh skeleton step or a fresh intrepid step.
    \end{enumerate}
\end{definition}

\begin{claim}   \label{claim:stale-occurred-before}
    Suppose there is a stale step from $u$ to $v$ at time $T$.  Then there is an occurrence of a step from $u$ to $v$ as well as from $v$ to $u$ in a detour at an earlier time.
\end{claim}
\begin{proof}
    Since the step at time $T$ between $u$ and $v$ occurs in a detour, the edge $\{u,v\}$ must be part of $D^*\cap F(W)$.  If $\{u,v\}$ is part of $\skel(F(W))$, then it must have been traversed in a detour at a time before $T$, since otherwise this step would be classified as a fresh skeleton step.  If $\{u,v\}$ is not part of $\skel(F(W))$, then it must be part of $F(W)\setminus\skel(F(W))$ and these edges are only traversed in detours; and if $\{u,v\}$ was not traversed in an earlier detour, it would have been classified as a fresh intrepid step.

    Thus, we have established that the edge $\{u,v\}$ is traversed by a detour.  Now, if $\{u,v\}$ was traversed in a detour, there must have been both a step from $u$ to $v$ and a step from $v$ to $u$ since if a directed edge is traversed in a detour then so is its reversal; in particular, a step between $u$ and $v$ occurs in a detour before time $T$.
\end{proof}

\begin{definition}
    We call a (possibly empty) contiguous sequence of steps a \emph{stretch}.
\end{definition}

\begin{observation} \label{obs:link-phases}
    Due to nonbacktracking nature of links and the tree structure of detours, every link in a detour can be broken into $4$ phases:
    \begin{itemize}
        \item Phase 1: an up-stretch,
        \item Phase 2: a stale stretch,
        \item Phase 3: a fresh skeleton stretch
        \item Phase 4: a fresh intrepid stretch.
    \end{itemize}
\end{observation}

\begin{lemma}   \label{lem:detour-enc}
    Given the encoding of $\skel(F(W))$ from \pref{lem:skeleton-encoding}, the encoding of endpoints of cruises from \pref{lem:encode-cruise-locs}, and the encoding of timestamps of detours from \pref{lem:encode-detour-loc}, it is possible to encode all detours in $W$ using a number in
    \[
        \left[\ell^{4k}\cdot\left((1+\eps)d\right)^{tm}\cdot \left((1+\eps)d\right)^{\frac{1}{2}(k\ell-2|D(W)|-|S(W)|)}\cdot(3\lambda(W)+1)^{5\lambda(W)}\cdot n^{f_i}\right].
    \]
\end{lemma}
\begin{proof}
    Let $\Det_1,\dots,\Det_b$ be the sequence of detours of $W$ in order of time.  We first specify how we encode detours, and then prove that the encoding is valid, i.e., recovery of all $\Det_a$ from the given encoding is possible.  As pointed out in \pref{obs:det-to-links} each $\Det_a$ can be broken into a sequence of links $L_1,\dots,L_{\tau}$.  We now describe how to encode each $L_j$.

    \paragraph{Encoding metadata.}  For each link $L_j$, we first specify four numbers in $[\ell]$ denoting the lengths of the up-stretch, stale stretch, fresh skeleton stretch, and fresh intrepid stretch in the detour.  Now we zoom in and encode each phase carefully.

    \paragraph{Encoding up-stretches.}  We don't specify any extra information about the up-stretch.

    \paragraph{Encoding a stale stretch.}  Given a stale stretch $\zeta$, let $E_{\zeta}$ denote the set of edges visited before $\zeta$ starts.  From \pref{claim:det-in-well-behaved-part} $\zeta$ is completely contained in $D^*(W)$.
    Since $\zeta$ is a stale stretch, it must be contained in $E_{\zeta}\cap D^*(W)$.  We first break $\zeta$ into $\left\lceil\frac{|\zeta|}{t}\right\rceil$ substretches $\zeta_1,\dots,\zeta_{\left\lceil\frac{|\zeta|}{t}\right\rceil}$ each of length at most $t$ and encode each substretch.  Let $v_i$ be the vertex at the start of $\zeta_i$ and $v_i'$ be the end of $\zeta_i$.  Since $E_{\zeta}\cap D^*$ has no $(t,r,\eps)$-vexing vertices, there are at most $((1+\eps)d)^t$ vertices within distance $t$ of $v_i$; in particular, there are at most $((1+\eps)d)^t$ possible candidates for $v_i'$.  We sort these candidates in increasing order of time first visited in a detour, and encode $\zeta_i$ with the index of $v_i'$ in this list of candidates.  Note that this index is a number in $\left[((1+\eps)d)^t\right]$.  To encode $\zeta$, we specify $\left\lceil\frac{|\zeta|}{t}\right\rceil$ such numbers, one corresponding to each $\zeta_i$.

    \paragraph{Encoding a fresh skeleton stretch.}  For each step $u\to v$ of the fresh skeleton stretch, we don't specify any information if the degree of $u$ within $\skel(F(W))$ is $\le 2$ and $u$ is not a terminal.  If the degree of $u$ is at least $3$ or if $u$ is a terminal, we create a list of neighbors of $u$ sorted in increasing order of their identities in $K_n$, and specify the index of $v$ in this list.  Note that this index is at most the degree of $u$ within $\skel(F(W))$, which from \pref{fact:tree-deg} is at most $3\times(\#\text{ leaves in $\skel(F(W))$})$, which in turn from \pref{obs:few-leaves-skel} is bounded by $3\lambda(W)$.
    
    \paragraph{Encoding a fresh intrepid stretch.}  For every fresh intrepid step $uv$, we specify the identity of $v$ in $K_n$, so each fresh intrepid step is encoded by a number in $[n]$.
   
    \paragraph{Recovery of detours.}  We now show how to recover the detours from the given encodings.  First, it is possible to recover the root of every detour from the encodings given by \pref{lem:skeleton-encoding}, \pref{lem:encode-cruise-locs} and \pref{lem:encode-detour-loc}.  We now show how to recover the detours in order
    \[
        \Det_1,\Det_2,\dots,\Det_b.
    \]
    Suppose $\Det_1,\dots,\Det_i$ have been recovered, we show how to recover $\Det_{i+1}$.  Let $L_1,\dots,L_{\tau}$ be the links in $\Det_{i+1}$.  We show how to sequentially recover the links.  Suppose $L_1,\dots,L_{j}$ have been recovered.  We now describe how to recover $L_{j+1}$.
    
    \paragraph{Recovering the up-stretch in $L_{j+1}$.} The length of the up-stretch, which is part of the ``metadata encoding'' is sufficient to reconstruct the up-stretch of $L_{j+1}$.

    \paragraph{Recovering the stale stretch in $L_{j+1}$.} By \pref{claim:stale-occurred-before} every step in the stale stretch of $L_{j+1}$ has been taken in a detour before.  Since we know the the length of the stale stretch in $L_{j+1}$ from the metadata encoding, and we have recovered all steps before the stale stretch in $L_{j+1}$ that are part of a detour, we can infer a list of candidate endpoints of the stale stretch.  Further, we also know the order in which these candidates were visited in detours, and hence we can recover the stale stretch in $L_{j+1}$ from the encoding of stale stretches we described.

    \paragraph{Recovering the fresh skeleton stretch in $L_{j+1}$.} Now we describe how to recover the fresh skeleton stretch of $L_{j+1}$.  Once the stale stretch of $L_{j+1}$ has been recovered, we know the start vertex of this stretch, $v$.  We also can infer the length of the fresh skeleton stretch $\mathsf{LenSkel}$ from the metadata encoding.  We recover this full stretch by performing the following walk, which traces the same steps as the fresh skeleton stretch of $L_{j+1}$:
    \begin{itemize}
        \item Let $x$ be a counter that is initially $0$.
        \item Let $v'$ be initially set to $v$ ($v'$ denotes the ``current vertex'' in our walk).
        \item While $x\le \mathsf{LenSkel}$:
        \begin{itemize}
            \item If the degree of $v'$ within $\skel(F(W))$ is $\le 2$ and $v'$ is not a terminal, then step along the unique unvisited edge incident to $v'$ (called $v'w$) and update $v'$ to $w$.  \emph{Note that if the first $x$ steps of this walk and those in the fresh skeleton stretch coincide, then $v'w$ must be the $(x+1)$-th step in the fresh skeleton stretch.} 
            \item If the degree of $v'$ within $\skel(F(W))$ is $\ge3$ or $v'$ is a terminal vertex: then assuming the first $x$ steps of the current walk match those of the fresh skeleton stretch, we can recover the next step $v'w$ of the fresh skeleton stretch from the encoding of $\skel(F(W))$ in \pref{lem:skeleton-encoding} combined the encoding of fresh skeleton stretches described earlier in this proof.  Thus, we update $v'$ to $w$.
            \item Increment $x$ by $1$.
        \end{itemize}
    \end{itemize}

    \paragraph{Recovering the fresh intrepid stretch in $L_{j+1}$.}  We can straightforwardly recover this stretch step-by-step since the identity of each vertex within $K_n$ is given in the encoding.
    
    \paragraph{Recovery wrapup.} Thus, we have established how we recover link $L_{j+1}$ from the given encoding and all links in all detours that occurred before.  Inductively, this gives us a method to recover all detours in $W$.
    
    \paragraph{Counting.}  Now we finally turn our attention to bounding the number of encodings of all detours.  We will bound the number of metadata encodings, the number of stale stretch encodings, the number of fresh skeleton stretch encodings and finally the number of fresh intrepid stretch encodings.
    
    \paragraph{Bounding the number of metadata encodings.}  Since there are at most $k$ links in detours and the metadata of each link contains $4$ numbers in $[\ell]$, there are at most $\ell^{4k}$ possible metadata encodings.

    \paragraph{Bounding the number of stale stretch encodings.}  Let us call the stale stretch corresponding to a link $L$ as $\zeta(L)$.  Each stale stretch $\zeta$ is encoded using $\left\lceil\frac{|\zeta|}{t} \right\rceil$ numbers in $\left[((1+\eps)d)^t\right]$.  The total number of stale stretch encodings is then bounded by
    \[
        \prod_{L\in\text{Links}(W)}\left(\left((1+\eps)d\right)^t\right)^{\left\lceil\frac{|\zeta(L)|}{t} \right\rceil} \le \left(\left((1+\eps)d\right)^t\right)^{\sum_{L\in\text{Links}(W)}\left(\frac{|\zeta(L)|}{t}+1\right)}.   \numberthis \label{eq:stale-stretch-bound}
    \]
    We turn our attention to bounding $\sum_{L\in\text{Links}(W)}\left(\frac{|\zeta(L)|}{t}+1\right)$.
    \begin{align*}
        \sum_{L\in\text{Links}(W)}\left(\frac{|\zeta(L)|}{t}+1\right) &= k + \frac{1}{t}\sum_{L\in\text{Links}(W)}|\zeta(L)| \numberthis \label{eq:bits-to-encode-stale}
    \end{align*}
    Note that $\sum_{L\in\text{Links}(W)}|\zeta(L)|$ is the total number of stale steps across all detours.  From \pref{claim:stale-occurred-before} the (undirected) edge that a stale step is taken on is being traversed for \emph{at least} the third time.  Further, since the stale step is a down-step, there must be a corresponding up-step that is the reversal of the down-step in the detour.  Thus, an edge $\{i,j\}$ is traversed by a stale step at most $\frac{a_{ij}-2}{2}$ times.  Further, since there are multiple steps that traverse the same edge that a given stale step traverses, every stale step must traverse an edge in $D(W)$.  Thus, we can bound \pref{eq:bits-to-encode-stale} by:
    \begin{align*}
        k + \frac{1}{t}\sum_{\{i,j\}\in D(W)} \frac{1}{2}(a_{ij}-2) &= k + \frac{1}{2t}\left(\sum_{ij:a_{ij}\ge 2}(a_{ij}-2) + \sum_{ij:a_{ij}=1}(a_{ij}-1)\right)\\
        &= k + \frac{1}{2t}\left(k\ell - 2|D(W)| - |S(W)|\right)
    \end{align*}
    Plugging in the above into \pref{eq:stale-stretch-bound} gives us a bound of:
    \[
        \left((1+\eps)d\right)^{tk}\cdot \left((1+\eps)d\right)^{\frac{1}{2}(k\ell-2|D(W)|-|S(W)|)}.
    \]

    \paragraph{Bounding the number of fresh skeleton stretch encodings.}  Let $P$ be the set of vertices that either are terminal vertices or have degree-$\ge 3$ in $\skel(F(W))$.  We can extract our encoding of fresh skeleton stretches from the following map $H$.
    \begin{displayquote}
        For every $v\in P$, $H(v)$ is equal to the list of numbers in $[\deg_{\skel(F(V))}(v)]$ such that number $i$ is in this list if $vw_i$ is a fresh skeleton step, where $w_i$ is the $i$th neighbor of $v$ in lexicographic order of names in $K_n$; further, this list is sorted in order of time the corresponding steps are taken. 
    \end{displayquote}
    There are at most $(\deg_{\skel(F(V))}(v)+1)^{\deg_{\skel(F(V))}(v)}$ possibilities for $H(v)$ since every edge in the skeleton can occur at most once in a fresh skeleton stretch.  Since the number of possible encodings is upper bounded by the number of candidates for $H$, we have a bound of
    \begin{align}
        \prod_{v\in P} (\deg_{\skel(F(V))}(v)+1)^{\deg_{\skel(F(V))}(v)} &\le (3\lambda(W)+1)^{\sum_{v\in P}\deg_{\skel(F(V))}(v)} \label{eq:bound-fresh-skel}
    \end{align}
    Now we focus on bounding $\sum_{v\in P}\deg_{\skel(F(V))}(v)$.
    \begin{align*}
        \sum_{v\in P}\deg_{\skel(F(V))}(v) &= \sum_{v:\deg_{\skel(F(V))}(v)\ge 3} \deg_{\skel(F(V))}(v) + \sum_{\substack{v:\deg_{\skel(F(V))}(v)\le 2\\ v\in T(W)}} \deg_{\skel(F(V))}(v)
    \end{align*}
    From \pref{fact:tree-deg} the first term is bounded by $3\times\#\text{ leaves in $\skel(F(W))$}$, which from \pref{obs:few-leaves-skel} is bounded by $3\lambda(W)$.  The second term is bounded by $2|T(W)|$, which from \pref{rem:terminal-bound} is at most $2\lambda(W)$.  As an upshot we have:
    \[
        \sum_{v\in P}\deg_{\skel(F(V))}(v) \le 5\lambda(W).
    \]
    Plugging this into \pref{eq:bound-fresh-skel} gives us a bound on the number of possible skeleton fresh stretch encodings of:
    \[
        (3\lambda(W)+1)^{5\lambda(W)}.
    \]

    \paragraph{Bounding the number of fresh intrepid stretch encodings:}  The encoding of fresh intrepid stretches comprises of $f_i$ identities of vertices in $K_n$, each of which is represented by a number in $[n]$.  Hence there are at most $n^{f_i}$ fresh intrepid stretch encodings.

    Combining all the above bounds, we get a bound on the total number of possible encodings of all the detours of
    \[
        \ell^{4k}\cdot\left((1+\eps)d\right)^{tk}\cdot \left((1+\eps)d\right)^{\frac{1}{2}(k\ell-2|D(W)|-|S(W)|)}\cdot(3\lambda(W)+1)^{5\lambda(W)}\cdot n^{f_i}
    \]
\end{proof}

Since it is possible to recover a linkage $W$ from $\skel(W)$, the endpoints of its cruises and the order in which the cruises occur, the timestamps of the detours, and the detours, by a combination of \pref{lem:skeleton-encoding}, \pref{lem:encode-cruise-locs}, \pref{lem:encode-detour-loc} and \pref{lem:detour-enc} along with a bound on $\lambda(W)$ from \pref{claim:bound-gamma} we have the following bound:
\begin{theorem}[Restatement of \pref{thm:main-count}] \label{thm:main-count-restated}
    The total number of $(k,\ell)$-linkages with $f$ fresh edges, $e$ excess edges, $s$ singleton edges and $\Delta$ profligate steps is at most:
    \begin{align*}
        n^{f+1}\cdot(4\lambda(W))^{7\lambda(W)+1}\cdot(k\ell)^{3\lambda(W)+1}\cdot(\ell+1)^{6k}\cdot
        ((1+\eps)d)^{tk+k\ell/2-|D(W)|-s/2}
    \end{align*}
    where $\lambda(W)\le 3e + \frac{12k\ell\ln(k\ell)}{r} + 3\Delta$.
\end{theorem}

%% file: Content/sbm-lower-bound.tex
\section{Lower Bounds in the Stochastic Block Model}
\label{sec:sbm-lb}

In this section, we finish the proof of \pref{thm:local-path-stats} by proving lower bounds for the level-$M$ path statistics SDP (as described by \pref{def:lvl-m-ps-sdp}) for every constant $M$ for detection in the stochastic block model under the Kesten-Stigum threshold.

An ingredient we will need is an Ihara--Bass formula for weighted graphs, which appears in \cite{WF11,FM17} as well as a related power series identity, which to our knowledge is novel.  We give a proof for the sake of being self-contained.

\subsection{Weighted Ihara-Bass and a Power Series Identity}    \label{sec:IB-power-series}

Let $G = (V,E)$ be any graph. For any edge weights $c : E \to \bbR$, write $A_c$ for the weighted adjacency matrix of $G$, and $D_c$ for the diagonal matrix of $c$-weighted vertex degrees. More generally let $A_c^{(\ell)}$ count $c$-weighted non-backtracking walks on $G$, $C \in \R^{2|E|\times 2|E|}$ be the diagonal matrix with $C_{i \to j, i\to j} = C(i \to j)$, and write $B_c = CB$ where $B$ is the nonbacktracking matrix of the complete graph.

\begin{theorem}[Weighted Ihara-Bass]    \label{thm:weighted-IB}
    For any weights $c: E \to \bbR$, let $\hat c = c(1 - c^2)^{-1}$. Then
    $$
        \det(1 - B_c) = \prod_{(i,j) \in E} (1 - c(i,j)^2) \det(1 - A_{\hat c} + D_{c\hat c}),
    $$
    and 
    $$
        (1 - A_{\hat c} + D_{c\hat c})^{-1} = \sum_{\ell \ge 0} A_{c}^{(\ell)}
    $$
    whenever this series converges.
\end{theorem}

\begin{proof}
    Regard each edge as a pair of directed edges in opposite directions. Write $S \in \bbR^{|V| \times 2|E|}$ and $T \in \bbR^{2|E| \times |V|}$ for the \textit{start} and \textit{terminal} matrices (i.e. if $(u,v) \in E$ the former has $S_{u,u \to v} = 1$ and the latter has $T_{u\to v,v} = 1$) and $\Pi \in \bbR^{2|E|\times 2|E|}$ for the involution that reverses directed edges. Let's adopt the convention that $B = TS - \Pi$, and note for later that $C\Pi = \Pi C$, since the weights $c$ are a function of undirected edges. Moreover $S\Pi C T = D_c$ and $S(CB)^\ell CT = A_c^{(\ell + 1)}$ for every $\ell \ge 0$; indeed analogous identities hold for any diagonal weight matrix commuting with $\Pi$.
    
    Now consider the matrix
    $$
        \frB_c \triangleq \begin{pmatrix} 
            1 & S \\ CT & 1 + C\Pi
        \end{pmatrix}.
    $$
    We can compute the determinant of $\frB_c$ using two different Schur complements:
    $$
        \det \frB_c = \det(1 - CB) = \det(1 + C\Pi)\det(1 - S(1 + C\Pi)^{-1}CT).
    $$
    It remains now to understand the matrix $1 - S(1 + C\Pi)^{-1}CT$. Since $C$ and $\Pi$ commute,
    $$
        (1 + C\Pi)^{-1} = (1 - C^2)^{-1}(1 - C\Pi)
    $$
    making
    \begin{align*}
        1 - S(1 + C\Pi)^{-1}CT 
        &= 1 - S\left((1 - C^2)^{-1} - C(1 - C^2)^{-1}\Pi\right)CT \\
        &= 1 - A_{\hat c} + D_{c\hat c};
    \end{align*}
    the second line follows from our initial discussion and the definition $\hat c = c(1 - c^2)^{-1}$.
    
    To prove the power series identity, let invert $\frB_c(z)$ with the Schur complement formula:
    \begin{align*}
        1 
        &= \frB_c\frB_c^{-1} \\
        &= \begin{pmatrix} 
                1 & S \\
                CT & 1 + C\Pi
            \end{pmatrix}
            \begin{pmatrix}
                (1 - A_{\hat c} + D_{c\hat c})^{-1} & -S(1 - CB)^{-1} \\
                -(1 - CB)^{-1}CT & (1 - CB)^{-1}
            \end{pmatrix}.
    \end{align*}
    Considering the upper left block, we see
    \begin{align*}
        (1 - A_{\hat c} + D_{c\hat c})^{-1} 
        &= 1 + S(1 - CB)^{-1}CT \\
        &= 1 + \sum_{\ell \ge 0}S(CB)^{\ell}CT \\
        &= \sum_{\ell \ge 0} A_c^{(\ell)}
    \end{align*}
\end{proof}

% We also require a closed form expression for the partial sums of the series expression for $L_c^{-1} \triangleq (1 - A_{\hat c} + D_{c\hat c})^{-1}$. The tails of this series can be expanded using identities from the other blocks of $\frB_c\frB_c^{-1}$. Note that
% \begin{align*}
%     CT L_c^{-1} &= (1 + C\Pi)(1 - CB)^{-1}CT
% \end{align*}
% meaning we can iterate to see
% \begin{align*}
%     (1 - CB)^{-1}CT &= CTL_c^{-1} - C\Pi(1 - CB)^{-1}CT \\
%     &= CTL_c^{-1} - C\Pi \left(CTL_c^{-1} - C\Pi(1 - CB)^{-1}CT\right) \\
%     &= (1 - C\Pi)CTL_c^{-1} + C^2(1 - CB)^{-1}CT \\
%     \Rightarrow (1 - CB)^{-1}CT &= (1 - C^2)^{-1}(1 - C\Pi)CTL_c^{-1}.
% \end{align*}
% Now,
% \begin{align*}
%     \sum_{\ell \ge m+1} A_c^{(\ell)}
%     &= S(CB)^m(1 - CB)^{-1}CT \\
%     &= S(CB)^m(1 - C^2)^{-1}(1 - C\Pi)CTL_c^{-1} \\
%     &= \left(S(CB)^m\hat CT - S(CB)^m C\hat C \Pi T\right) L_c^{-1}
% \end{align*}

\subsection{Construction of SDP solution}
Let $\bG$ be a $G(n,d/n)$ graph.  Our goal is to construct a solution to the SDP given in \pref{def:lvl-m-ps-sdp} when $\bG\sim G(n, d/n)$, and $d$ is under the KS threshold.  We instead construct a solution to the following simpler SDP, and obtain a solution for the SDP in \pref{def:lvl-m-ps-sdp} via an identical procedure to the one described after the statement of \pref{prop:paths-suffice}.  Given parameters $\lambda,M,\delta$ and graph $G$:
\begin{align*}
    \text{Find $n\times n$ matrix $Y\psdge 0$ s.t.}\\
    Y_{i,i} &= 1 &\forall i\in[n]\\
    \left\langle Y, \left(A_{G}-\frac{d}{n}11^{\top}\right)^{(\ell)} \right\rangle &= d^{\ell}\lambda^{\ell} n \pm O(\delta n) &\forall \ell\le M.  \numberthis \label{eq:main-SDP-lower-bound}
\end{align*}
Our main technical result in this section is:
\begin{theorem} \label{thm:SBM-lower-bound}
    For $\bG\sim G(n, d/n)$, for $|\lambda|<\frac{1}{\sqrt{d}}$, and for any $\delta, M > 0$, the SDP \pref{eq:main-SDP-lower-bound} is feasible with high probability.
\end{theorem}

Let $\eps > 0$ be an arbitrary constant, $\ell_0 \in [\ceil{\log n \log\log n},2\ceil{\log n\log\log n}]$, $t = \ell_0^{1/3}$, $r=\frac{2\ell_0}{\ln^3(2\ell_0)}$; let $\bG_{t,r,\eps}$ be its $(t,r,\eps)$-truncation and let $\bA_{t,r,\eps}$ denote the adjacency matrix of $\bG_{t,r,\eps}$.  Now, let $S$ be the set of vertices deleted in truncating $\bG$, and define edge weights $c:E\to \R$ so that
\[
    A_c = \bA_{t,r,\eps} - \frac{d}{n}1_{[n]\setminus S}1_{[n]\setminus S}^{\top}
\]
Define $A_c^{(m)}$ as $\1$ when $m = 0$ and akin to how $\overline{A}^{(m)}$ was defined in \pref{sec:proof-sdp-val-m-walk} when $m\ge 1$.  And finally define $B_c$ the way it is defined in \pref{sec:IB-power-series}.  Our next ingredient is establishing an operator norm bound on $B_c^{\ell_0}$.  Indeed:
\[
    \|B_c^{\ell_0}\| \le \sqrt{\Tr\left(B_c^{\ell_0}(B_c^*)^{\ell_0}\right)}.
\]
The above quantity can be seen to be upper bounded by:
\[
    \sqrt{n^2\Tr\left(\left(A_c^{(\ell_0-1)}\right)^2\right)}
\]
which from \pref{eq:trace-power-root-bound} is bounded by:
\[
    n\cdot\left((1+\eps)^4\sqrt{d}\right)^{\ell_0-1},
\]
which by our choice of $\ell_0$ is at most
\[
    \left((1+\eps)^5\sqrt{d}\right)^{\ell_0}.
\]
Note that the following is true for \emph{any} $\ell_0\in I\coloneqq [\ceil{\log n\log\log n}, 2\ceil{\log n\log\log n}]$:
\[
    \|B_c^{\ell_0}\| \le \left((1+\eps)^5\sqrt{d}\right)^{\ell_0}.  \numberthis \label{eq:Bell-seed-bound}
\]
Since any $\ell\ge 2\ceil{\log\log n}$ can be expressed as
\[
    \ell \coloneqq \ell_1 + \dots + \ell_s
\]
for $\ell_i\in I$, we can conclude from a combination of submultiplicativity of operator norm and \pref{eq:Bell-seed-bound} that
\[
    \|B_c^{\ell}\| \le \|B_c^{\ell_1}\|\cdots\|B_c^{\ell_s}\| \le \left((1+\eps)^5\sqrt{d}\right)^{\ell}.  \numberthis \label{eq:Bell-norm-bound}
\]
Via the expression $A_c^{(\ell)} = SB_c^{\ell-1}CT$ in the proof of \pref{thm:weighted-IB} and the fact that $\|S\|\le n$ and $\|CT\|\le n$, we know:
\[
    \|A_c^{(\ell)}\| \le \left((1+\eps)^6\sqrt{d}\right)^{\ell} \numberthis \label{eq:Anb-bound}
\]
for all $\ell\ge\ell_0$.  Another consequence of $\pref{eq:Bell-norm-bound}$ is
\[
    \rho(B_c) \le \|B_c^{\ell}\|^{1/\ell} \le (1+\eps)^5\sqrt{d}.   \numberthis \label{eq:specrad-bound}
\]

Now, let
\[
    M_s(z) \coloneqq \sum_{0\le\ell\le s} A_c^{(\ell)}z^{\ell}.
\]
Define $\hat{c}$ in terms of $c$ identically to how it is defined in the statement of \pref{thm:weighted-IB}.  From \pref{thm:weighted-IB},
\[
    M_{\infty}(z) = (\1-A_{\widehat{cz}}+D_{cz\widehat{cz}})^{-1}.
\]
Next, we use a proposition that is similar to (and whose proof follows) a similar statement in \cite{WF11,FM17}:
\begin{proposition}
    Suppose $z\in\R$ and $|z| < \min\{1/\rho(B_c),1\}$, then
    $M_{\infty}(z)\psdge 0$.
\end{proposition}
\begin{proof}
    $M_\infty(0)$ is the identity matrix and hence is certainly positive definite, which means all its eigenvalues are positive.  Additionally, by the fact that all edge weights $c(i,j)$ are bounded by $1$ and the weighted Ihara--Bass formula (\pref{thm:weighted-IB}), we can deduce that for all real $z$ such that $|z|<\min\{1/\rho(B_c),1\}$, $\det(M_\infty(z))>0$.  Since the determinant (which is the product of eigenvalues) is strictly positive on a continuous interval, the eigenvalues of $M_{\infty}(z)$ are a continuous function of $z$ on this interval, and the eigenvalues of $M_{\infty}$ are strictly positive at one point in this interval, all eigenvalues of $M_{\infty}$ must be positive for all real $z$ where $|z| < \min\{1/\rho(B_c),1\}$.  Thus, the proposition follows.
\end{proof}
Our next goal will be to lower bound the minimum eigenvalue of $M_{r/2-1}(z)$, i.e. prove that the minimum eigenvalue is not too negative when $z$ is in an appropriate range.
\begin{proposition} \label{prop:lb-min-eig-M}
    Suppose $|z| < \frac{1}{(1+2\eps)^{6}\sqrt{d}}$, then $\lambda_{\min}(M_{r/2-1}(z))\ge-\delta(n)$ where $\delta(n) = o_n(1)$.
\end{proposition}
\begin{proof}
$\lambda_{\min}(M_{r/2-1}(z)) = \lambda_{\min}\left(M_{\infty}(z)-\sum_{\ell\ge \floor{r/2}}A_c^{(\ell)}z^{\ell}\right)$, which by positive semidefiniteness of $M_{\infty}(z)$ is lower bounded by
\[
    -\left\|\sum_{\ell\ge r/2-1} A_c^{(\ell)}z^{\ell}\right\|\ge-\sum_{\ell\ge r/2-1} \left\|A_c^{(\ell)}\right\||z|^{\ell}.
\]
By a combination of \pref{thm:spec-norm-bound} and \pref{eq:Anb-bound}, along with the assumption on $|z|$ we know the above is lower bounded by
\[
    -\sum_{\ell\ge r/2-1}\left(\frac{1+\eps}{1+2\eps}\right)^{6\ell}\ge -\alpha\left(\frac{1+\eps}{1+2\eps}\right)^{3r-6}.
\]
where $\alpha\coloneqq\sum_{\ell\ge 0}\left(\frac{1+\eps}{1+2\eps}\right)^{6\ell}$ is an absolute constant depending only on $\eps$.  The proposition follows from the choice of $r$.
\end{proof}
Let $\delta(n)$ be the function in the statement of \pref{prop:lb-min-eig-M}, and define
\[
    X(z) \coloneqq (1-\delta(n))\cdot M_{r/2-1}(z) + \delta(n)\cdot\1.
\]
By \pref{prop:lb-min-eig-M}, $X(z)$ is positive semidefinite when $|z| < \frac{1}{(1+2\eps)^{6}\sqrt{d}}$.  At this point, we state a fact that will be used later.
\begin{fact}    \label{fact:deg-bound}
    With probability $1-o_n(1)$, the maximum degree of a vertex in $\bG$ is bounded by $\log^2 n$.
\end{fact}

Our next goal is to argue that the diagonal entries of $X(z)$ are $1\pm o_n(1)$.  Towards this, let us try to understand the contribution of $A_c^{(\ell)}$ to the diagonal.  In particular:
\begin{proposition} \label{prop:bound-on-diag-nbA}
    Let $\ell<r$.  The diagonal entries of $A_c^{(\ell)}$ are all bounded in magnitude by $\frac{\left(2\log^2 n\right)^{\ell}}{n}$ with probability $1-o_n(1)$.
\end{proposition}
\pref{prop:bound-on-diag-nbA} is a direct consequence of \pref{fact:deg-bound} and the forthcoming \pref{prop:bound-off-paths} which we prove and use later.

\begin{proposition} \label{prop:diag-entries-X-1}
    The diagonal entries of $X(z)$ are all $1\pm \delta'(n)$ with probability $1-o_n(1)$ as long as $|z|<1$ where $\delta'(n)$ is some function which is $o_n(1)$.
\end{proposition}
\begin{proof}
    As long as the maximum degree of $G$ is bounded by $\log^2 n$, which happens with probability $1-o_n(1)$, by \pref{prop:bound-on-diag-nbA} the diagonal entries of $X(z)$ are bounded by
    \begin{align*}
        \sum_{0\le \ell \le r/2-1} \text{contribution of $A_c^{(\ell)}$ to diagonal} \le 
        r\cdot\frac{\left(2\log^2 n\right)^{r}}{n},
    \end{align*}
    which is $o_n(1)$.
\end{proof}

Finally for $|z| < \frac{1}{(1+2\eps)^6\sqrt{d}}$ we define $Y(z) \coloneqq (1-\delta'(n))X\left(z\right)+\Gamma$ where $\Gamma$ is a diagonal matrix chosen so that the diagonal of $Y(z)$ is all-ones.  By \pref{prop:diag-entries-X-1}, $\Gamma$ is positive semidefinite and combined with the fact that $X(z)$ is positive semidefinite, we can conclude that $Y(z)$ is positive semidefinite as well.

\subsection{Matching path statistics}
In this section, we are interested in understanding the value of $\left\langle \left(A-\frac{d}{n}11^{\top}\right)^{(\ell)},Y(z)\right\rangle$ where $A$ is the adjacency matrix of $\bG$.  It suffices to understand the value of each $\left\langle \left(A-\frac{d}{n}11^{\top}\right)^{(\ell)},A_c^{(m)}\right\rangle$ where $\ell,m\le\frac{r}{2}-1$.  To lighten the notation, in this subsection we will use $A'$ to denote $\bA_{t,r,\eps}$.

To set up \pref{prop:bound-off-paths}, we introduce some notation --- let $\mathrm{NB}_{i,j,\ell}$ for $i,j\in[n]$ and $\ell\in\bbN$ denote the set of all nonbacktracking walks that start at $i$, end at $j$, and are of length-$\ell$.  Given a nonbacktracking walk $W$, we will interpret $W$ as a set of tuples $(i,ab)$ for $i\in[\ell]$ where $ab$ is the edge walked on in the $i$th step of $W$.  For any $S\subseteq W$, we will overload notation and use $S$ to also denote the subset of edges with the timestep-indices removed.
\begin{proposition} \label{prop:bound-off-paths}
    Suppose $G$ is a $n$-vertex graph with maximum degree bounded by $\Delta\ge d$, then for all $i,j\in[n]$:
    \[
        \left|\sum_{\substack{(W,S):S\subseteq W, S\ne W \\ W\in\mathrm{NB}_{i,j,\ell} \\ S\subseteq E(G)}} \left(-\frac{d}{n}\right)^{\ell-|S|}\right| \le \frac{\ell\cdot(2\Delta)^{\ell}}{n}.
    \]
\end{proposition}
\begin{proof}
    \begin{align*}
        \left|\sum_{\substack{(W,S):S\subseteq W, S\ne W \\ W\in\mathrm{NB}_{i,j,\ell} \\ S\subseteq E(G)}} \left(-\frac{d}{n}\right)^{\ell-|S|}\right| &\le \sum_{k=1}^{\ell-1} \sum_{\substack{(W,S):S\subseteq W, |S|=k \\ W\in\mathrm{NB}_{i,j,\ell}}} \left(\frac{d}{n}\right)^{\ell-k}\\
        &= \sum_{k=1}^{\ell-1}\left(\frac{d}{n}\right)^{\ell-k}\cdot|\{(W,S):S\subseteq W,|S|=k,W\in\mathrm{NB}_{i,j,\ell},S\subseteq E(G)\}| \numberthis \label{eq:count-weight}.
    \end{align*}
    For each $k$, we bound the above summand via an encoding argument.
    Given a walk $W=u_0u_1u_2\dots u_{\ell-1}u_{\ell}$ where $u_0 = i$ and $u_{\ell}=j$ along with a proper subset of steps $S$ which are all contained in $E(G)$, we first find the last $u_{t-1}u_{t}$ on the segment which is not in $S$ (which always exists since $S$ is always a proper subset).  Therefore the segment $u_t\dots u_{\ell}$ must be composed only of edges in $G$.  In our encoding we specify:
    \begin{itemize}
        \item The set $S$ via timestamps (for which there are $2^{\ell}-1$ choices),
        \item The value of $t$ (for which there are $\ell$ choices),
        \item For $t\le s\le \ell-1$, the index $a\in[\Delta]$ such that $u_s$ is the $a$-th neighbor of $u_{s+1}$ in $G$.
        \item For $1\le s\le t-1$, if $u_{s-1}u_s$ is in $S$, we specify index $a\in[\Delta]$ such that $u_s$ is the $a$-th neighbor of $u_{s-1}$ in $G$; and if $u_{s-1}u_s$ is \emph{not} in $S$, we specify the identity of $u_s$ (of which there are $n$ choices).
    \end{itemize}
    
    The total number of possible encodings is bounded by
    \[
        \ell \cdot 2^{\ell} \cdot \left(\Delta\right)^{k} n^{\ell-k-1}.
    \]
    Thus, the $k$-th summand is bounded by
    \[
        \frac{\ell \cdot 2^{\ell} \cdot \left(\frac{\Delta}{d}\right)^k d^{\ell}}{n}.
    \]
    Plugging in this bound into \pref{eq:count-weight} implies the desired statement.
\end{proof}

As an upshot of \pref{prop:bound-off-paths} and \pref{fact:deg-bound} we have:
\begin{align*}
    \left(A-\frac{d}{n}11^{\top}\right)^{(\ell)} &= A^{(\ell)} + R_{\ell}\\
    A_c^{(\ell)} &= A'^{(\ell)}+R'_{\ell}
\end{align*}
where $R$ and $R'$ are entrywise bounded by $\frac{\ell\cdot(2\log^2 n)^{\ell}}{n}$ with high probability.  Thus,
\begin{align*}
    \left\langle \left(A-\frac{d}{n}11^{\top}\right)^{(\ell)},A_c^{(m)}\right\rangle &= \langle A^{(\ell)}, A'^{(m)}\rangle + \langle A^{(\ell)},R'_m\rangle + \langle A'^{(m)},R_{\ell} \rangle + \langle R_{\ell}, R'_m\rangle.  \numberthis \label{eq:cross-terms-nb}
\end{align*}
The entrywise $\ell_1$ norms of $A^{(\ell)}$ and $A'^{(m)}$ are the total number of nonbacktracking walks of length-$\ell$ and $m$ in $\bG$ and $\bG_{t,r,\eps}$ respectively, which by the degree-bound of $\log^2 n$ and the fact that $\ell, m\le r$, are each at most $n(\log^2 n)^r$.  Since the entries of $R_{\ell}$ and $R'_{m}$ are bounded by $\frac{\ell\cdot(2\log^2 n)^{r}}{n}$, the second and third terms of \pref{eq:cross-terms-nb} are bounded by $r\cdot(2\log^2 n)^{2r}$.  It is easy to see that the fourth term is bounded by $r^2(2\log^2 n)^{2r}$.  It remains to understand the first term.

By \pref{lem:heavy-bound} and our choice of $t$, there is no $(t,\eps)$-heavy vertex in the graph with high probability.  Thus, with high probability the only vertices in $\bG_{t,r,\eps}$ which were truncated are the ones that are part of a cycle of length at most $r$.  To get a high probability bound on the number of such vertices, we observe that the number of cycles of length exactly $k$ passing through a vertex $v$ in the complete graph is at most $n^{k-1}$, whereas the probability of a fixed cycle occuring in $\bG$ is $\left(\frac{d}{n}\right)^{k}$, which implies via a union bound that the probability that $v$ is part of a $k$-cycle is at most $\frac{d^k}{n}$.  Consequently, the probability that $v$ is part of a length-$\le r$ cycle is at most $\frac{rd^r}{n}$, which means the expected number of vertices that are part of a length-$\le r$ cycle is at most $rd^r$.  By Markov's inequality, with high probability the number of of such vertices is bounded by, say, $rd^r\cdot\log n$.

We now use $S$ to denote the set of vertices that are within distance $r-1$ of a truncated.  From the high probability bound on the maximum degree in $G$ of $\log^2 n$, the size of $S$ is, with high probability, at most $r\left(d\log^2 n\right)^{r}$.  For a matrix $L$ and $T,U\subseteq[n]$, let's use $L_{T,U}$ to denote the principal submatrix obtained from the rows indexed by $T$ and columns indexed by $U$.  For all $t\le r/2-1$:
\begin{align*}
    A_{[n]\setminus S,[n]}^{(t)} &= A'^{(t)}_{[n]\setminus S,[n]} \\
    A_{[n],[n]\setminus S}^{(t)} &= A'^{(t)}_{[n],[n]\setminus S}.
\end{align*}
Thus,
\begin{align*}
    \langle A^{(\ell)}, A'^{(m)}\rangle = \langle A_{S,S}^{(\ell)}, A'^{(m)}_{S, S} \rangle + \langle A^{(\ell)}_{S,[n]\setminus S}, A'^{(m)}_{S,[n]\setminus S} \rangle + \langle A^{(\ell)}_{[n]\setminus S, S}, A'^{(m)}_{[n]\setminus S, S} \rangle + \langle A^{(\ell)}_{[n]\setminus S,[n]\setminus S}, A'^{(m)}_{[n]\setminus S,[n]\setminus S} \rangle.
\end{align*}
By \pref{fact:deg-bound} and the bound on $|S|$, the first term is bounded by $r\left(d\log^2 n\right)^{2r}$ with high probability.  If $\ell\ne m$, then each of the second to fourth terms is equal to $0$.  If $\ell = m$, the sum of the second to fourth terms is sandwiched between the total number of length-$\ell$ self-avoiding walks in $\bG$ that also avoid $S$ and the total number of length-$\ell$ self-avoiding walks in $\bG$.  By \pref{fact:deg-bound} and the bound on $|S|$ with high probability the two quantity differ by at most $r\left(d\log^2 n\right)^{2r}$, and by \pref{thm:loc-stat} the latter quantity is $(1\pm o_n(1))d^{\ell}n$ with high probability.  Since $M$ from the statement of \pref{thm:SBM-lower-bound} is a constant, by a union bound, the latter quantity is $(1\pm o_n(1))d^{\ell}n$ with high probability simultaneously for all $\ell\le M$.  Additionally observe that when $\ell\ne m$, the inner product quantity we wish to bound remains bounded as long as the desired bounds on $|S|$ and the maximum degree in the graph hold.   In conclusion, with high probability the following holds simultaneously for \emph{all} $\ell\le M$ and $m\le r/2-1$:
\[
    \left\langle \left(A-\frac{d}{n}11^{\top}\right)^{(\ell)},A_c^{(m)}\right\rangle =
    \begin{cases}
        (1\pm o_n(1))\cdot d^{\ell}n &\text{if $\ell=m$}\\
        O(2r^2(2\log^2 n)^{2r}) &\text{if $\ell\ne m$}.
    \end{cases}
\]
As a consequence:
\[
    \left\langle \left(A-\frac{d}{n}11^{\top}\right)^{(\ell)},Y(z)\right\rangle = (1\pm o_n(1))d^{\ell}z^{\ell}n.
\]
Since $Y(z)$ is PSD with high probability for all $|z|<\frac{1}{(1+\eps)^6\sqrt{d}}$ and the choice of $\eps$ was arbitrary, \pref{thm:SBM-lower-bound} follows.

%% file: Content/recovery.tex
\section{Conjectural recovery in the DRBM} % (fold)
\label{app:recovery}

As discussed in the introduction, this paper will not settle fully the question of recovering the planted communities. However, we can at least reduce some key aspects of this problem to \pref{conj:planted-spec} regarding the spectrum of $A_{\bG}$ when $\bG \sim \Planted_{(d,k,M,\pi)}$.

There are numerous ways to pose the recovery task, and as many metrics of success, but let us set ourselves the modest goal of, given $\bG$ drawn from a planted model with $\lambda_1^2,...,\lambda_\ell^2 > (d-1)^{-1}$ and knowledge of the parameters $(d,k,M,\pi)$, recovering a vector in $\bbR^n$ with constant correlation to each of the vectors $\check \bx_1,...,\check \bx_\ell$ from the \pref{sec:drbm}. If $\ell = k$, we can use this and our knowledge of $M$ to apply the change-of-basis $F^{-1}$ and recover vectors correlated to the indicators $\bx_1,...,\bx_k$ for each of the $k$ communities. 

Our first claim is that, assuming \pref{conj:planted-spec}, the eigenvectors of $A_{\bG}$ can be used to approximate the $\check \bx_i$'s. In \pref{sec:simple-DRBM} we showed that there exists a polynomial $f$ strictly positive on $(-2\sqrt{d-1},2\sqrt{d-1}) \cup \{d\}$ with the property that
$$
	\check x_i^T f(A) \check x_i < -\delta n
$$
for some constant $\delta$. Writing $\mu_1,...,\mu_n$ for the eigenvalues of $A_{\bG}$ and $\Pi_1, \cdots \Pi_n$ for the orthogonal projectors onto their associated eigenspaces, we can expand this as
\begin{align*}
	-\delta n &> \sum_{u \in [n]} f(\mu_u) \check \bx_i^T \Pi_u \check \bx_i \\
	&= \sum_{|\mu_u| < 2\sqrt{d-1}}f(\mu_u) \check \bx_i^T \Pi_u \check \bx_i + \sum_{|\mu_u| \ge 2\sqrt{d-1}}f(\mu_u) \check \bx_i^T \Pi_u \check x_i \\
	&\ge \sum_{|\mu_u| \ge 2\sqrt{d-1}}f(\mu_u) \check \bx_i^T \Pi_u \check \bx_i & & f(x) \text{ positive on $(-2\sqrt{d-1},2\sqrt{d-1})$} \\
	&\ge \inf_{|x| \le d}f(x) \cdot \check \bx_i^T \left(\sum_{|\mu_u| \ge 2\sqrt{d-1}} \Pi_u \right) \check \bx_i.
\end{align*}
Thus, even if there are only constantly many eigenvectors outside the bulk, a (for instance) random vector in their span will have $O(n)$ correlation with each of the $\check \bx_i$'s.

In order to recover \emph{robustly} we will lean on the results of \pref{sec:robust-drbm}. If we begin with $\bG$ from the planted model, perform $\epsilon n$ adversarial edge insertion or deletions, and then run the SDP again, we showed that the old SDP solution will \emph{still} be feasible. Thus, if we take $\check X$ from the SDP run on the corrupted graph, we will still have
\begin{align*}
	-\delta n > \langle f(A_{\bG}), \check X_{i,i} \rangle \ge \inf_{|x| \le d} f(x) \cdot \langle \sum_{|\mu_u| \ge 2\sqrt{d-1}} \Pi_u, \check X_{i,i} \rangle,
\end{align*}
so a, say, Gaussian vector with covariance $\check X_{i,i}$ will have constant correlation with the subspace spanned by the outside-the-bulk eigenvectors of $A_{\bG}$, the adjacency matrix of the \emph{unperturbed} graph, which we showed above have the same correlation guarantee with the $\check \bx_{i}$'s.

% section recovery (end)

%% file: Content/loc-stat-calc.tex
\section{Local Statistics in the DRBM}  \label{app:loc-stat-calc}

In this section we will prove \pref{thm:local-stats} and \pref{prop:paths-suffice}. Since the first posting of this paper, we have updated and streamlined the arguments using the framework developed by one of the authors in \cite{bandeira2019computational}. Several lemmas below have analogues in that work with similar proofs, and we will point these out along the way.

We first prove \pref{thm:local-stats} by computing the quantities $\expected p_{(H,S,\tau)}(\bx,\bG)$ in the planted model. Fix parameters $d,k,M,\pi$, recalling that $M$ is symmetric, and $M\Diag(\pi)$ is stochastic. For any partially labelled graph $(H,S,\tau)$, let $\chi(H) = |V(H)| - |E(H)|$, $c(H)$ denote its number of connected components, and recall 
\begin{align*}
    C_{H}(d) &\triangleq \frac{ \prod_{v \in V(H)} (d)_{\deg(v)}}{d^{|E(H)|}} \\
    L_{(H,S,\tau)}(M,\pi) &\triangleq \sum_{\hat\tau : \hat\tau|_S = \tau} \prod_{v \in V(H)} \pi(\tau(v)) \cdot \prod_{(u,v) \in E(H)} M_{\tau(u),\tau(v)},
\end{align*}
where the latter sum is over extensions $\hat\tau: V(H) \to [k]$ of $\tau$. Note that both quantities are well-defined, by the symmetry relation of $M$ and $\pi$, and that both are multiplicative on disjoint unions. We are aiming to show if $(H,S,\tau)$ has $O(1)$ edges, then with high probability over $(\bx,\bG) \sim \Planted$,
    $$
        p_{(H,S,\tau)}(\bx,\bG) = n^{\chi(H)} L_{(H,S,\tau)}(M,\pi) \cdot C_H(d) \pm o(n^{c(H)}).
    $$

To do so, we will work in the \textit{configuration model} $\hat\Planted$, a distribution over $d$-regular multigraphs which, (i) outputs a simple graph with probability bounded away from zero in $n$, and (ii) conditioned on simple output agrees with $\Planted$. It is routine that high probability statements in the configuration model, like the conclusion of this proposition, therefore transfer to $\Planted$. 
    
To sample a multigraph $\hat\bG$ from $\hat\Planted$, first choose a random $\pi$ balanced labelling $\sigma$, and adorn each vertex $v \in [n]$ with $d$ half-edges $v_1,...,v_d$. Then, for each $i \in [k]$ randomly label the half-edges on the $i$th group $i\to 1,...,i\to k$ so that a $\pi(i)\pi(j)M_{i,j}dn$ have label $i \to j$ for each $j \in [k]$. Finally, choose random perfect matchings joining the $i\to j$ and $j\to i$ edges for each $i,j \in [k]$.
    
\begin{lemma}
    \label{lem:match-prob}
    Condition on a $\pi$-balanced labelling $\sigma : [n] \to [k]$, and let $\bP$ be the random perfect matching on half-edges output by $\hat\Planted$. Let $R$ be a simple partial matching involving constantly many half-edges, and for short write $(u,v) \in R$ if some pair $(u_a,v_b)$ appears in the matching. Then
    $$
        \hat\Planted[R \subset \bP] = (1 \pm o_n(1)) \prod_{(u,v) \in R} \frac{M_{\sigma(u),\sigma(v)}}{d n}.
    $$
\end{lemma}

\begin{proof}
    Throughout this proof, we will call a pair $(u_a,v_b)$ appearing in $R$ an \textit{edge}. Write $S_i$ for the collection of half-edges that adorn the vertices in $\sigma^{-1}(i)$, $U_i$ for the number of half-edges in $R$ that belong to $S_i$, and $U_{i,j}$ the number of edges in $R$ with one half-edge each in $S_i$ and $S_j$, respectively. We have
    $$
        \hat\Planted[R \subset \bP] = \frac{\prod_{i<j} \frac{(\pi(i)\pi(j) M_{i,j}d n)!}{(\pi(i) \pi(j)M_{i,j}d n - U_{i,j})!} \prod_i \frac{(\pi(i)^2 M_{i,i} d n)!}{(\pi(i)^2 M_{i,i}d n - 2 U_{i,i})!} \frac{\pi(i)^2M_{i,i}dn - 2U_{i,i} - 1)!!}{(\pi(i)^2M_{i,i}dn - 1)!!}}{\prod_i \frac{(\pi(i) dn)!}{(\pi(i) dn - U_i)!}}.
    $$
    Up to $o_n(1)$ fluctuations, this is equal to
    $$
        \frac{\prod_{i < j} (\pi(i)\pi(j)M_{i,j}dn)^{U_{i,j}} \prod_i(\pi(i)^2 M_{i,i}dn)^{U_{i,i}}}{\prod_i(\pi(i) dn)^{U_i}}.
    $$
    For each edge $(u_a,v_b) \in R$, the numerator has a term $\pi(\sigma(u))M_{\sigma(u),\sigma(v)}$, and the denominator has two terms $\pi(\sigma(u))dn$ and $\pi(\sigma(v))dn$; the dropped subscripts are intentional here. Since $R$ is simple, we can alternatively account for these terms by looking at pairs $(u,v) \in R$. Thus, again suppressing $o_n(1)$ fluctuations, we can rewrite as
    $$
        \prod_{(u,v) \in R}\frac{\pi(\sigma(u)) \pi(\sigma(v))M_{\sigma(u),\sigma(v)}dn}{\pi(\sigma(u))dn \cdot \pi(\sigma(v))dn} = \prod_{(u,v)\in R} \frac{M_{\sigma(u),\sigma(v)}}{dn}.
    $$
\end{proof}

\begin{proof}[Proof of Theorem \ref{thm:local-stats}]

    Let's begin by computing the expectation of $p_{(H,S,\tau)}(\bx,\hat\bG)$ over $(\bx,\hat\bG)$ sampled from the configuration model. This necessitates that we extend the quantity $p_{(HS,\tau)}(\bx,\hat\bG)$ to the case when $\hat\bG$ is a multigraph---we will simply take it to mean the evaluation of $p_{(H,S,\tau)}$ on the simple graph obtained by removing all self-loops and merging all multi-edges.
    
    Choose an extension $\hat\tau : V(H) \to [k]$ of $\tau$, and an injection $\phi : V(H) \to [n]$ that agrees on labels. The image of each vertex in $V(H)$ has $d$ half-edges, so there are
    $$
        \prod_{v \in V(H)} (d)_{\deg(v)}
    $$
    matchings that ``collapse'' to $H$. For each, Lemma \ref{lem:match-prob} tells us the probability of inclusion in $\hat\bG \sim \hat \Planted$. Finally, there are $\prod_{v \in (H)}(\pi(\tau(v)) n)$ such injective maps $\phi$. Putting this all together, and summing over all extensions $\hat\tau$,
    $$
        \expected p_{(H,S,\tau)}(\bx,\hat\bG) = n^{\chi(H)} L_{(H,S,\tau)}(dM,\pi)\cdot C_H(d) + O(n^{\chi(H) - 1)}).
    $$
    
    If $H$ has at least one cycle, then $c(H) > \chi(H)$, and an application of Markov's inequality finishes the proof. If instead $H$ is a forest, then the assertion will follow from an application of Chebyshev's inequality. In particular, $\expected p_{(H,S,\tau)}(\bx,\hat\bG)^2$ is a sum over pairs of injective maps $\phi_1,\phi_2$ of the probability that both are occurrences. We can think of each pair as a single injective map $\phi' : V(H') \to [n]$, where $(H',S',\tau')$ is the image of the union of the two copies of $(H,S,\tau)$ under $\phi_1,\phi_2$ respectively. In other words,
    $$
        \expected p_{(H,S,\tau)}(\bx,\hat\bG)^2 = \sum_{H'} p_{(H',S',\tau')}(\bx,\hat\bG),
    $$
    where the sum is over all $(H',S',\tau')$ that can arise by identifying some pairs of vertices in two copies of $(H,S,\tau)$. Since $H$ has no cycles, $\chi(H') \le 2\chi(H)$, with equality only if $H' = H \sqcup H$. Thus, as $L_{(H,S,\tau)}$ and $C_H$ are multiplicative on disjoint unions,
    $$
        \expected p_{(H,S,\tau)}(\bx,\hat\bG)^2 = \left(\expected p_{(H,S,\tau)}(\bx,\hat\bG)\right)^2 + O(n^{2\chi(H) - 1}).
    $$
    We finally apply Chebyshev and note that $c(H) = \chi(H)$ for forests.
\end{proof}

It remains now to prove \pref{prop:paths-suffice}, the content of which is that in constructing a feasible pseudoexpectation in the planted model, it suffices to check only certain moment constraints. We first show that the moment constraints involving partially labelled subgraphs which contain a cycle are automatically satisfied. The argument below is essentially identical to \cite[Lemma 5.19]{bandeira2020spectral}

\begin{lemma} \label{lem:cycles-for-free}
    Let $\bG \sim \Null$, and assume that $\pseudo$ is a degree-$D_x$ pseudoexpectation---perhaps dependent on $\bG$---satisfying $\calB_k$. For every $\delta > 0$, with high probability 
    $$
        \pseudo p_{(H,S,\tau)}(x,\bG) = \expected_{(\bx,\bG)\sim \Planted} p_{(H,S,\tau)}(\bx,\bG) \pm \delta n^{c(H)}
    $$
    for every $(H,S,\tau)$ with constantly many edges and containing a cycle.
\end{lemma}

\begin{proof}
    Using Cauchy-Schwartz for pseudoexpectations, for every multilinear monomial $m(x)$ we have $(\pseudo m(x))^2 \le \pseudo m(x)^2 = \pseudo m(x)$, since $\pseudo x_{u,i}^2 = \pseudo x_{u,i}$; thus $\pseudo m(x) \in [0,1]$. In other words, for any $(H,S,\tau)$, we have
    $$
        \left| \pseudo p_{(H,S,\tau)}(x,\bG)\right| = \left|\sum_{\phi:V(H) \inj [n]}\prod_{(u,v) \in E(H)} \bG_{\phi(u),\phi(v)} \prod_{u \in S} x_{\phi(u),\tau(u)}\right| \le \left| \sum_{\phi : V(H) \inj [n]} \prod_{(u,v) \in E(H)}\bG_{\phi(u),\phi(v)}\right|.
    $$
    The latter is the number of occurrences of $H$ in $\bG$, with both regarded as unlabelled graphs; from the proof of \pref{thm:local-stats} above, if $H$ has a cycle, then this is $o(n^{c(H)}$. Thus with high probability
    $$
        \pseudo p_{(H,S,\tau)}(x,\bG) = o(n^{c(H)}) = n^{\chi(H)}L_{(H,S,\tau)}(M,\pi) C_H(d) \pm \delta n^{(c(H)}
    $$
    for any $\delta > 0$, as $\chi(H) < c(H)$.
\end{proof}

This lemma leaves us to check only the partially labelled trees, and we next show that in fact it suffices to verify only a subset of these. The following appeared as Definition 5.20 in \cite{bandeira2020spectral}.

\begin{definition}
    The \textit{pruning} of a partially labelled tree $(H,S,\tau)$ is the unique maximal subtree with the property that all leaves are distinguished vertices; if $(H,S,\tau)$ is unlabelled, its pruning is empty. The pruning of a forest is obtained by taking the pruning of each tree.
\end{definition}

\begin{lemma}
    Let $(H,S,\tau)$ be a partially labelled, tree, $(\tilde H,S,\tau)$ its pruning. Then
    \begin{align*}
        L_{(H,S,\tau)}(M,\pi) &= L_{(\tilde H,S,\tau)}(M,\pi)
    \end{align*}
\end{lemma}

\begin{proof}
    We'll argue inductively that one can delete any unlabelled leaf without affecting $L_{(H,S,\tau)}$. Let $v$ be such a leaf, $w$ its parent, and $(H',S,\tau)$ be obtained by deleting $v$. Then
    \begin{align*}
        L_{(H,S,\tau)}(M,\pi) = L_{(H',S,\tau)}(M,\pi)\sum_{\ell \in [k]}M_{\tau(w),\ell}\pi(\ell) = L_{(H',S,\tau)}(M,\pi),
    \end{align*}
    as $M\Diag(\pi)$ is Stochastic.
\end{proof}

\begin{lemma}
    Let $\bG \sim \Null$, and let $(H,S,\tau)$ and $(\tilde H, S, \tilde \tau)$ be a partially labelled forest and its pruning, respectively. Then, with high probability,
    $$
        \left\| p_{(H,S,\tau)}(x,\bG) - n^{c(H) - c(\tilde H)} \frac{ C_H(d)}{C_{\tilde H}(d)} p_{(\tilde H,S, \tilde \tau)}(x,\bG) \right\|_1 = o(n^{c(H)}),
    $$
    where by $\|\cdot\|_1$ we mean the $L_1$ norm of the coefficients.
\end{lemma}

\begin{proof}
    This argument adapted with minor elaboration from \cite[Lemma 5.21]{bandeira2020spectral}. For each occurrence $\tilde \phi : V(\tilde H) \inj [n]$ of $(\tilde H,S,\tilde \tau)$, call an occurrence $\phi : V(H) \inj [n]$ of $H$ an extension of $\tilde \phi$ if they agree on $V(\tilde H) \subset V(H)$. Write $\tilde \Phi$ for the set of occurrences of the pruning, and for each $\tilde \phi \in \tilde \phi$, write $\Phi(\tilde \phi)$ for its set of extensions. 
    
    Again using the fact that each multilinear monomial has $\pseudo m(x) \in [0,1]$, we may write
    \begin{align*}
        \left\| p_{(H,S,\tau)}(x,\bG) - n^{c(H) - c(\tilde H)} \frac{ C_H(d)}{C_{\tilde H}(d)} p_{(\tilde H,S, \tilde \tau)}(x,\bG) \right\|_1
        \le \sum_{\tilde \phi \in \tilde \Phi}
        \left| |\Phi(\tilde \phi)| - n^{c(H) - c(\tilde H)} \frac{C_H(d)}{C_{\tilde H}(d)}\right|
    \end{align*}
    Only $o(n^{c(\tilde H)})$ occurrences in this sum have the property that their $|E(H)|$ neighborhoods in $\bG$ contain a cycle, so to prove the assertion in the lemma we are free to ignore these terms entirely. For each remaining occurrence $\tilde \phi$, and each connected component $\tilde J$ of $\tilde H$ containing a distinguished vertex, and the corresponding component $J$ of $H$, there are precisely
    $$
        \prod_{v \in V(H)} \prod_{q = \deg_{\tilde J}(v)}^{\deg_J(v) - 1} (d - q) = \frac{C_J(d)}{C_{\tilde J}(d)}
    $$
    ways to extend it to an occurrence of $J$. To finish counting the number of extensions of $\tilde \phi$, we need to choose an occurrence of $K$, the disjoint union of every connected component of $H$ which contains no distinguished vertex, that does not interact with $\tilde \phi(\tilde H)$ or the already-chosen extensions of the $\tilde J$'s. But, there are
    $$
        n^{c(K)}\prod_{v \in V(K)}\prod_{q = 0}^{\deg_K(v) - 1}(d-q) + o(n^{c(K)}) = n^{c(H) - c(\tilde H)} C_K(d) + o(n^{c(H) - c(\tilde H)})
    $$
    ways to do this. Thus, using multiplicativity of $C_H(d)$ on disjoint unions,
    $$
        |\Phi(\tilde \phi)| = n^{c(H) - c(\tilde H)}C_H(d) + o(n^{c(H) - c(\tilde H)}),
    $$
    there are $O(n^{c(\tilde H)}$ possible occurrences of $\tilde H$, and we are done.
\end{proof}

Taking a union bound and applying the above lemma, we immediately obtain:

\begin{lemma} \label{lem:pruned-suffice}
    Let $\bG \sim \Null$, and assume that $\pseudo$ is is a degree-$D_x$ pseudoexpectation, which may depend on $\bG$. If $\pseudo$ satisfies the affine moment constraints for every pruned, partially labelled forest with at most $D_G$ edges and $D_x$ distinguished vertices, then with high probability
    $$
        \left|\pseudo p_{(H,S,\tau)}(x,\bG) - \expected_{(\bx,\bG)\sim \Planted}p_{(H,S,\tau)}(\bx,\bG)\right| = o(n^{\chi(H)})
    $$
    for every partially labelled forest $(H,S,\tau)$ with at most $D_G$ edges and $D_x$ distinguished vertices.
\end{lemma}

\begin{proof}
    For each $(H,S,\tau)$, let $(\tilde H, S, \tilde \tau)$ be its pruning. Recalling again that each monomial has pseudoexpectation in the interval $[0,1]$, we have, with high probability,
    \begin{align*}
        \pseudo p_{(H,S,\tau)}(x,\bG) = n^{c(H) - c(\tilde H)}\frac{C_H(d)}{C_{\tilde H(d)}} p_{(\tilde H,S,\tilde \tau)}(x,\bG) + o(n^{c(H)}) = n^{c(H)}L_{(H,S,\tau)}(M,\pi) C_H(d) + o(n^{\chi(H)}).
    \end{align*}
    Taking a union bound over the finitely many extensions of the finitely many pruned partially labelled forests, we are done.
\end{proof}

To prove \pref{prop:paths-suffice}, we need to specialize this result to the case of pruned partially labelled forests with at most two distinguished vertices. These are exactly the paths with labelled endpoints, a pair of labelled vertices, and a single labelled vertex. Recalling the notation of $X \in \bbR^{nk \times nk}$ as the matrix $X_{(u,i),(v,j)} = \pseudo x_{u,i}x_{v,j}$ and $l \in \bbR^{nk}$ as the vector with $l_{(u,i)} = \pseudo x_{u,i}$, our argument in \pref{lem:sdp-affine} may be rephrased to say that the moment constraints on $\pseudo$ for the first two cases at any error tolerance $\delta' > \delta$ are implied by
\begin{align*}
    \langle X_{i,j}, \nb{s}{\bG}\rangle &= \pi(i) T_{i,j}^s \|q_s\|^2_{\km}n \pm \delta n \\
    \langle X_{i,j},\bbJ \rangle = \pi(i)\pi(j) n^2 \pm \delta n^2.
\end{align*}
The third case, of a single labelled vertex, is implied by
$$
    \langle l_i,e \rangle = \pi(i)n \pm \delta n
$$
\pref{prop:paths-suffice} now follows from \pref{lem:cycles-for-free} and \pref{lem:pruned-suffice}.

%% file: Content/singleton.tex
\section{Bounding Singleton Expectation}

Let 
\begin{align*}
\zeta_{W}(A) \defeq \prod_{i \in V(W)}    \beta_i(A) \cdot \gamma_i(A) 
\end{align*}
where $\gamma_i,\beta_i$ are boolean functions given by,
\begin{align*}
\beta_i(A) & \defeq 1[ i \text{ is } (t,d') \text{ bounded in }A] \\
\gamma_i(A) & \defeq 1[i \text{ is NOT in a cycle of length} \leq \girth \text{ in } A]  
\end{align*}
This section is devoted to showing the following bound on the expectation of products involving singleton edges $S(W)$.
\begin{theorem} \label{thm:expectation-estimate}
For every $d' > d > 1$ and $\delta \in (0,1)$, the following holds for all sufficiently large $t$.
Suppose $S(W)$ is the singleton edges and $J \subseteq D(W)$ a set of duplicative edges in a $(k,\ell)$-linkage $W$, and $g\le\frac{\log n}{\log \log n}$ we have,
\begin{equation} \label{eq:myeq20}
\left|\bbE\left[\prod_{ij \in S(W)} \left(A_{ij} - \dovern\right) \cdot  \prod_{ij \in J} A_{ij} \cdot \zeta_{W}(A) \right]\right| \leq C \log^2 n \cdot \left(\dovern\right)^{|S(W) \cup J|}   \cdot n^{0.8 e(W)} \cdot 4^{|S(W)|} \delta^{|S(W)|-24kt}
\end{equation}
for some absolute constant $C$.  Here $\excess(W)$ is the excess edges in the walk defined as $e(W) = |E(W)| - |V(W)| + 1$.
\end{theorem}

We wish to emphasize that the key aspect of \pref{eq:myeq20} is the term $\delta^{|S(W)|}$, showing that the expectation decays exponentially in $|S(W)|$.
\begin{proof}
Henceforth in this section, We will use $S$ to denote $S(W)$.  We begin the proof of the theorem by expanding out the expectation in \pref{eq:myeq20}. 
\begin{align*}
&\bbE\left[\prod_{ij \in S} \left(A_{ij} - \dovern\right) \cdot  \prod_{ij \in J} A_{ij} \cdot \zeta_{W}(A) \right]\\
& = \sum_{\alpha \in \{0,1\}^{S} } \Pr[A_{S} = \alpha] \bbE\left[\left(1-\dovern\right)^{|\alpha|} \left(-\dovern\right)^{|S| - |\alpha|} \cdot \prod_{ij \in J} A_{ij} \cdot \zeta_W(A) \right] \mper
\end{align*}
Using $\Pr[A_{S} = \alpha] = \left(\dovern\right)^{|\alpha|} \cdot \left(1-\dovern\right)^{|S| - |\alpha|}$, we can simplify the above expression to,
\begin{align*}
    = \left(1-\dovern\right)^{|S|} \cdot \left(\dovern\right)^{S \cup J} \cdot \sum_{\alpha \in \{0,1\}^{S}} (-1)^\alpha \Exp_{A^{c}} \left[ \zeta_W\left(A^{c}, A_{S} = \alpha,A_J = 1\right) \right] \mcom
\end{align*}
where $A^{c}$ denotes the random variables $\{A_{ij} | ij \in \overline{ S \cup J} \}$, each of which is an independent Bernoulli random variable with expectation $\dovern$.
We will now select a subset of edges $Q \subseteq S(W)$ such that the following two conditions hold:
\begin{enumerate}
    \item Edges in $Q$ are far from each other in the graph $G(W)$.  Formally, for all $ij, i'j' \in Q$, $$\dist_{G(W)}(i,i') \geq 4t \ .$$
    \item Neighborhoods of each of the edges in $Q$ have a small number of vertices.  Specifically, for all $ij \in Q$, $|\sfB_{2t}(i,G_0)| \leq 2t+2$.
\end{enumerate}
We will show in \pref{lem:farawaysingletons} that there exists such a set $Q$ with $|Q| \geq |S(W)|/8t - 3 k - 6 e(W)$.

Let $R \defeq S(W) \setminus Q$.  Let $\alpha = \alpha_Q \cup \alpha_R$ where $\alpha_Q \in \{0,1\}^Q$ and $\alpha_R \in \{0,1\}^R$.  We can upper bound the above term by,
\begin{align}\label{eq:myeq1}
\leq \left(\dovern\right)^{S \cup J} 2^{|R|} \cdot \max_{\alpha_R \in \{0,1\}^R} \left| \sum_{\alpha_Q \in \{0,1\}^{Q}} (-1)^{|\alpha_Q|} \Exp_{A^{c}} \left[ \zeta_W\left(A^{c}, A_{Q} = \alpha_Q,A_R = \alpha_R ,A_J = 1\right) \right]  \right|
\end{align}
For any fixed choice of $A^{c}$, let $\zeta_{A^{c},\alpha_R} : \{0,1\}^Q \to \{0,1\}$ denote the function,
\[\zeta_{A^c,\alpha_R}(z) \defeq \zeta_W\left(A^{c}, A_{Q} = z,A_R = \alpha_R ,A_J = 1\right) \mper\]
% \pnote{I want to use something like $(*1*)$ to carry denote what the LHS is, need to figure out the right latex way to do it}
Rewriting the LHS of \eqref{eq:myeq1} in terms of $\zeta_{A^{c},\alpha_R}$,
\[ 
\leq \left(\dovern\right)^{S \cup J} 2^{|R|} \cdot \max_{\alpha_R \in \{0,1\}^R} \left| \Exp_{A^{c}}\left[   \sum_{\alpha_Q \in \{0,1\}^{Q}} (-1)^{|\alpha_Q|} \cdot \zeta_{A^c, \alpha_R}(\alpha_Q) \right]  \right|
\]
Observe that for any function $\psi : \{0,1\}^Q \to \{0,1\}$, $\sum_{z \in \{0,1\}^Q} (-1)^{|z|} \psi(z) = 0$ if $\psi$ is independent of any bit in $z$.  Otherwise, the sum is upper bounded by $2^{|Q|}$.  Therefore, we can rewrite the above bound as,
\begin{equation}
\label{eq:myeq2}
\leq \left(\dovern\right)^{S \cup J} 2^{|Q \cup R|} \cdot \max_{\alpha_R \in \{0,1\}^R} \left( \Pr_{A^{c}}\left[    \zeta_{A^c, \alpha_R} \text{ depends on all bits in } Q \right]  \right)
\end{equation}
Recall that $\zeta_W(A) = \beta_W(A) \cdot \gamma_W(A) $ where $\beta_W(A) = \prod_{i \in W} \beta_i(A)$ and $\gamma_W(A) = \prod_{i \in W} \gamma_i(A)$.  

Analogous to the definition of $\zeta_{A^C, \alpha_R}$, define corresponding boolean functions $\beta_{A^C,\alpha_R}$ and $\gamma_{A^C,\alpha_R}$ over $\{0,1\}^Q$, i.e.,
\[\beta_{A^c,\alpha_R}(z) \defeq \beta_W\left(A^{c}, A_{Q} = z,A_R = \alpha_R ,A_J = 1\right) \mper\]
\[\gamma_{A^c,\alpha_R}(z) \defeq \gamma_W\left(A^{c}, A_{Q} = z,A_R = \alpha_R ,A_J = 1\right) \mper\]
By a simple union bound, we can write
\begin{align}
&\Pr_{A^{c}}\left[    \zeta_{A^c, \alpha_R} \text{ depends on all bits in } Q \right]  \nonumber \\
& \leq
\sum_{Q' \subset Q} \Pr_{A^{c}}\left[    \beta_{A^c, \alpha_R} \text{ depends on all bits in } Q' \bigwedge \gamma_{A^c, \alpha_R} \text{ depends on all bits in } Q\setminus Q'  \right] \nonumber\\
& \leq 
\sum_{Q' \subset Q} \min \left(\Pr_{A^{c}}\left[    \beta_{A^c, \alpha_R} \text{ depends on all bits in } Q'\right], \Pr_{A^{c}} \left[ \gamma_{A^c, \alpha_R} \text{ depends on all bits in } Q\setminus Q'  \right]\right) \label{eq:myeq21}
\end{align}
We will the probabilities in the above sum in \pref{claim:prob1} and \pref{claim:prob2} respectively.  Substituting these bound on probabilities into \pref{eq:myeq21}

\begin{align}
& \leq 
\sum_{Q' \subset Q} \min \left(\delta^{16t |Q'|}, C (\log^2 n) n^{-0.7(|Q\setminus Q'|/\girth - \#_c(Q \cup R \cup J)} \right) \\
& \leq  C (\log^2 n) n^{0.7 \#_c(Q \cup R \cup J)} \sum_{Q' \subset Q}  \cdot \min((\delta^{16t})^{|Q'|}, (n^{-0.7/g})^{|Q \setminus Q'|})\\
& \leq C (\log^2 n) n^{0.7 e(W)} \cdot 2^{|Q|} \cdot (\delta^{16t})^{|Q|/2}\\
& \leq C (\log^2 n) n^{0.7 e(W)} \cdot 2^{|S(W)|} \cdot  (\delta)^{|S(W)|-24k t - 48 e(W)} \\
& \leq  C\log^2 n \cdot n^{0.8 e(W)} \cdot 2^{|S(W)|} \cdot \delta^{|S(W)|-24kt}
\end{align}
Substituting back in \pref{eq:myeq2} we get the bound in the theorem.

\end{proof}

\begin{lemma} \label{lem:farawaysingletons}\pnote{move it to section on walks}
For all $t < \ell$, in a $k \times \ell$-linkage there exists $Q \subset S(W)$ with $|Q| \geq \frac{|S(W)|}{8t} - 3 \ell -6e(W)$ such that, 
\begin{enumerate}
    \item For all $ij, i'j' \in Q$, $\dist_{G(W)}(i,i') \geq 4t$.
    \item For all $ij \in Q$, $|\sfB_{2t}(i,G_0)| \leq 2t+2$.
\end{enumerate}
\end{lemma}
\begin{proof}
All the steps of the walk are divided into consecutive segments of singleton edges (``singleton stretches")  and duplicative edges (``duplicative stretch").

The walk can step from a singleton stretch  into a duplicative stretch, either by a turn-around or at an edge that creates a cycle.  \pnote{add more explanation here}  The number of such transitions is therefore at most $\ell + 2e(W)$ where $2 e(W)$ is the number of excess edges.

Hence there are $|S(W)|$ singletons split into $\ell + 2e(W)$ disjoint path segments.  Given a path of length $\Delta$, delete segments of length $8t$ from both end, and then pick edges at a regular intervals of length $8t$ from the each other in the remaining. This yields $\lfloor \frac{(\Delta - 16t)}{8t} \rfloor$ edges which are pairwise $8t$ away, and the $2t$ neighborhood around each of them is a path and thus has only $2t+2$ edges.  Perform this operation on each of the singleton segments to select a subset $Q$ of singleton edges.
By construction edges in $Q$ satisfy the conditions of the Lemma above.  It remains to lower bound the size of $Q$.
If $\Delta_1,\ldots,\Delta_q$ are the lengths of the singleton stretches, we can write 
\begin{align*}
 |Q| \geq \sum_{i = 1}^q\left \lfloor \frac{(\Delta_i - 16t)}{8t} \right\rfloor & \geq \sum_{i = 1}^q  \frac{(\Delta_i - 16t)}{8t} - 1 \\
& \geq \frac{|S(W)|}{8t} - 3q \geq \frac{|S(W)|}{8t} - 3\ell - 6 e(W) 
\end{align*}
edges.
\end{proof}

\subsection{Away from short cycles}

\begin{claim} \label{claim:prob2}
For any subset $Q^* \subset Q$,
\[ \Pr_{A^{c}}\left[    \gamma_{A^c, \alpha_R} \text{ depends on all bits in } Q^*\right] \leq C (\log^2 n) n^{-0.7(|Q^*|/\girth - \numCycles(Q \cup R \cup J)}\]
\end{claim}
\begin{proof}
The function $\gamma_{A^c, \alpha_R}(z) = \gamma_W(A^c, A_Q = z, A_R = \alpha_R, A_J = 1)$ is a anti-monotone \pnote{better terminology?} function of $z$.  

For every pair $ij \in Q^*$, since $\gamma_{A^c, \alpha_R}$ depends on $z_{ij}$ there is some setting of $z_{Q \setminus \{ij\}}$ such that $\gamma_{A^c,\alpha_R}(z_{ij}=0,z_{Q \setminus \{ij\}}) = 1$ but $\gamma_{A^c,\alpha_R}(z_{ij}=1,z_{Q \setminus \{ij\}}) = 0$.
By definition of $\gamma_W$, this implies that addition of edge ${ij}$ creates a cycle of length at most $\girth$.    

Therefore, in the graph given by $A' = (A^c, A_R = \alpha_R, A_J = 1, A_{Q} = 1)$ every edge $ij \in Q^*$ is in a cycle of length at most $\girth$.  There are at least $|Q^*|/\girth$ cycles in the graph $A'$, and at least $|Q^*|/\girth - \numCycles(Q \cup R \cup J)$ involve edges of the random graph $A^c$.

Now we appeal to Lemma A.3 \pnote{state it here for convenience, when we get the time} in \cite{FM17} to conclude the claim.
\end{proof}

\subsection{Heavy Vertices}
The goal of this section is to prove the following claim, a component in the proof of \pref{thm:expectation-estimate}.
\begin{claim} \label{claim:prob1}
Given $d' > d > 1$ and $\delta > 0$, for all sufficiently large value of $t$ the following holds for every subset $Q^* \subset Q$,
\[ \Pr_{A^{c}}\left[    \beta_{A^c, \alpha_R} \text{ depends on all bits in } Q^*\right] \leq (\delta)^{t |Q^*|}\]
\end{claim}

First, let us setup some notation.  
For a graph $G$ and a set of vertices $S$, we make the following definitions.
\[\sfB_{r}(S,G) \defeq \{ i | \dist_{G}(i,S) \leq r \} \]
\[\sfN_{r}(S,G) \defeq \{ i | \dist_{G}(i,S) = r\} \]
Here $\dist_G$ refers to the shortest path distance on the graph $G$.
We borrow the following tail bound on the sizes of neighborhoods in $\sfG(n,\dovern)$ from \cite{FM17}.
\begin{lemma}
Fix $d > 1$ and consider the Erdos-Renyi\pnote{fix this} graph $G \sim \sfG(n,\dovern)$.  Then there exists $C,c > 0$ such that for any $s \geq 0$, $t \geq 1$ and $v \in [n]$
\[\Pr\left[|\sfB_{t}(v;G)| \geq s d^t\right] \leq Ce^{-cs}\]
\end{lemma}

Critical to the proof of \pref{claim:prob1} is the notion of being heavy vertex, and close-to heavy vertices.   A heavy vertex is any vertex with $\sfB_{t}(v,G) \geq (d')^t$.
A vertex is marked as close-to-heavy if it is within distance $t$ of a heavy vertex.  Formally, we have the following definition

\begin{definition}
A vertex $v$ in a graph $G = (V,E)$ is {\it $(t,d')$-close to heavy} if there exists $v'$ such that $\dist_G(v,v')\leq t$ such that $|\sfB_t(v';G)| > (d')^t$.
\end{definition}

First, we will bound the probability that a vertex in an Erd\H{o}s-R\'enyi graph is $(t,d')$-close to heavy.

\begin{lemma} \label{lem:mylemma1}
There exists absolute constant $C,c$ such that for all $t$ and $d' > d$, for a graph $G \sim \sfG(n, \dovern)$ and a vertex $v$,
\[ \Pr_{G}\left[v \text{ is } (t,d')\text{-close to heavy}\right] \leq C d^t    e^{-c(d'/d)^t} \] 
\end{lemma}
\begin{proof}
Let $X$ be the random variable denoting the number of vertices that are $(t,d')$-close to heavy in a graph $G \sim \sfG(n,\dovern)$.
Clearly the above probability is given by $\frac{1}{n} \bbE[X]$.
Suppose a vertex $v$ has $|\sfB_{t}(v;G)| = \gamma (d')^t$ for some $\gamma > 1$. %
Then every vertex $u \in \sfB_{t}(v;G)$ is $(t,d')$-close to heavy.
Therefore, we can upper bound the expected number of vertices that are $(t,d')$-close to heavy in a graph $G \sim \sfG(n,\dovern)$ by,
\begin{align*}
\frac{1}{n} \E[X] & \leq \int_{\gamma = 1}^{\infty} \Pr[|\sfB_{t}(v;G)| = \gamma (d')^t] \cdot (\gamma(d')^t) d\gamma \\
& \leq \int_{s = (d'/d)^t}^{\infty} \Pr[|\sfB_{t}(v;G)| = s (d)^t] \cdot (s d^t) \cdot \left( \frac{d^t}{(d')^t} ds \right) \\
& \leq \frac{d^{2t}}{(d')^t} \int_{s = (d'/d)^t}^{\infty} Ce^{-cs} s ds  \leq \frac{C}{c^2} \cdot \frac{d^{2t}}{(d')^t} \cdot \left[ - e^{-z} z-e^{-z}\right]_{c (d'/d)^t }^{\infty} < C' d^t e^{-c (d'/d)^t}
\end{align*}
where the last inequality holds whenever $(d'/d) > 1$ and $C' > \frac{C(1+c)}{c^2}$ .
\end{proof}

The following Lemma upper bounds the probability of a vertex $v$ being close to heavy in a more complicated setup.  Here a subgraph $G' = (V',E')$ is chosen to be included in the graph, and a set of   vertices $\calF$ are forbidden in the neighborhood of $v$.

\begin{lemma} \label{lem:typical-nbd-sample}
There exists absolute constants $C,c$ such that the following holds for all $t$ and $d' > d > 1$.

Suppose $G' = (V',E')$ be a subgraph of the complete graph and let $v \in V'$ be a vertex in $G'$.

Let $\calF \subseteq [n]$ be a set of vertices disjoint from $V'$, i.e., $\calF \cap V' = \emptyset$.
Suppose we draw $G^c \sim \sfG(n,\dovern)$ and set $G = G' \cup G^c$ then,

\[ \Pr_{G}\left[v \text{ is } (t,d')\text{-close to heavy in } G | \sfB_{2t}(v,G) \cap \calF = \emptyset \right] \leq C | V'| d^t e^{-c \frac{1}{|V'|+1} \cdot \left(\frac{d'}{d}\right)^t} \] 

\end{lemma}
\begin{proof}
Notice that the indicator of the event 
\[\calE_1 = 1[v \text{ is } (t,d')\text{-close to heavy in } G]\]
is a monotone function of the edges $G^c$.
On the other hand, the event $\calE_2 = 1[\sfB_{4t}(v,G) \cap \calF = \emptyset]$
is an anti-monotone function.

By FKG inequality, the two events are negatively correlated and therefore conditioning on $\calE_2$ reduces the chance of $\calE_1$, i.e.,
\[ \Pr_{G}\left[v \text{ is } (t,d')\text{-close to heavy in } G | \sfB_{4t}(v,G) \cap \calF = \emptyset \right] \leq \Pr_{G}\left[v \text{ is } (t,d')\text{-close to heavy in }G \right]  \ . \] 

Now we make the following claim which will prove subsequently.
\begin{claim} \label{claim:123} Let $\rho \defeq \left(\frac{1}{|V'|+1}\right)^{1/t}$.  If no vertex $w \in V'$ is $(t,\rho d')$-close to heavy in $G^c$, then $v$ is not $(t,d')$-close to heavy in $G$.
\end{claim}

Assuming the above claim, we can use the union bound to argue
\begin{align*}
\Pr_{G}\left[v \text{ is } (t,d')\text{-close to heavy in } G \right] & \leq 
\Pr_{G^c}\left[\exists u \in V' \text{ which is } (t,\rho d')\text{-close to heavy in } G^c \right] \\
& \leq \sum_{u \in V'} \Pr_{G^c}\left[u \text{ is } (t,\rho d')\text{-close to heavy in } G^c \right] \\
& \leq C | V'| d^t e^{-c \frac{1}{|V'|+1} \cdot \left(\frac{d'}{d}\right)^t}
\end{align*}
where the last inequality follows from \pref{lem:mylemma1}
\end{proof}
Now we return to proving \pref{claim:123}.
\begin{proof} (Proof of \pref{claim:123})
Suppose $v \in V'$ is $(t,d')$-close to heavy in $G$, and let $u \in [n]$ be the heavy vertex with $\dist_{G}(u,v) \leq t$.
%

%The path from $v \to u$ is either completely contained in $G^c$ in which case $\dist_{G^c}(u,u') \leq t$ or  the path from $u \to u'$ uses edges in $G'$ which implies that $u' \in \sfB_t(V',G^c)$.

Now we will lower bound $|\sfB_t(u,G^c)|$.  To this end, consider any $u' \in [n]$ with $\dist_{G}(u,u') \leq t$. The path from $u \to u'$ is either completely contained in $G^c$ in which case $\dist_{G^c}(u,u') \leq t$ or  the path from $u \to u'$ uses edges in $G'$ which implies that $u' \in \sfB_t(V',G^c)$.
Therefore, we can write 
\[ |\sfB_{t}(u,G)| \leq |\sfB_t(u,G^c)| + |\sfB_t(V',G^c)|. \]
Since $u$ is $(t,d')$-heavy, $|\sfB_{t}(u,G)| \geq (d')^t$.
If no vertex $w \in V'$ is  $(t, \rho d')$-close to heavy in $G^c$, then
\begin{align*}
    |\sfB_t(V',G^c)| & \leq \sum_{w \in V'} |\sfB_t(w,G^c)| 
    \leq |V'| \cdot \left(\rho d' \right)^t
\end{align*}
One can thus conclude that,
\[ |\sfB_t(u,G^c)| \geq (d')^t \left(1 - \rho^t |V'|\right) \geq (\rho d')^t  \ .\]
Finally since $\dist_{G}(v,u) \leq t$, there exists some vertex $w \in V'$ such that $\dist_{G^c}(w,u) \leq t$.  Thus $w$ is $(t,\rho d')$-close to heavy in $G^c$.
\end{proof}
%

%Consider any vertex $w \in [n]$ with distance 

%Observe that 
%\begin{align*}
% \Pr_{G}\left[v \text{ is } (t,d')\text{-close to heavy in } G \right] & \leq \Pr_{G^c}\left[\exists w \in V' which \text{ is } (t,d')\text{-close to heavy in }G^c \right] \\
 %& \leq \sum_{w \in V'}  \Pr_{G^c}\left[\exists w \in V' which \text{ is } (t,d')\text{-close to heavy in } G^c \right] \\
 %\end{align*}

\begin{lemma} \label{lem:all-close-to-heavy}
For every $d' > d$ and $\delta > 0$, there exists $t$ such that the following holds.
Fix  a subset $V_0 \subset [n]$ of vertices and a graph $G_0 = (V_0,E_0)$ with at most $|E_0| < \log^2 n$ edges.  
Suppose $V^* \subset V_0$ be such that,  
\begin{enumerate}
\item For every vertex $i \in V^*$,  $|\sfB_{2t}(i;G_0)| < t^2$.

%\item No vertex $i \in V^*$ is part of a cycle smaller than $4t$.

\item $\dist_{G_0}(i,j) \geq 4t$ for all $i,j \in V^*$.
\end{enumerate}
Then if we sample a graph $G$ by including each of the remaining edges $\binom{n}{2}-E_0$ independently with probability $\dovern$,
\[ \Pr\left[\forall v \in V^*, v \text{ is } (t,d')\text{-close to heavy in } G\right] \leq \delta^{t |V^*|} \]
\end{lemma}
\begin{proof}
Let $G^c = ([n],E^c)$ denote the graph consisting of edges in $\binom{n}{2} - E_0$ each of which is included independently with probability $\dovern$.

Consider the neighborhood $\sfB_{2t}(v;G)$ around a vertex $v \in V^*$.
Clearly, the neighborhood contains the sub-graph $\sfB_{2t}(v,G_0)$ since $G_0 \subset G$.
All the additional vertices (and edges) in $\sfB_{2t}(v;G)$ are those reachable by taking the newly sampled edges in $G^c$.

Intuitively, up to constant distances, the graph $G^c$ will be ``tree-like".  More specifically, for a typical sample, one would expect that the neighborhood can be decomposed as,

\[ \sfB_{2t}(v,G) = \sfB_{2t}(v,G_0) \cup \bigcup_{w \in \sfB_{2t}(v,G_0)} \sfT_w\]

where $\sfT_w$ is a tree with vertex $w$ as root, and no other vertices in $V_0$.
%Let $\sfB_{4t}(V_0;G)$ denote the set of vertices within distance $4t$ from $V_0$ in graph $G$.
%
Call a vertex $v \in V^*$ to be {\it typical } if the above assumptions hold.

We will first show that there is a significant fraction of vertices in $|V^*|$ are {\it typical} with all but negligible probability.

To this end, consider the graph $\calH$ formed by the edges in
\[  E[\sfB_{2t}(V^*;G)] -  E[\sfB_{2t}(V^*;G_0)] \ ,\]
where $E[\calS]$ denotes the set of edges contained in a set of vertices $\calS$.

Consider a vertex $v \in V^*$.  For every vertex $w \in \sfB_{2t}(v; G_0)$ and $\dist(w,v) = d$, the graph $\calH$ contains the subgraph $\sfB_{2t - d}(w,G) - \sfB_{2t-d}(w,G_0)$.  In fact, in a typical vertex $v \in V^*$, this would be a tree of depth $2t-d$ with vertex $w$ as root.

\begin{claim}
The number of typical vertices is at least $|V^*| - 2s$ where
$s \defeq \#_c(\sfB_{2t}(V_0;G))- \#_c(G_0)$.
\end{claim}
\begin{proof}
Consider the execution of a depth-first-traversal on the graph $\calH$.  More precisely, consider the execution of the following algorithm:

\begin{itemize}
        \item ExploreGraph()
        \begin{itemize}
            \item Set $visited[w] = false$ for all $w \in \calH \cup V_0$
            \item For each vertex $w \in   \sfB_{2t}(V^*;G_0) $
            \begin{itemize}
                \item If $visited[w] = false$ then Mark $w$ as {\it isolated} and  $Explore(w)$
            \end{itemize}
        \end{itemize}
        
    \item Explore(v) 
        \begin{itemize}
       \item for each edge $(v,w) \in \calH$ do
        \begin{itemize}
            \item If $w \in V_0$,  mark $(v,w)$ as {\it stale edge} and set $visited[w] = true$. 
            \item If $w \notin V_0$ and $visited[w] = true$, mark the edge $(v,w)$ as {\it back edge} 
            
            \item If $w \notin V_0$ and $visited[w] = false$ set visited[w] = true and call Explore(w)
            
        \end{itemize}
        \end{itemize}

\end{itemize}

Execution of ExploreGraph will consist of a sequence of DFS traversals each producing a connected component of $\calH$. 
Each traversal starts at some node  $w \in \sfB_{2t}(V^*,G_0)$ that has not been visited yet.
The traversal goes through edges in $\calH$, visiting new nodes, marking some edges as back and stale.

Observe that every stale edge or a back-edge increases the cycle number of $\sfB_2t(V_0;G)$ by adding an edge, but no new vertex.
Therefore, the total number of stale/back edges is at most $\#_c(\sfB_{2t}(V_0;G))- \#_c(G_0)$. 
For brevity, let us denote $s \defeq \#_c(\sfB_{2t}(V_0;G))- \#_c(G_0)$.

A vertex $v \in V^*$ is typical if the following hold:

\begin{enumerate}
    \item Every vertex $w \in \sfB_{2t}(v;G_0)$ is marked {\it isolated} (never visited via a stale edge).

    \item For every vertex $w \in \sfB_{2t}(v;G_0)$, the corresponding call $Explore(w)$ did not produce a {\it stale} or {\it back} edge in one of its descendants.
    
    Since there
    
\end{enumerate}

    As there are at most $s$-stale edges, at most $s$ vertices $v \in V^*$ have some vertex $w \in \sfB_{2t}(v,G_0)$ visited by a stale edge.
    Furthermore, at most $s$ vertices $v \in V^*$ have a vertex $w \in \sfB_{2t}(v,G_0)$ that produced a {\it stale} or {\it back} edge.
    Hence at least $|V^*| - 2s$ vertices are typical.
\end{proof}

%Since every pair of vertices $v,v' \in V^*$ are distance $8t$-away from each other in $G_0$, there is $\sfB_{2t}(v;G_0) \cap \sfB_{2t}(v';G_0) = \emptyset$.  Hence, if {\it every vertex $v \in V^*$ were typical}, the number of connected components in $\calH$ would be,
%\[\text{ Number of connected components in } \calH = \sum_{v \in V^*} |\sfB_{2t}(v,G_0)| \]
%

Returning to the proof of \pref{lem:all-close-to-heavy}, let $\Vtyp = \{i_1,\ldots,i_R\}  \subseteq V^*$ denote the set of {\it typical} vertices in $V^*$.
Now we will describe how to sample a graph $G = G^c \cup G_0$ from the conditional distribution:  $\left( G \mid  \Vtyp \text{are typical}\right)$.
%We can equivalently sample the random graph $G^c$ in stages as follows. 

\begin{itemize}
    \item For $j = 1$ to $R$
        \begin{itemize}
            \item Sample $\sfB_{2t}(i_j; G)$ conditioned on $i_j$ being {\it typical} or equivalently,  $\sfB_{2t}(i_j; G) \cap G_{j-1}= \sfB_{2t}(i_j,G_0)$.  
            
            For sake of concreteness, we will outline how to sample $\sfB_{2t}(i_j;G)$.  For each vertex $w \in \sfB_{2t}(i_j,G_0)$, with distance $\dist(w,i_{j}) = D$, sample the local neighborhood tree $\sfT_w$ of depth $D$ in a breadth first manner, while avoiding vertices in $G_{j-1}$.
            %in breadth-first fashion, inductively sampling vertices at each distance $1$ to $2t$.  Formally, the procedure can be described as follows.
            %\begin{itemize}
            %    \item $\sfN_0 \leftarrow \{i_j\}$.
            %    \item For $d = 1$ to $2t$
            %        \begin{itemize}
            %            \item For each vertex $v \in \sfN_{d-1}(i;G)$
            %                \begin{itemize}
             %                       \item Include each incident edge at $(v,u)$ with $u \notin V[G_{j-1}]$ with probability $\dovern$.  Let $\sfN_v$ denote the set of neighbors of $v$ sampled as above.
               %                 
              %              \end{itemize}
                 %       
                %        \item $\sfN_d (i;G) \leftarrow \sfN_{d}(i;G_0) \cup \left( \bigcup_{v \in \sfN_{d-1}(i:G)} \sfN_v\right) $.

                    %\end{itemize}
                %\end{itemize}
            
            \item $G_{j} = G_{j-1} \bigcup  \sfB_{2t}(i_{j};G)$.
        \end{itemize}
        \item Sample all the remaining unrevealed edges by including them independently with probability $\dovern$, conditioned on $\Vtyp$ being typical.
\end{itemize}

By virtue of the above order of sampling, we can write
\begin{align}\label{eq:myeq17}
&\Pr\left[ \forall v \in \Vtyp, v \text{ is }(t,d')\text{-close to heavy in }  G | \Vtyp \text{ is typical}\right] \\
& = \prod_{j = 1}^R \Pr\left[ i_{j} \text{ is }(t,d')\text{-close to heavy in }  G_j \Big |   G_{j-1} \right]
\end{align}
%where the condition $\sfB_{2t}(i_j,G) \cap G_{j-1} = \sfB_{2t}(i_j,G_0)$ is equivalent to $i_j$ being typical.
%
where recall that $G_j$ is sampled conditioned on $i_j$ being typical.  

By virtue of being typical, the neighborhood of $i_j$, $\sfB_{2t}(i_j, G_{j-1})$ is same as its original neighborhood $\sfB_{2t}(i_j,G_0)$ in $G_0$.
Let $\calF_j$ be the remaining vertices in $G_{j-1}$ namely, 
\[\calF_j = G_{j-1} - \sfB_{2t}(i_j,G_0) \]
Since $i_j$ conditioned on being typical, vertices in $\calF_j$ are forbidden to be chosen in the neighborhood $\sfB_{2t}(i_j;G_j)$.
Therefore deleting all edges among $\calF_j$ has no effect on whether $i_j$ is $(t,d')$-close to heavy.  That is,
\[i_j \text{ is }(t,d')\text{-close to heavy in }G_j \iff
  i_j \text{ is }(t,d')\text{-close to heavy in }\sfB_{2t}(v,G_0) \cup G^c
\]
We will bound the probability of the latter event using \pref{lem:typical-nbd-sample}.  Specifically, apply \pref{lem:typical-nbd-sample} with $G' = \sfB_{t}(i_j,G_0)$ and $\calF_{j}$ we get that,

\[ \Pr\left[i_j \text{ is } (t,d')\text{-close to heavy in }  | \sfB_{2t}(v,G) \cap \calF = \emptyset \right] \leq C t^2 \cdot  d^t    e^{-c\left(\frac{d'}{d}\right)^t \cdot \frac{1}{1+t^2}} \defeq \Delta(d,d',t) \] 

Substituting back in \eqref{eq:myeq17}, 
\begin{align}\label{eq:myeq18}
 \Pr\left[ \forall v \in \Vtyp, v \text{ is }(t,d')\text{-close to heavy in }  G | \Vtyp \text{ is typical}\right] \leq \Delta(d,d',t)^{|\Vtyp|} \leq \Delta(d,d',t)^{|V^*| - 2s}
 \end{align}
where recall that $s \defeq \#_c(\sfB_{2t}(V_0,G)) - \#_c(G_0)$.
By Lemma A.3 in \cite{FM17}, we know that or all $k < \log^2 n$, 
\[\Pr[\#_c(\sfB_{2t}(V_0,G)) - \#_c(G_0) \geq k] \leq C(\log^2 n) n^{-0.7 k}\]
Along with \eqref{eq:myeq18}, this implies that
 \begin{align*}
 & \Pr\left[\forall v \in V^*, v \text{ is } (t,d')\text{-close to heavy in } G\right] \\
 & \leq \Pr\left[s > \frac{|V^*|}{4}\right] + \Delta(d,d',t)^{|V^*| - 2 (\frac{|V^*|}{4})} \leq 2 \Delta(d,d',t)^{|V^*|/2} 
 \end{align*}
for large enough $n$.  Finally, the lemma follows by observing that for all fixed $d,d',\delta$, we can make $\Delta(d,d',t) \leq \delta^{4t}$ for sufficiently large $t$.

\end{proof}

We now have all the pieces need to prove \pref{claim:prob1}.
\begin{proof}(Proof of \pref{claim:prob1})
The function $\beta_{A^c, \alpha_R}(z) = \beta_W(A^c, A_Q = z, A_R = \alpha_R, A_J = 1)$ is a anti-monotone \pnote{better terminology?} function of $z$.  

For every pair $ij \in Q^*$, since $\beta_{A^c, \alpha_R}$ depends on $z_{ij}$, there is some setting of $z_{Q \setminus \{ij\}}$ such that $\beta_{A^c,\alpha_R}(z_{ij}=0,z_{Q \setminus \{ij\}}) = 1$ but $\beta_{A^c,\alpha_R}(z_{ij}=1,z_{Q \setminus \{ij\}}) = 0$.
This implies that addition of edge $ij$ creates a vertex $v$ that is not $(t,d')$-bounded.  
The vertex $v$ is within distance $t$ of both endpoints $i$ and $j$, since otherwise the addition of the edge $ij$ has no effect on the $(t,d')$-boundedness of vertex $v$.

Therefore we can upper bound,
\begin{align*}
 & \Pr_{A^{c}}\left[    \beta_{A^c, \alpha_R} \text{ depends on all bits in } Q^*\right] \\
 & \leq \Pr_{A^{c}} \left[\forall ij \in Q^*, i \text{ is }(t,d')\text{-close to heavy in graph } A^c\cup J \cup Q \cup R \right]
\end{align*}
By construction of the set of edges $Q$, the set $V^* = \{i | ij \in Q^*\}$ satisfy the conditions in the hypothesis of \pref{lem:all-close-to-heavy} in the graph formed by $Q \cup R \cup J$.
Hence the claim follows by appealing to \pref{lem:all-close-to-heavy}.
\end{proof}

%% file: Content/robustness.tex
\section{Robustness in the Stochastic Block Model} \label{sec:sdprobustness}

\newcommand{\cc}{\ell}

In this section, we will show that the local statistic SDP relaxation yields a robust algorithm.
Throughout this section, let $\bG$ be drawn from either the \ER or Stochastic Block Model on $n$ vertices, with average degree $d$. We will prove that an adversarial modification of $\epsilon n$ edges, for sufficiently small $\epsilon$, cannot meaningfully later subgraph occurrences, except by creating vertices of high degree. Therefore, if we run the Local Statistics SDP after deletion of sufficiently high-degree vertices, the resulting algorithm is robust to adversarial edge meddling.

Let us make this intuition precise. In a similar vein to the partially labelled graph formalism from the main body of the paper, let us define now a \textit{pinned graph} to be a pair $(H,R)$ where $R \subset V(H)$ contains exactly one vertex from each connected component of $H$. Write $\ell(H) = |R|$ for the number of such components. Given a graph $G$ and a subset $T$ of $\ell(H)$ vertices, an \textit{occurrence} of $(H,R)$ in $(G,T)$ is an (injective?) homomorphism that maps the pinned vertices $R$ to the target set $T$. Let's write $\Gamma_{H,R}(G,T)$ for the set of such occurrences.

\begin{claim} \label{claim:expectations}
    For every pinned graph $(H,R)$ and any then for any $T \subset [n]$,
        \[ \lim_{n \to \infty} \E_{G} [|\Gamma_{H,W(H)}(G,T)|^2] = c_{H} \]
    for a constant $c_H$ dependent only on $H$.
\end{claim}
\begin{proof}
    First let us consider the following expectation:
    \[
        \E [|\Gamma_{H,R}(\bG,S)|] = \sum_{\phi : V(H) \to [n]} \bbP \left[\text{$\phi$ is an occurrence}\right]
    \] 
    The number of nonzero terms in the summation is $\binom{n}{|V(H)|-\cc(H)}(|V(H)|-\cc(H))!  = n^{|V(H)| - \cc(H)} + O(n^{|V(H)| - \cc(H)})$. For each term, the probability that $\phi$ is an occurrence is $O(n^{-|E(H)|})$ Since $|E(H)| \geq |V(H)| - \cc(H)$ in a graph with at most $\cc$ connected components, the above expectation is a constant depending on graph $H$.
    
    Now we turn our attention to $\E [|\Gamma_{H,R}(\bG,T)|^2]$, which we can expand as
    \[
        \E [|\Gamma_{H,R}(\bG, T)|^2] = \sum_{\phi,\psi : V(H) \to [n]} \bbP[\phi,\psi \text{ are occurrences}]
    \]
    Each nonzero term gives rise to a graph $H^\ast$ obtained by taking the union of the images of $\phi(H)$ and $\psi(H)$; this union is a graph with $\cc$ connected components, each of which contains one of the target vertices $T$. There are only finitely many graphs on at most $2|V(H)|$ vertices that have this form, so we can write the expectation of concern to us as a sum of expected occurrence counts of these types, and apply our initial observation.
\end{proof}

\begin{lemma}
Fix $d > 0, \epsilon \in (0,1)$, and a finite pinned graph $(H,R)$. There exists $\Delta(H, R, d,\epsilon) > 0$ such that the following holds: with probability $1- \epsilon$ for all $Q \subset [n]$ with $|Q| \leq \Delta(H,R,\epsilon) n$,
\[ 
    \left|\bigcup_{T \cap Q \neq \emptyset} \Gamma_{H,R}(\bG,T)\right| \leq \epsilon n^{\cc(H)} 
\]
\end{lemma}
\begin{proof}
    We can expand the size of this union as a sum over subsets $T \in {[n]\choose \cc(H)}$:
    \begin{align}
        \left|\bigcup_{T \cap Q \neq \emptyset} \Gamma_{H,R}(\bG,T)\right| 
        &= \sum_{T \in \binom{[n]}{\cc(H)}} \left|\Gamma_{H,R}(\bG,T)\right| \cdot \indicator{T \cap Q \neq \emptyset} \nonumber \\
        &\leq \left( \sum_{T \in \binom{[n]}{\cc}} |\Gamma_{H,R}(\bG,T)|^2  \right)^{1/2} \cdot \left( \sum_{T \in \binom{[n]}{\cc}} \indicator{T \cap Q \neq \emptyset}^2 \right)^{1/2} \label{eq:29}
    \end{align}
    With probability at least $1 - \epsilon$, we have 
    \[
        \sum_{T \in \binom{[]n]}{\cc(H)}} \left|\Gamma_{H,R}(\bG,T)\right|^2 \leq \frac{1}{\epsilon} \bbE\left[\sum_{T \in \binom{[n]}{\cc(H)}} \left|\Gamma_{H,R}(\bG,T)\right|^2 \right] \leq \frac{c_H}{\epsilon} n^{\cc(H)}
    \]
    where $c_H$ is the constant depending on $H$ from \pref{claim:expectations}. Set $\Delta(H,d,\epsilon) = \frac{\epsilon^3}{\cc c_T}$. Notice that for a set $Q$ smaller than $\Delta(H,d,\epsilon)$, the number of $T \subset \binom{[n]}{\cc(H)}$ is at most $\cc \cdot \frac{\epsilon^3}{\cc(H) c_H} n^{\cc(H)}  = \frac{\epsilon^3}{c_H} n^{\cc(H)}$. Conditioned on this event of probability $1- \epsilon$, we can use \pref{eq:29} to conclude that,
    $$
        \left|\bigcup_{T \cap Q \neq \emptyset} \Gamma_{H,R}(\bG,T)\right| \leq \epsilon n^{\cc(H)}
    $$ 
    whenever $|Q| \leq \Delta(H,d,\epsilon) n$.
\end{proof}

By taking a union bound over all trees of size $k$ and all choices of designated vertices, we have the following corollary.
\begin{corollary} \label{cor:boundedcounts}
    For every $d,k > 0$ and $\epsilon \in (0,1)$, there exists $\eta$ such that following holds. Denoting by $\calH$ the set of all graphs with at most $m$ edges, then with probability $1 - \epsilon$, for all $Q \subset [n], |Q| \leq \eta n$ and $H \in \calH$ we have
    \[
        \left|\Gamma_{H}(\bG,Q)\right| \leq \epsilon n^{\cc(H)}.
    \]
\end{corollary}

Now we are ready to prove the main theorem of this section, namely robustness of local statistics SDP relaxation.
\begin{theorem} (Robustness of Local Statistics SDP) \label{thm:sbmrobustness}
    For every $d, \epsilon, k$, there exist $B$ and $
    \gamma$ such that, with probability at least $1-\epsilon$ over $\bG = (\bG,\bE)$, the following holds:
    
    Let $\tilde{G} = ([n],\tilde{E})$ be an arbitrary graph such that $| \bE \triangle \tilde{E}| \leq \gamma n$; write $G^* = ([n],E^*)$ for the graph obtained by deleting edges incident to all vertices of degree $> B$ in $\tilde{G}$. Then for every graph $H$ with at most $m$ edges,
    \[ 
        |\Gamma_H(G) \triangle \Gamma_H(G^*)| \leq \epsilon n^{\ell(H)}
    \] 
    Consequently, if $\pE : \R[x]_{\leq 2} \to \R$ is a pseudoexpectation that is a feasible solution to the level $(2,m)$ local statistics SDP on $G^*$ (or $G$)  with tolerance $\delta$, then $\pE$ is a feasible solution on level $(2,m)$ local statistics SDP with tolerance $\delta + \epsilon$ on $G$ (or $G^*$).  Further, if $\pE$ is infeasible for the level $(2,m)$ local statistics SDP on $G^*$ (or $G)$ by a margin of $\delta$, then $\pE$ remains infeasible on the level $(2,m)$ local statistics SDP by margin of $\delta-\epsilon$ on $G$ (or $G^*$).
\end{theorem}

\begin{proof}
%By definition $\calT(G) \subseteq \calT(\tilde{G})$.
%
%Define $G^* = ([n],E^*)$ by deleting all vertices of degree $ > B$ in the graph $\tilde{G}$.
Let $\eta > 0$ be the choice for which \pref{cor:boundedcounts} holds given $d,k, \epsilon/4$.  Set $B \defeq \lceil \frac{2d}{\eta} \rceil$ and $\gamma = \frac{\epsilon}{4 m 2^m B^{m^3}}$. We will express $\Gamma_H(\bG) \triangle \Gamma_H(G^*) = \Gamma_{del} \cup \Gamma_{trunc} \cup \Gamma_{add}$ and bound the size of each of the three sets. 
\begin{itemize}
\item $\Gamma_{del} = \Gamma_H(\bG) - \Gamma_H(\tilde{G})$ are the occurrences of $H$ in $\bG$ that were deleted by the adversarial corruption of edges.  

Since the corruption deletes at most $\gamma n$ edges, which are incident on at most $2\gamma n < \eta n$ vertices, we can use \pref{cor:boundedcounts} to conclude that this set is at most $\epsilon n^{\ell(H)}/4$

\item $\Gamma_{trunc} = (\Gamma_H(\bG) \cap \Gamma_H(\tilde{G})) \setminus \Gamma_{H}(G^*)$ are the occurrences of $H$ that were deleted due to the removal of edges incident to high-degree vertices while constructing $G^*$.  

The average degree of the graph $\bG$ is $d + o(1)$ with $1-o_{n}(1)$.  Therefore, the average degree of $\tilde{G}$ is at most $d+ 2 \gamma < 2d$.  Hence, the number of vertices of degree $> B$ is at most $\left(2d/B\right) \cdot n < \eta n$.  Again by \pref{cor:boundedcounts}, $|\Gamma_{trunc}| \leq \epsilon n/4$. 

\item $\Gamma_{add} = \Gamma_H(G^*) \setminus \Gamma_H(\bG)$ are the occurrences of $H$ in $\bG$ that were added by the adversarial corruption, and survived the truncation of high-degree vertices.  

Every occurrence in $\Gamma_{add}$ includes one of the $\gamma n$ edges in $\tilde{E} - E$.  

Since the degree of each vertex of $G^*$
 is at most $B$, there are at most $B^m$ vertices in their neighborhood of radius $m$ around every vertex $v$.  Hence, for any given connected component $\calC \subseteq H$, the number of occurrences of $\calC$ that contain a vertex $i \in [n]$ is at most $|\calC| \cdot \left(B^m \right)^{|\calC|}$.  

For every edge $e = (u,v) \in \tilde{E} - E$, there are at most $2B^m$ vertices in their neighborhood of radius $m$.  
The number of occurrences of any connected component $\calC$ in this neighborhood is thus at most $(2B^m)^{|\cal C|}$.  

Hence the number of occurrences that use at least one edge in $|\tilde{E} - E|$ is at most
\begin{align*} \sum_{\calC \subset H} \left(n^{\cc(H) - 1} \cdot \left(B^m \right)^{|V(H)|-|\calC|}\right) \cdot (|(\tilde{E} - E) \cap E^*| \cdot (2B^m)^{|\calC|}  \leq \cc(H) \cdot 2^m B^{m^2  \cc(H)} \gamma n^{\cc(H)}
\end{align*}
By the choice of $\gamma$, the desired bound follows.

\end{itemize}
Conditioned on the event that assertion in \pref{cor:boundedcounts} holds, for every choice of corruptions, we have that
\[ \Gamma_H(\bG) \triangle \Gamma_H(G^*) \leq \epsilon n/4 + \epsilon n/4 + \epsilon n/4 <  \epsilon n \]

The claim about the solution to the level $(2,m)$-local statistics SDP is immediate by observing that for any partially labelled subgraph $(H,S, \tau)$,
\[ 
    \pE[|p_{H,S,\tau}(G^*,x) - p_{H,S,\tau}(G,x)]| \leq |\Gamma_H(G) \triangle \Gamma_H(G^*)|
\]
for any $\pE$ that satisfies $\calB_k$.
\end{proof}